\documentclass[A4,12pt]{article}

\usepackage{fullpage}
\usepackage{color}
\usepackage{url}
\usepackage{graphicx}
\usepackage{epstopdf}
\usepackage{xspace}
\usepackage{amsmath}
\usepackage{amssymb}
\usepackage{amsthm}
\usepackage{algorithmicx}
\usepackage[noend]{algpseudocode}
\usepackage{listings}

\usepackage{pgfplots}
\pgfplotsset{compat=newest}

\usepackage{fancyvrb}
\usepackage{ifthen}
\usepackage{tikz}
\usepackage[caption=false,font=footnotesize]{subfig}

\usepackage{titlesec}
\titleformat{\paragraph}
{\normalfont\normalsize\bfseries}{\theparagraph}{1em}{}
\titlespacing*{\paragraph}
{0pt}{3.25ex plus 1ex minus .2ex}{1.5ex plus .2ex}

\makeatletter
\newcounter{subsubparagraph}[subparagraph]
\renewcommand\thesubsubparagraph{\thesubparagraph.\@arabic\c@subsubparagraph}
\newcommand\subsubparagraph{\@startsection{subsubparagraph}{6}{\parindent}%
                                       {3.25ex \@plus1ex \@minus .2ex}%
                                       {-1em}%
                              {\normalfont\normalsize\bfseries}}
\newcommand*\l@subsubparagraph{\@dottedtocline{6}{10em}{5em}}
\makeatother

\newcommand{\FINDATE}{27.02.2015} 
\newcommand{\DNUM}{D2.2}
\newcommand{\DNAME}{White-box methodologies, programming abstractions and libraries}

\usepackage{siunitx}

\newcommand\UB{\mathit{UB}}

\algrenewcommand\alglinenumber[1]{\scriptsize #1:}

\algnewcommand{\LComment}[1]{\Statex  \(\triangleright\) #1 \hfill~}

\algdef{SE}[DOWHILE]{Do}{doWhile}{\algorithmicdo}[1]{\algorithmicwhile\ #1}%

\algnewcommand\EMPTY{\textbf{EMPTY}}

\algblockdefx[NAME]{Start}{End}%
	[1]{\textbf{Struct} #1:}%
	{EndStruct}
\algnotext[NAME]{End}

\algblockdefx[Name]{START}{END}%
	[1]{#1}%
	{Ending}
\algnotext[Name]{END}

\newtheorem{theorem}{Theorem}[section]
\newtheorem{corollary}{Corollary}[theorem]
\newtheorem{lemma}[theorem]{Lemma}

\usepackage{arrayjobx}
\usetikzlibrary{matrix}

\newcommand{\ie}{\textit{i.e.}\xspace}
\newcommand{\etal}{\textit{et al.}\xspace}
\newcommand{\eg}{\textit{e.g.}\xspace}

\newcommand{\aka}{\textit{a.k.a.}\xspace}

\newcommand{\ema}[1]{\ensuremath{#1}\xspace}

\newcommand{\powi}{\ema{P}}
\def\stat{\mathit{stat}}
\def\dyn{\mathit{dyn}}
\def\act{\mathit{active}}

\newcommand{\expo}[2]{\ema{#1^{\ifthenelse{\equal{#2}{}}{}{(#2)}}}}

\newcommand{\doexpo}[3]{\ema{#1^{(#2,#3)}}}

\newcommand{\pow}[1]{\expo{\powi}{#1}}
\newcommand{\pstat}[1]{\doexpo{\powi}{\stat}{#1}}
\newcommand{\pdyn}[1]{\doexpo{\powi}{\dyn}{#1}}
\newcommand{\pact}[1]{\doexpo{\powi}{\act}{#1}}

\newcommand{\powb}[2]{\powi_{#2}^{(#1)}}

\newcommand{\cas}{\textit{Compare-and-Swap}\xspace}

\newcommand{\ghz}[1]{\ema{#1\,\text{GHz}}}
\newcommand{\second}[1]{\ema{#1\,\text{sec}}}
\newcommand{\seconds}[1]{\ema{#1\,\text{secs}}}

\newcommand{\ps}{parallel section\xspace}
\newcommand{\pss}{parallel sections\xspace}
\newcommand{\rl}{retry loop\xspace}
\newcommand{\rls}{retry loops\xspace}
\newcommand{\ds}{data structure\xspace}

\newcommand{\itemx}[1]{}

\newcommand{\facf}{\ema{\lambda}}
\newcommand{\freq}{\ema{f}}

\graphicspath{{figs-chalmers/},{figs-chalmers/mandelbrot/},{figs-chalmers/results/},%
  {figs-chalmers/Mandelbrot_Prediction/},{figs-chalmers/legend/},{figs-chalmers/final/},%
{figs-chalmers/ani/}}

\newcommand{\leaveout}[1]{}
\newcommand{\remove}[1]{}
\newcommand{\comment}[2]{}{}

\newcommand{\op}[1]{\FuncSty{#1}}
\newcommand{\var}[1]{{\textsf{#1}}}
\newcommand{\code}[1]{{\textsf{\tt {#1}}}}

\newcommand{\cw}{\ema{\mathit{cw}}}
\newcommand{\pw}{\ema{\mathit{pw}}}
\newcommand{\thr}{\ema{\mathcal{T}}}
\newcommand{\nth}{\ema{n}}
\def\mrl{\mathit{RL}}
\def\msl{\mathit{SL}}
\def\mps{\mathit{PS}}

\newcommand{\psma}{\ema{p_\textrm{sma}}}
\newcommand{\pmid}{\ema{p_\textrm{mid}}}
\newcommand{\pbig}{\ema{p_\textrm{big}}}

\def\ub{\texttt{+}}
\def\lb{\texttt{-}}
\def\bb{\mathit{b}}
\def\be{\mathit{e}}
\def\bd{\mathit{d}}
\def\bx{\mathit{o}}

\newcommand{\gendu}[3]{\ema{#1_{#2}^{\ifthenelse{\equal{#3}{}}{}{(#3)}}}}
\newcommand{\thrxy}{\gendu{\thr}{\bx}{\bb}}
\newcommand{\thrx}{\gendu{\thr}{\bx}{}}
\newcommand{\pwx}{\gendu{\pw}{\bx}{}}
\newcommand{\cwxy}{\gendu{\cw}{\bx}{\bb}}
\newcommand{\et}[1]{\ema{\operatorname{t}\left(#1\right)}}

\newcommand{\slxy}{\ema{\gendu{\msl}{\bx}{\bb}}}
\newcommand{\rlx}{\ema{\gendu{\mrl}{\bx}{}}}
\newcommand{\slx}{\ema{\gendu{\msl}{\bx}{}}}

\newcommand{\psx}{\ema{\gendu{\mps}{\bx}{}}}

\newcommand{\wh}[1]{\ema{\widehat{#1}}}

\newcommand{\thre}{\gendu{\thr}{\be}{}}
\newcommand{\thred}{\gendu{\thr}{\be}{\lb}}
\newcommand{\thrend}{\gendu{\thr}{\be}{\ub}}
\newcommand{\pwe}{\gendu{\pw}{\be}{}}
\newcommand{\cwe}{\gendu{\cw}{\be}{}}

\newcommand{\cwed}{\gendu{\cw}{\be}{\lb}}
\newcommand{\cwend}{\gendu{\cw}{\be}{\ub}}
\newcommand{\thrd}{\gendu{\thr}{\bd}{}}
\newcommand{\thrde}{\gendu{\thr}{\bd}{\ub}}
\newcommand{\thrdne}{\gendu{\thr}{\bd}{\lb}}
\newcommand{\pwd}{\gendu{\pw}{\bd}{}}
\newcommand{\cwd}{\gendu{\cw}{\bd}{}}
\newcommand{\rld}{\ema{\gendu{\mrl}{\bd}{}}}

\newcommand{\psd}{\ema{\gendu{\mps}{\bd}{}}}
\newcommand{\cwde}{\gendu{\cw}{\bd}{\ub}}
\newcommand{\cwdne}{\gendu{\cw}{\bd}{\lb}}

\def\drl{\mathit{RL}}
\def\dps{\mathit{PS}}

\newcommand{\powe}[1]{\ema{\powb{#1}{\be}}}
\newcommand{\powerl}[1]{\ema{\powb{#1}{\be,\drl}}}
\newcommand{\poweps}[1]{\ema{\powb{#1}{\be,\dps}}}
\newcommand{\powd}[1]{\ema{\powb{#1}{\bd}}}
\newcommand{\powdrl}[1]{\ema{\powb{#1}{\bd,\drl}}}
\newcommand{\powdps}[1]{\ema{\powb{#1}{\bd,\dps}}}
\newcommand{\powx}[1]{\ema{\powb{#1}{\bx}}}
\newcommand{\powxrl}[1]{\ema{\powb{#1}{\bx,\drl}}}
\newcommand{\powxps}[1]{\ema{\powb{#1}{\bx,\dps}}}

\newcommand{\rat}{\ema{r}}
\newcommand{\ratx}{\ema{\rat_{\bx}}}
\newcommand{\rate}{\ema{\rat_{\be}}}
\newcommand{\ratd}{\ema{\rat_{\bd}}}

\newcommand{\codae}[1]{\gendu{\rho}{\be}{#1}}
\newcommand{\codad}[1]{\gendu{\rho}{\bd}{#1}}

\newcommand{\nulle}{\ema{\textsc{Null}}}

\newcommand{\incgrapr}[1]{
\includegraphics[width=\textwidth]{#1-r1}
}

\newarray\alg
\readarray{alg}{MS&Val&TZ&Gid&caca&Hof&Moi}
\newcounter{algi}
\newcommand{\walg}[1]{%
\setcounter{algi}{#1}\addtocounter{algi}{1}
\ema{\text{{\textbf{\alg(\thealgi)}}}}
}

\newcommand{\syscha}{{\textit System~A}\xspace}

\newcommand{\FuncSty}[1]{\texttt{#1}\xspace}

\graphicspath{{figs-chalmers/},{figs-chalmers/mandelbrot/},{figs-chalmers/results/},%
  {figs-chalmers/Mandelbrot_Prediction/},{figs-chalmers/legend/},{figs-chalmers/final/},%
{figs-chalmers/ani/},{figs-uit/}}

\begin{document}
\readarray{alg}{MS&Val&TZ&Gid&caca&Hof&Moi}

\thispagestyle{empty}

\vspace{-3cm}
\begin{center}
\textbf{SEVENTH FRAMEWORK PROGRAMME}\\
\textbf{THEME ICT-2013.3.4}\\
Advanced Computing, Embedded and Control Systems
\end{center}
\bigskip

\begin{center}
\includegraphics[width=\textwidth]{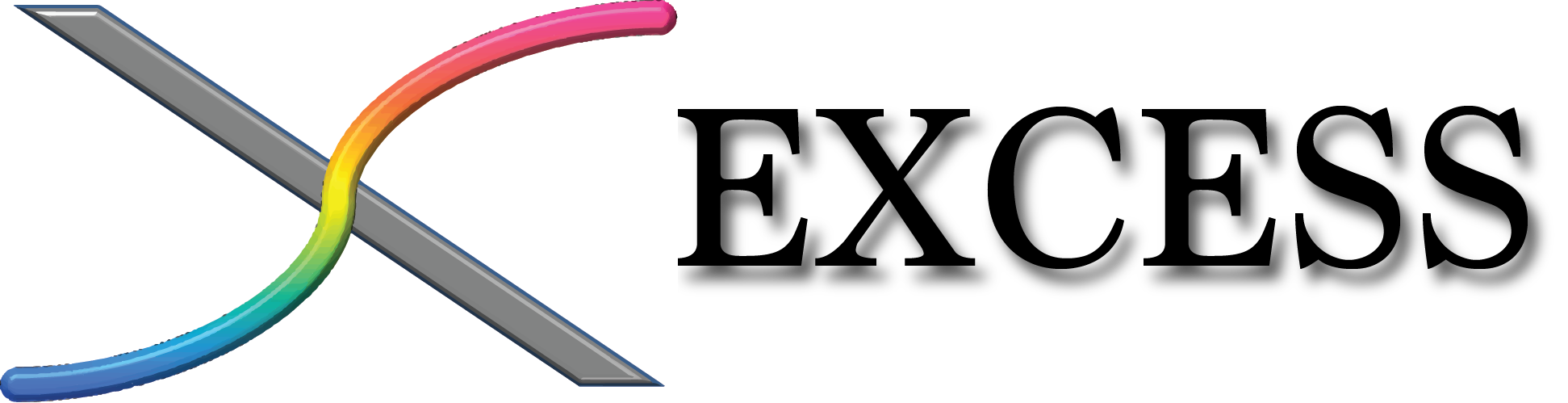}
\end{center}
\bigskip

\begin{center}
Execution Models for Energy-Efficient
Computing Systems\\
Project ID: 611183
\end{center}
\bigskip

\begin{center}
\Large
\textbf{\DNUM} \\
\textbf{\DNAME}
\end{center}
\bigskip

\begin{center}
\large
Phuong Ha, Vi Tran, Ibrahim Umar,
Aras Atalar, Anders Gidenstam, Paul Renaud-Goud, Philippas Tsigas
\end{center}

\vfill

\begin{center}
\includegraphics[width=3cm]{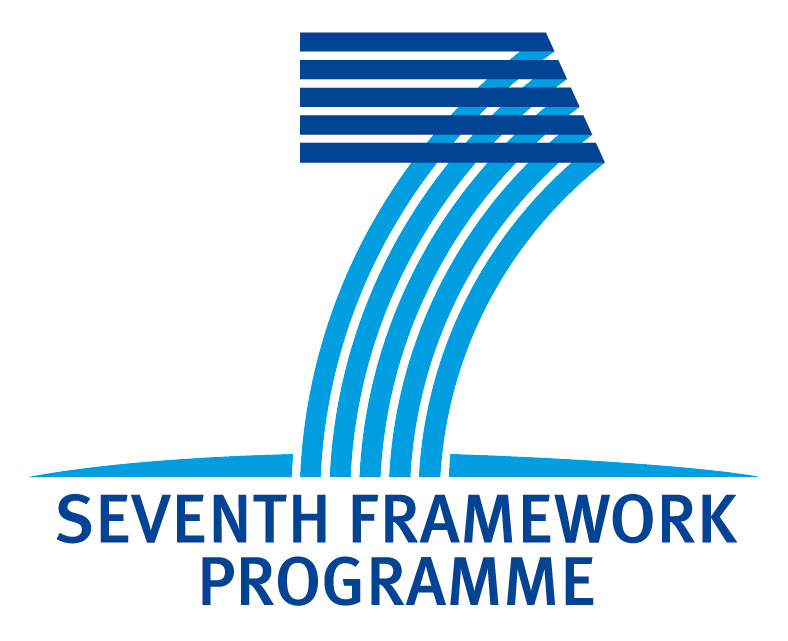}\\
Date of preparation (latest version): \FINDATE \\
Copyright\copyright\ 2013 -- 2016 The EXCESS Consortium \\
\hrulefill \\
The opinions of the authors expressed in this document do not
necessarily reflect the official opinion of EXCESS partners or of
the European Commission.
\end{center}

\newpage

\section*{DOCUMENT INFORMATION}

\vspace{1cm}

\begin{center}
\begin{tabular}{ll}
\textbf{Deliverable Number} & \DNUM \\
\textbf{Deliverable Name} & \DNAME \\
\textbf{Authors}

& Phuong Ha \\
& Vi Tran \\
& Ibrahim Umar \\
& Aras Atalar\\
& Anders Gidenstam \\
& Paul Renaud-Goud\\
& Philippas Tsigas \\

\textbf{Responsible Author} & Phuong Ha\\
& e-mail: \url{phuong.hoai.ha@uit.no} \\
& Phone: +47 776 44032 \\
\textbf{Keywords} & High Performance Computing; \\
& Energy Efficiency \\
\textbf{WP/Task} & WP2/Task 2.2 \\
\textbf{Nature} & R \\
\textbf{Dissemination Level} & PU \\
\textbf{Planned Date} &  28.02.2015\\
\textbf{Final Version Date} & 27.02.2015 \\
\textbf{Reviewed by} & Christoph Kessler (LIU), Michael Gienger (HLRS) \\
\textbf{MGT Board Approval} & YES\\
\end{tabular}
\end{center}

\newpage

\section*{DOCUMENT HISTORY}

\vspace{1cm}

\begin{center}
\begin{tabular}{llll}
\textbf{Partner} &
\textbf{Date} &
\textbf{Comment} &
\textbf{Version} \\
UiT (P.\ Ha, V.\ Tran) & 02.12.2014 & Deliverable skeleton & 0.1 \\
UiT (P.\ Ha, V.\ Tran, I.\ Umar) & 23.01.2015 & First version sent to WP2 partners & 0.2 \\
CTH (P. Renaud-Goud) & 28.01.2015 & Chalmers' part included & 0.3 \\
UiT (P.\ Ha, V.\ Tran, I.\ Umar) & 06.02.2015 & Consolidated version sent to internal  & 0.4 \\
														& 					& reviewers 													&			\\
UiT (P.\ Ha, V.\ Tran) & 27.02.2015 & Final version & 1.0 \\
\end{tabular}
\end{center}

\newpage

\begin{abstract}
This deliverable reports the results of white-box methodologies and early results of the first prototype of libraries and programming abstractions as available by project month 18 by Work Package 2 (WP2). It reports i) the latest results of Task 2.2 on white-box methodologies, programming abstractions and libraries for developing energy-efficient data structures and algorithms and ii) the improved results of Task 2.1 on investigating and modeling the trade-off between energy and performance of concurrent data structures and algorithms. The work has been conducted on two main EXCESS platforms: Intel platforms with recent Intel multicore CPUs and Movidius Myriad1 platform.
\begin{itemize} 
\item Regarding white-box methodologies, we have devised new relaxed cache-oblivious models and proposed a new power model for Myriad1 platform and an energy model for lock-free queues on CPU platforms. 
For Myriad1 platform, the improved model now considers both computation and data movement cost as well as architecture and application properties. The model has been evaluated with a set of micro-benchmarks and application benchmarks.
For Intel platforms, we have generalized the model for concurrent queues on CPU platforms to offer more flexibility according to the workers calling the \ds (parallel section sizes of enqueuers and dequeuers are decoupled). 

\item Regarding programming abstractions and libraries, we have continued investigating the trade-offs between energy consumption and performance of data structures such as concurrent queues and concurrent search trees based on the early results of Task 2.1. Based on the investigation, we have implemented a set of libraries and programming abstractions including concurrent queues and concurrent search trees that are energy-aware. The preliminary results show that our concurrent trees are faster and more energy efficient than the state-of-the-art on commodity HPC and embedded platforms. 
\end{itemize}

\end{abstract}

\newpage

\section*{Executive Summary}
Computing technology is currently at the beginning of the disruptive transition from petascale to exascale computing (2010 –- 2020), posing a great challenge on energy efficiency. High performance computing (HPC) in 2020 will be characterized by data-centric workloads that, unlike those in traditional sequential/parallel computing, are comprised of big, divergent, fast and complex data. In order to address energy challenges in HPC, the new data must be organized and accessed in an energy-efficient manner through novel fundamental data structures and algorithms that strive for the energy limit. Moreover, the general application- and technology-trend indicates finer-grained execution (i.e. smaller chunks of work per compute core) and more frequent communication and synchronization between cores and uncore components (e.g. memory) in HPC applications. Therefore, not only concurrent data structures and memory access algorithms but also synchronization is essential to optimize the energy consumption of HPC applications. However, previous concurrent data structures, memory access algorithms and synchronization algorithms were designed without energy consumption in mind. The design of energy-efficient fundamental concurrent data structures and algorithms for inter-process communication in HPC  remains a largely unexplored area and requires significant efforts to be successful.

Work package 2 (WP2) aims to develop interfaces and libraries for energy-efficient inter-process communication and data sharing on the new EXCESS platforms integrating Movidius embedded processors. In order to set the stage for these tasks, WP2 needs to investigate and model the trade-offs between energy consumption and performance of data structures and algorithms for inter-process communication. 
WP2 also concerns supporting energy-efficient massive parallelism through scalable concurrent data structures and algorithms that strive for the energy limit, and minimizing inter-component communication through locality- and heterogeneity-aware data structures and algorithms.

The latest results of Task 2.2 (PM7 - PM36) on white-box methodologies, programming abstractions and libraries available by project month 18 as well as the improved results of Task 2.1 on investigating and modeling the trade-off between energy and performance \cite{HaTUTGRWA14}, are summarized in this report. 

\subsubsection*{White-box methodologies} 
The white-box methodologies presented in this report include a new cache-oblivious methodology and new energy models that help develop energy-efficient data structures and algorithms:

\begin{itemize}
\item We have devised a new {\em relaxed} cache oblivious methodology that is appropriate for developing energy-efficient concurrent data structures and algorithms. 
\item We have proposed a new power model for Movidius embedded platform (Myriad1) which is able to predict the power consumed by a program running on a specific number of cores. We have validated the model with a set of micro-benchmarks and real applications such as  sparse/dense linear algebra kernels and the graph application kernels (e.g., the Graph 500 kernels \footnote{http://www.graph500.org/}). 

\item We have proposed a way to model the energy behavior of
lock-free queue implementations and parallel applications that use
them on CPU platforms. Focusing on steady state behavior, we have decomposed
energy behavior into throughput and power dissipation which can be
modeled separately and later recombined into several useful metrics,
such as energy per operation.
\end{itemize}

\subsubsection*{Programming abstractions and libraries}
We describe a set of implemented concurrent search trees and queues as well as their energy and performance analyses:
\begin{itemize}
\item We have developed a concurrent search tree library that contains several state-of-the-art concurrent search trees such as the non-blocking binary search tree, the Software Transactional Memory (STM) based red-black tree, AVL tree, and speculation-friendly tree, the fast concurrent B-tree, and the static cache-oblivious binary search tree. A family of novel locality-aware and energy efficient concurrent search trees, namely the DeltaTree, the Balanced DeltaTree, and Heterogeneous DeltaTree, are also enclosed in the concurrent search tree library. The DeltaTrees are platform-independent and up to 140\% faster and 220\% more energy efficient than the state-of-the-art on commodity HPC and embedded platforms. 

\item On lock-free queues on CPU platforms, we have automatized the process of estimating the performance and the power dissipation of any queue implementation, and integrated it in the EXCESS software.
\end{itemize} 

This report is organized as follows. Section ~\ref{sec:introduction} provides the background and motivations of the work presented in this deliverable. Section ~\ref{sec:descr} describes two EXCESS platforms and their setting to measure power consumption. The white-box methodologies such as relaxed cache-oblivious methodology and power models are presented in Section ~\ref{sec:white-box}.
Section ~\ref{sec:libraries} describes the first prototype of EXCESS libraries and programming abstractions including numerous concurrent search tree and queue implementations as well as their performance and energy analyses. Section ~\ref{sec:Conclusion} concludes the report with future works. 





\newpage

\tableofcontents

\newpage


\section{Introduction} \label{sec:introduction}
\subsection{Purpose}
In order to address energy challenges in HPC and embedded computing, data must be organized and accessed in an energy-efficient manner through novel fundamental data structures and algorithms that strive for the energy limit. Due to more frequent communication and synchronization between cores and memory components in HPC and embedded computing, not only efficient design of concurrent data structures and memory access algorithms but also synchronization is essential to optimize the energy consumption. However, previous concurrent data structures, memory access algorithms and synchronization algorithms were designed without considering energy consumption. Although there are existing studies on the energy utilization of concurrent data structures demonstrating non-optimal results on energy consumption, the design of energy-efficient fundamental concurrent data structures and algorithms for inter-process communication in HPC and embedded computing is not yet widely explored and becomes a challenging and interesting research direction.

EXCESS aims to investigate the trade-offs between energy consumption and performance of concurrent data structures and algorithms as well as inter-process communication in HPC  and embedded computing. By analyzing the non-intuitive results, EXCESS will devise a comprehensive model for energy consumption of concurrent data structures and algorithms for inter-process communication, especially in the presence of component composition. The new energy-efficient technology will be delivered through novel execution models for the energy-efficient computing paradigm, which consist of complete energy-aware software stacks (including energy-aware component models, programming models, libraries/algorithms and runtimes) and configurable energy-aware simulation systems for future energy-efficient architectures.

The goal of Work package 2 (WP2) is to develop interfaces and libraries for inter-process communication and data sharing on EXCESS platforms integrating Movidius embedded processors, along with investigating and modeling the trade-offs between energy consumption and performance of data structures and algorithms for inter-process communication. WP2 also concerns supporting energy-efficient massive parallelism through scalable concurrent data structures and algorithms that strive for the energy limit, and minimizing inter-component communication through locality- and heterogeneity-aware data structures and algorithms.  

In addition to investigating the trade-offs and devising comprehensive models for energy consumption of concurrent data structures and algorithms (Task 2.1), WP2 also identifies essential concurrent data structures and algorithms for inter-process communication in HPC with the focus on how to customize them (Task 2.2). We exploit common data-flow patterns to create generalized communication abstractions with which application designers can easily create
and exploit the customization for the data-flow patterns. This task constitutes the interfaces and libraries
for inter-process communication and data sharing on EXCESS platforms. The results also constitute
a white-box methodology for tuning energy efficiency and performance of concurrent data structures and
algorithms. 

This report summarizes i) the early results of Task 2.2 on white-box methodologies,  programming abstractions and libraries and ii) the improved results of Task 2.1 on investigating and modeling the trade-off between energy and performance of concurrent data structures and algorithms. The improved results of Task 2.1 constitute the theoretical basis for the whole work package.





\subsection{White-box Methodologies}
White-box methodology is a general study or a theoretical analysis of the principles and methods applied into a field of study or research to outline how the study or research should be taken. The term "white-box" means that the principles and methods in the methodology have prior knowledge of the inner workings and structures of the objects involved in the study. In the scope of EXCESS project, the energy efficiency, architecture and inner workings of a system must be well understood, and likewise for algorithms and data-sets. The "white-box" methodologies studied in this report include cache-oblivious methodology, a power model for Myriad1 platform and an energy model for lock-free concurrent queues.

\subsubsection{Cache-oblivious Methodology}  \label{sec:wb-method}
Energy efficiency is one of the most important factors in 
designing high performance systems.
As a result, data must be organized and accessed  in an energy-efficient manner through novel fundamental data structures and algorithms that strive for the energy limit.
Unlike conventional locality-aware algorithms that only concern
about whether the data is on-chip (e.g., cache) or not (e.g., DRAM),
new energy-efficient data structures and algorithms must consider data locality
in finer-granularity: \textit{where on chip the data is}. Dally \cite{Dally11} 
predicted that for chips using the 10nm technology, the energy required between accessing data in nearby
on-chip memory and accessing data across the chip will differ as much as 75x 
(2pJ versus 150pJ), whereas the energy required 
between accessing the on-chip data
and accessing the off-chip data will only differ by 2x (150pJ versus 300pJ). Therefore, 
in order to construct energy efficient software systems, data
structures and algorithms must support not only high parallelism but also fine-grained data locality~\cite{Dally11}.


In order to devise locality-aware algorithms, 
we need theoretical execution models that promote data locality. 
One example of such models is the the cache-oblivious (CO) models 
\cite{Frigo:1999:CA:795665.796479}, 
which enable the analysis of data transfer between two levels of the memory hierarchy.
CO models are using the same analysis as 
the widely known I/O models \cite{AggarwalV88} 
except in CO models an optimal replacement is assumed.
Lower data transfer complexity implies better data locality 
and higher energy efficiency as energy consumption caused by data 
transfer dominates the total energy consumption \cite{Dally11}.
These models require the knowledge of the algorithm and some parameters
of the architecture to be known beforehand, hence they are
white-box methods.

The cache-oblivious (CO) models (cf. Section \ref{cache-oblivious}) support not only fine-grained data locality but also portability. A CO algorithm that is optimized for 2-level memory, is asymptotically optimized for unknown multilevel memory (e.g., register, L1C, L2C, ..., LLC, memory), enabling fine-grained data locality (e.g., minimizing data movement between L1C and L2C). As cache sizes and block sizes in the CO models are unknown, CO algorithms are expected to be portable across different systems. For example, the memory transfer cost of an algorithm (e.g., how many data blocks need to be transferred between two level of memory), which is analyzed using the CO model, will be applicable on both HPC machines and embedded platforms (e.g., Myriad1/2 platforms), irrespective of the variations in the hardware parameters such as memory hierarchy, specifications and sizes. The performance portability is useful for analyzing the data movement and energy consumption of an algorithm in a platform-independent manner.

The memory transfer cost of an algorithm obtained using the CO model can be regarded as a first piece of information that can enable software designers to rapidly analyze the performance and energy consumption of their algorithms. After all, memory transfer is one of the parameters that dominate the total energy consumption. As for the next step, the transfer cost can be fed directly into the energy model of a specific platform to get a good approximation on the energy consumption of the algorithm on the platform. 

Algorithms and data structures analyzed using the 
cache-oblivious models \cite{Frigo:1999:CA:795665.796479} are 
found to be cache-efficient and disk-efficient \cite{Brodal:2004aa, Demaine:2002aa}, making them suitable
for improving energy efficiency in modern high performance systems. Nowadays, multilevel memory
hierarchies in commodity systems are becoming more prominent as modern CPUs tend to have at least 3 level of caches and disks start to incorporate hybrid-SSD cache memories. With minimal effort, cache-oblivious algorithms
are expected to be always locality-optimized irrespective of 
variations in memory hierarchies, enabling less data transfers
between memory levels that directly translate into runtime energy savings.

Since their inception, cache-oblivious models have been extensively used 
for designing locality-aware fundamental algorithms and data 
structures \cite{Brodal:2004aa, Demaine:2002aa, Fagerberg:2008aa}. Among those algorithms
are scanning algorithms (e.g., traversals, aggregates, and array reversals), divide 
and conquer algorithms (e.g., median and selection, and matrix multiplication), and sorting
algorithms (e.g., mergesort and funnel-sort \cite{Frigo:1999:CA:795665.796479}). Several static data structures (e.g., static search trees, and funnels) 
and dynamic data structures (e.g., ordered files, b-trees, priority queues, and linked-list) have
been also analyzed using the cache-oblivious models. Performance of the said cache-oblivious 
algorithms and data structures have been reported similar to or sometimes better than the performance of their traditional  cache-aware counterparts.

\subsubsection{Power and Energy Models}

The energy consumed by worldwide computing systems increases 7\% annually and becomes a major concern in information technology society. In order to tackle this issue, the research community and industry have proposed several research approaches  to reduce the energy consumption of IT systems \cite{Orgerie2014}.

One of the key research directions to improve energy-efficiency is to understand how much energy a computing system consumes and characterize the energy consumed by an individual component. By knowing the energy consumption of an algorithm on a specific computing architecture, researchers and practitioners can design and implement new approaches to reduce the energy consumed by a certain algorithm on a specific platform.

The energy and power consumption of computing systems can be either measured by integrated sensors and external multimeters or estimated by  models. Energy and power measurement equipment and sensors are not always available and can be costly to deploy and set up. Therefore, energy and power models are an alternative and convenient method to estimate the energy consumption of a computing component or a whole computing system. The models, however, need to be simple to use and should not interfere with the energy estimated results \cite{Orgerie2014}.

We have conducted a study on a power model of Myriad1 platform (cf. Section~\ref{sec:myriad1-model})and an energy model for lock-free queues on CPU platform (cf. Section~\ref{sec:cpu-model-inst}).



\subsection{Libraries and Programming Abstractions}

Concurrent data structures (e.g., search trees and queues) and algorithms used as the building blocks of energy-efficient software systems must support high parallelism and fine-grained data locality. In this work we focus on two libraries of the most widely used concurrent data structures, namely the search trees and queues. 

Concurrent search trees are fundamental data structures used widely in high performance file systems (e.g., TokuFS and XFS) and database systems (e.g., InnoDB, MyISAM, PostgreSQL and TokuDB). Concurrent search trees are usually used as the back end of dictionaries supporting search, insertion and deletion of records.


Concurrent FIFO queues are fundamental data structures
that are key components in applications, algorithms, run-time and
operating systems. The producer/consumer pattern, \eg, is a common
approach to parallelizing applications where threads act as either
producers or consumers and synchronize and stream data items between
them using a shared collection.
A concurrent queue, \aka shared ``first-in, first-out'' or FIFO
buffer, is a shared collection of elements which supports at least the
basic operations \op{Enqueue} (adds an element) and \op{Dequeue} (removes
the oldest element). \op{Dequeue} returns the element removed or, if the
queue is empty, \nulle.
A large number of lock-free (and wait-free) queue implementations have
appeared in the literature,
\eg,~
\cite{Val94,lf-queue-michael,TsiZ01b,MoirNSS:2005:elim-queue,%
DBLP:conf/opodis/HoffmanSS07,Gidenstam10:OPODIS} being some of the
most influential or most efficient results.
Each implementation of a lock-free queue has obviously its strong and weak
points so the impact on performance and energy when choosing one
particular implementation for any given situation may not be obvious.
As the number of known implementations of lock-free concurrent
queues is growing, it is of great interest to describe
a framework within which the different implementations can be ranked,
according to the parameters that characterize the situation.

In this deliverable, we report our results on concurrent data structures such as concurrent search trees and lock-free queues (cf. Section~\ref{sec:libraries}).

\subsection{Contributions}
The main achievements in this report are summarized in two main parts as below.

\begin{description}
\item [White-box methodologies] include new cache-oblivious methodology and energy models.

\begin{itemize}
\item We have devised a new {\em relaxed} cache oblivious model that are appropriate for developing energy-efficient concurrent data structures and algorithms. 
\item We have proposed an power model which is able to predict the power consumed by a program running on a specific number of cores. Given a certain platform and the computation intensity, the model can predict the power consumed by an algorithm, answering the question how many cores are required to run a program to achieve the optimized energy consumption. The model considers both platform and algorithm properties, giving more insights into how to design the algorithm to achieve better energy efficiency. The model has been validated by a set of micro-benchmarks and application kernels such as sparse/dense linear algebra kernels and graph kernels on Movidius embedded platform (Myriad1). 
\item We have continued the work done in D2.1~\cite{EXCESS:D2.1} on the modeling of
queue implementations. We have moved from a grey-box
model, where the performance and the power consumption of enqueue and
dequeue operations were hidden, to a white-box model, where the impact
of those operations are studied separately, and combined at the
end. Additionally, we have generalized the model to offer more
flexibility according to the workers calling the \ds (parallel section
sizes of enqueuers and dequeuers are decoupled).
\end{itemize}

\item [Programming abstractions and libraries] provide a set of implemented concurrent search trees and their energy and performance analyses,
as well as a set of implemented lock-free queues, together with a comparison between the predicted and measured throughput and power. 
\begin{itemize}
\item We have developed a concurrent search tree library that contains several state-of-the-art concurrent search trees such as the non-blocking binary search tree, Software Transactional Memory (STM) based red-black tree, AVL tree, and speculation-friendly tree, fast concurrent B-tree, and static cache-oblivious binary search tree. A family of novel locality-aware and energy efficient concurrent search trees, namely DeltaTree, Balanced DeltaTree, and Heterogeneous DeltaTree are also enclosed in the concurrent search tree library. All the components in this library support the stand-alone benchmark program mode for the purpose of micro-benchmarking and library mode that is pluggable into any C/C++ based programs.
\item The DeltaTrees are platform-independent and up to 140\% faster and 220\% more energy efficient than the state-of-the-art on commodity HPC and embedded platforms. In single thread evaluation, Heterogeneous DeltaTree is 38\% faster than std::Set of GCC standard library in the theoretical worst-case scenario of inserting a sorted sequence of keys; and is 50\% faster than std::Set for inserting a random sequence of keys — the theoretical average-case scenario.
\item On lock-free queues on CPU platforms, we have automatized the process of estimating the performance and the power dissipation of any queue implementation, and integrated it in the EXCESS software.
\end{itemize}
\end{description}

\newpage
\section{EXCESS Platforms and Energy Measurement Settings}
\label{sec:descr}
In this section, we introduce briefly two EXCESS platforms that we work with.  The system descriptions, which have been mentioned in Deliverable D2.1, are presented here to make this deliverable self-contained. As compared to Deliverable D2.1, this section is added with more updated and detailed information on measurement set-up for Myriad1 platform.

\subsection {System~A: CPU-based Platform}
\label{sec:chalmers-system}

\subsubsection{System Description}
\label{sec:sysDscrA}

\begin{itemize}
  \item CPU: Intel(R) Xeon(R) CPU E5-2687W v2
    \begin{itemize}
      \item 2 sockets, 8 cores each
      \item Max frequency: 3.4GHz, Min frequency: 1.2GHz, {frequency speedstep by DVFS: 0.1-0.2GHz.} Turbo mode: 4.0GHz.
      \item Hyperthreading (disabled)
      \item L3 cache: 25M, internal write-back unified, L2 cache: 256K, internal write-back unified. L1 cache (data): 32K internal write-back
    \end{itemize}
  \item DRAM: 16GB in 4 4GB DDR3 REG ECC PC3-12800 modules run at
    1600MTransfers/s. Each socket has 4 DDR3 channels, each supporting 2
    modules. In this case 1 channel per socket is used.
  \item Motherboard: {Intel} Workstation W2600CR, BIOS version: 2.000.1201 08/22/2013
  \item Hard drive: Seagate ST10000DM003-9YN162 1TB SATA
\end{itemize}

\subsubsection{Measurement Methodology for Energy Consumption}
\label{sec:SystemA:overview}
\begin{figure}
\begin{center}
\includegraphics[width=0.5\textwidth]{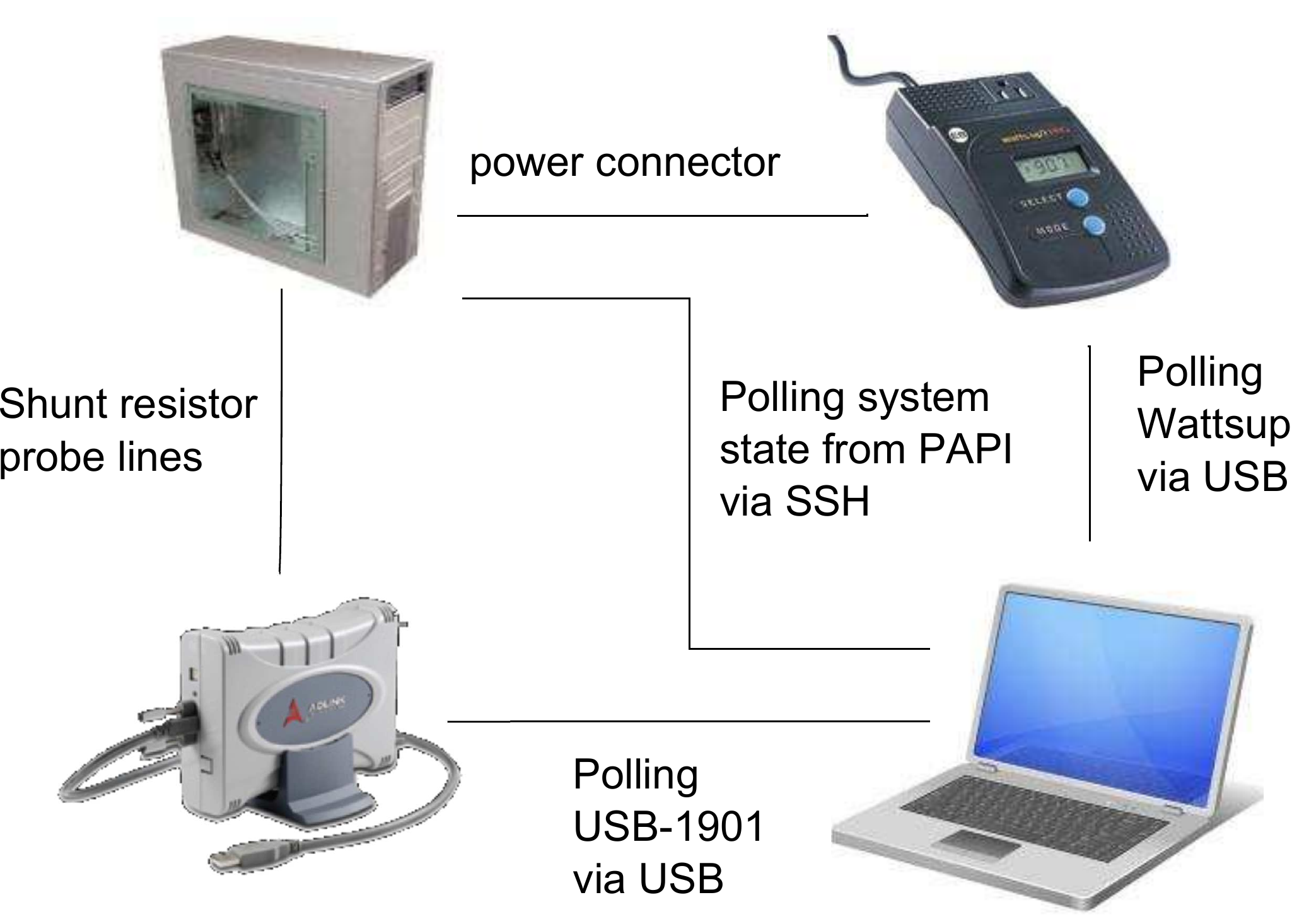}
\end{center}
\caption{Deployment of energy measurement devices for System~A.}
\label{fig:devicesSystemA}
\end{figure}

The energy measurement equipment for System~A at CTH, described in
Section~\ref{sec:sysDscrA}, is shown in Figure~\ref{fig:devicesSystemA}
and outlined below. It has previously been
described in detail in EXCESS D1.1~\cite{D1.1} and
D5.1~\cite{EXCESS:D5.1}.

The system is equipped with external hardware sensors for two levels
of energy monitoring as well as built in energy sensors:
\begin{itemize}
  \item At the system level using an external Watts Up
    .Net~\cite{EED:2003:WattsUp} power meter, which is connected
    between the wall socket and the system.
  \item At the component level using shunt resistors inserted between the
    power supply unit and the various components, such as CPU, DRAM and
    motherboard. The signals from the shunt resistors are captured with an
    Adlink USB-1901~\cite{Adlink:2011:USB1900} data acquisition unit (DAQ)
    using a custom utility.
  \item Intel's RAPL energy counters are also available for the CPU
    and DRAM components. A custom utility based on the PAPI
    library~\cite{BrDoGaHoMu:2000:PAPI,Weaver:2012:MEP:2410139.2410475}
    is used to record these counters and other system state parameters
    of interest.
\end{itemize}

For the work presented in this report the component level hardware
sensors and the RAPL energy counters have mainly been used.

\subsection{System~B: Movidius Myriad1 Embedded Platform}
\subsubsection{Movidius Myriad1 Architecture}
The Myriad1 platform developed by Movidius contains eight separate SHAVE (Streaming Hybrid Architecture Vector Engine) processors and one RISC core namely LEON. Each SHAVE one resides on one solitary power island.

The SHAVE processor contains a set of register files and a set of arithmetic units as Figure ~\ref{fig:ShaveInstructionUnits}. In this work, we consider the following registers and functional units as described below.
\begin{itemize}
\item Integer Register File (IRF) – Register file for storing integers from either the IAU or the SAU. 
\item Scalar Register File (SRF) – Register file for storing integers from either the IAU or the SAU. 
\item Vector Register File (VRF) – Register file for storing integers from either the VAU.  
\item Integer Arithmetic Unit (IAU) – Performs all arithmetic instructions that operate on integer numbers, accesses the IRF.
\item Scalar Arithmetic Unit (SRF) – Performs all Scalar integer/floating point arithmetic.
\item Vector Arithmetic Unit (VAU) – Performs all Vector integer/floating point arithmetic.
\item Load Store Unit (LSU) – There are two LSUs (LSU0 \& LSU1) and they perform any memory access and IO instructions.
\item Control Move Unit (CMU) – This unit interacts with all register files, and allows for comparing and moving between the register files.
\end{itemize}

\begin{figure}[!t] \centering
\resizebox{0.9\columnwidth}{!}{ \includegraphics{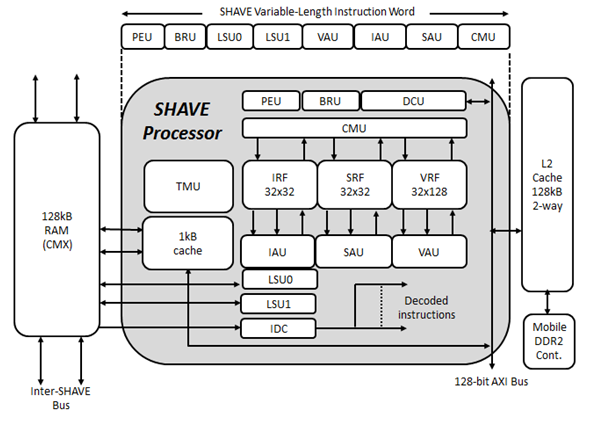}}
\caption{SHAVE Instruction Units}
\label{fig:ShaveInstructionUnits}
\end{figure}

The memory architecture of Myriad1 obtained from deliverable D4.1 is shown in Figure ~\ref{fig:MemoryArchitecture}. Eight SHAVE cores can access Double Data Rate Random Access Memory (DDR RAM) via L2 cache or bypass L2 cache. 


Except from DDR RAM, Movidius introduces a new memory component on-chip RAM containing one megabyte of internal memory with high bandwidth local storage of data and instruction code for the SHAVE processors. This memory component is named CMX. The CMX is constructed to allow eight SHAVEs parallel access to data and program code memory without stalling. Each SHAVE can access data on its own slice of CMX. It can also access other CMX slices of other SHAVE cores with slower time as the trade-off. 
In the later sections, the energy model is validated with both memory components CMX and DDR.

\begin{figure}[!t] \centering
\resizebox{0.9\columnwidth}{!}{ \includegraphics{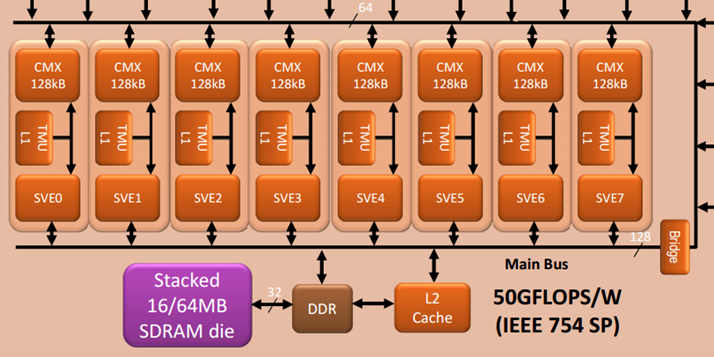}}
\caption{Myriad1 Memory Hierarchy}
\label{fig:MemoryArchitecture}
\end{figure}

\subsubsection{Myriad1 Measurement Set-up}
The platform supports the measurement of the power consumed only by the Myriad1 chip. We use a bench setup consisting of Myriad1 MV153 board, a DC step down converter down-regulating the 5V wall PSU to the 1.2V core voltage  and one HAMEG multimeter measuring all the voltage, current and consumed power values. 

\begin{figure}[!t] \centering
\resizebox{0.9\columnwidth}{!}{ \includegraphics{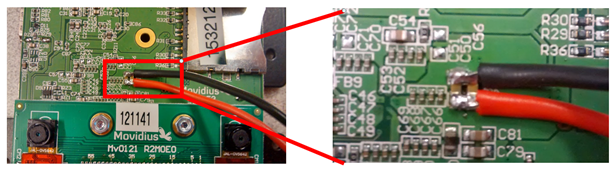}}
\caption{Myriad1 Power Supply Modification}
\label{fig:PowerSupplyModification}
\end{figure}

The modifications were made to the MV153 board to bypass the on-board voltage regulator which down-regulates the 5V wall PSU to the 1.2V core voltage required by Myriad1. That allows an external bench power-supply to be used in its place as shown in Figure \ref{fig:PowerSupplyModification}.  The MV153 board is modified in order to measure the voltage, current and the consumed power of only Myriad1 chip instead of the whole board. The HAMEG multimeter provides the measured data at a rate of 50 times per second which is able to capture the measurements for benchmarks with execution time longer than 20 milliseconds.

\newpage
\section{White-box Methodologies} \label{sec:white-box}
In this section, we present our studies on white-box methodologies including cache-oblivious methodology, a power model for Myriad1 platform and an energy model for lock-free concurrent queues.
\subsection{Cache-oblivious Methodology} \label{sec:cache-oblivious} \label{sec:COM-pre}
This section presents the cache-oblivious (CO) methodology. 
As pointed in Section \ref{sec:wb-method}, theoretical execution models that promote data locality are needed
in order to devise energy-efficient concurrent algorithms and data structures. Better data locality will result in higher energy efficiency since energy consumption caused by data transfer is predicted to dominate the total energy consumption \cite{Dally11}. 

We first present two memory models: 1) the I/O model \cite{AggarwalV88} and 2) cache-oblivious model \cite{Frigo:1999:CA:795665.796479}, which enable the analysis of data transfer between two levels of the memory hierarchy. We then present some examples of CO algorithms and data structures that are analyzed using the CO models in Section \ref{sec:COM-coalg} and \ref{sec:COM-cods}, respectively. Section  \ref{sec:COM-rco} presents a new relaxed cache-oblivious 
model on which a new concurrency-aware van Emde Boas (vEB) layout is devised (cf. Section  \ref{sec:concurrentvEB}).

\subsubsection{Preliminaries} \label{sec:COM-pre} 

\paragraph{I/O model.} \label{sec:iomodel}
The I/O\footnote{The term "I/O" is from now on used a shorthand for block I/O operations} model was introduced
by Aggarwal and Vitter \cite{AggarwalV88}. In their seminal paper, Aggarwal and Vitter
postulated that the memory hierarchy consists of two levels, an internal memory with size
$M$ (e.g., DRAM) and an external storage of infinite size (e.g., disks). Data is transferred 
in $B$-sized blocks between those two levels of memory and the CPU can only access 
data that are available in the internal memory. In the I/O model, an algorithm's time complexity 
is assumed to be dominated by how many block transfers are required, 
as loading data from disk to memory 
takes much more time than processing the data.

For this I/O model, B-tree \cite{Bayer:1972aa} is an optimal search tree \cite{CormenSRL01}.
B-trees and its concurrent variants \cite{BraginskyP12, Comer79, Graefe:2010:SBL:1806907.1806908,
Graefe:2011:MBT:2185841.2185842}
are optimized for a known memory block size $B$ (e.g., page size) to minimize the
number of memory blocks accessed by the CPU during a search, thereby improving data
locality. The I/O transfer complexity of B-tree is $O(\log_B N)$, the optimal.

However, the I/O model has its drawbacks. 
Firstly, to use this model, an algorithm has to know the $B$ and $M$ (memory size) 
parameters in advance.
The problem is that these parameters are sometimes unknown (e.g., when memory is shared with other applications)
and most importantly not portable between different platforms.
Secondly, in reality there are different block sizes at different levels of the memory hierarchy
that can be used in the design of locality-aware data layout for search trees. 
For example in \cite{KimCSSNKLBD10, Sewall:2011aa}, 
Intel engineers have come out with very fast search trees by crafting  
a platform-dependent data layout based on the register size, SIMD width, cache line size,
and page size. 

Existing B-trees limit spatial locality optimization to the memory level
with block size $B$, leaving access to other memory levels with different block size
unoptimized.
For example a traditional B-tree that is optimized for searching data in disks (i.e.,
$B$ is page size), where each node is an array of sorted keys, 
is optimal for transfers between a disk and RAM (cf. Figure \ref{fig:vEB-bfs}c).
However, data transfers between RAM and last level cache (LLC) 
are no longer optimal.
For searching a key inside each $B$-sized block in RAM, the transfer complexity 
is $\Theta (\log (B/L))$ transfers between RAM and LLC, 
where $L$ is the cache line size.
Note that a search with optimal cache line transfers of $O(\log_L B)$ is achievable
by using the van Emde Boas layout \cite{BrodalFJ02}. This layout has been proved 
to be optimal for search using the cache-oblivious model \cite{Frigo:1999:CA:795665.796479}. 

\paragraph*{Cache-oblivious model} \label{cache-oblivious}
The cache-oblivious model was introduced 
by Frigo et al. in \cite{Frigo:1999:CA:795665.796479},
which is similar to the I/O model except that the block size $B$ and memory size
$M$ are unknown.
Using the same
analysis of the Aggarwal and Vitter's two-level I/O model, an algorithm is categorized as 
\textit{cache-oblivious} if it has no variables that need to be tuned with respect to 
hardware parameters, 
such as cache size and cache-line length in order to achieve optimality, 
assuming that I/Os are performed by an optimal off-line cache replacement strategy.

If a cache-oblivious algorithm is optimal for arbitrary two-level memory, the
algorithm is also optimal for any adjacent pair of available levels of the memory hierarchy. 
Therefore without knowing anything about memory level hierarchy and the size of each level, a cache-oblivious
algorithm can automatically adapt to multiple levels of the memory hierarchy.
In \cite{Brodal:2004aa}, cache-oblivious algorithms were reported performing better 
on multiple levels of memory hierarchy and 
more robust despite changes in memory size parameters compared to the cache-aware
algorithms. 

One simple example is that in the cache-oblivious model, B-tree is no longer optimal
because of the unknown $B$. Instead, the van Emde Boas (vEB) layout-based trees that are 
described by Bender  
\cite{BenderDF05, BenderFFFKN07, BenderFGK05} and Brodal, 
\cite{BrodalFJ02}, are optimal.
We would like to refer the readers to \cite{Brodal:2004aa, Frigo:1999:CA:795665.796479} 
for a more comprehensive overview of the I/O model and cache-oblivious model.

We provide some of the examples of cache-oblivious algorithms 
and cache oblivious data structures in the following texts.

\subsubsection{Cache-oblivious Algorithms}  \label{sec:COM-coalg}

\paragraph{Scanning algorithms and their derivatives} 
One example of a naive cache-oblivious (CO) algorithm is the 
\textit{linear scanning} of an $N$ element array that requires $\Theta(N/B)$ I/Os or transfers.
Bentley's \textit{array reversal algorithm} and Blum's \textit{linear time selection algorithm}
are primarily based on the scanning algorithm, therefore they
also perform in $\Theta(N/B)$ I/Os \cite{Brodal:2004aa, Demaine:2002aa}.

\paragraph{Divide and conquer algorithms.} Another example of CO algorithms
in divide and conquer algorithms is the matrix operation algorithms.
Frigo et al. proved that \textit{transposition} of an $n \times m$ matrix 
was optimally solved in $\mathcal{O}(mn/B)$ I/Os and 
the \textit{multiplication} of an $m\times n$-matrix and an $n \times p$-matrix 
was solved using $\mathcal{O}((mn + np + mp)/B + mnp/(B\sqrt{M}))$ 
I/Os, where $M$ is the memory size \cite{Frigo:1999:CA:795665.796479}. As for square matrices (e.g., $N \times N$), using the Strassen's algorithm,
the required I/O bound has been proved to be $O(N^2/B + N^{\lg7}/B\sqrt{M})$, the optimal.

\paragraph{Sorting algorithms.} Demaine gave two examples of cache-oblivious 
sorting algorithm in his brief survey paper \cite{Demaine:2002aa}, namely the
\textit{mergesort} and \textit{funnelsort} \cite{Frigo:1999:CA:795665.796479}. 
In the same text he also wrote that both
sorting algorithms achieved the optimal $\Theta(\frac{N}{B} \log_2 \frac{N}{B})$ I/Os,
matching those in the original analysis of Aggarwal and Vitter \cite{AggarwalV88}.

\subsubsection{Cache-oblivious Data Structures}  \label{sec:COM-cods}

\paragraph{Static data structures}

One of the examples of cache-oblivious (CO) static data structures is the \textit{CO search trees} that 
can be achieved using the van Emde Boas (vEB) layout \cite{Prokop99, vanEmdeBoas:1975:POF:1382429.1382477}. 
The vEB-based trees recursively arrange related data in contiguous memory locations, 
minimizing data transfer between any two adjacent levels of the memory hierarchy (cf. Figure \ref{fig:vEB}). 

\begin{figure}[!t]
\centering  \includegraphics[width=0.8\columnwidth]{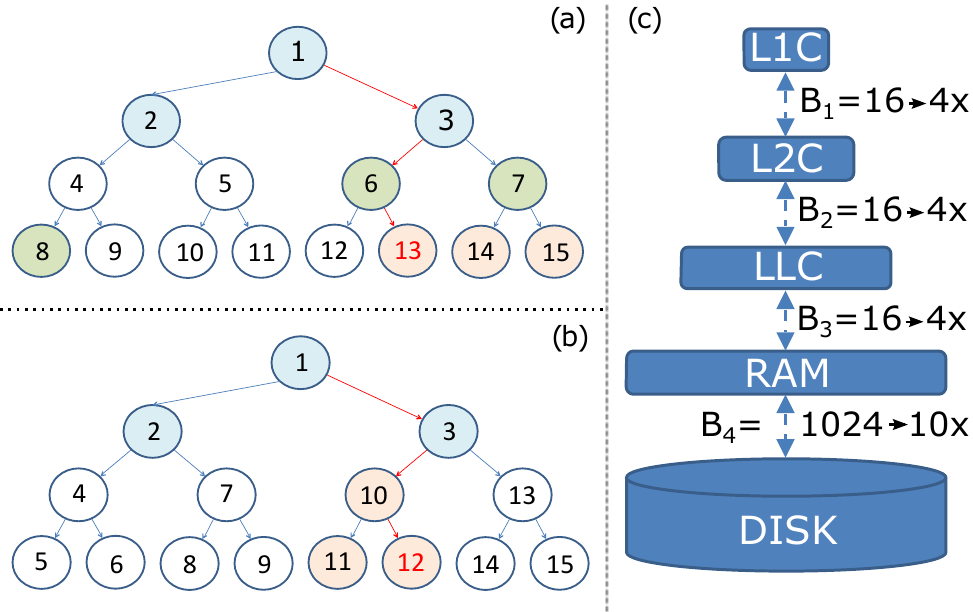}
\caption{Illustration of required memory transfers in searching for (a) key 13 in BFS tree and 
(b) key 12 in vEB tree, where node's value is its memory location. 
An example of multi-level memory is shown in (c), where $B_x$ is the block size 
$B$  between levels of memory.}\label{fig:vEB-bfs}
\end{figure}

Figure \ref{fig:vEB-bfs} illustrates the vEB layout, where the size $B$ of memory blocks transferred  
between 2-level memory in the I/O model \cite{AggarwalV88} is 3 (cf. Section \ref{sec:iomodel}). Traversing a complete binary tree with the Breadth First Search layout (or BFS tree for short) (cf. Figure \ref{fig:vEB-bfs}a) with height 4 will need three memory transfers to locate the key at leaf-node 13. The first two levels with
three nodes (1, 2, 3) fit within a single block transfer while the next two levels need to be loaded
in two separate block transfers that contain nodes (6, 7, 8)\footnote{For simplicity, we assume that the memory controller    transfers a memory block of 3 nodes starting at the address of the node requested.} 
and nodes (13, 14, 15), respectively. Generally, the number of memory transfers for a BFS tree of size $N$  is $(\log_2N-\log_2B) = \log_2 N/B \approx \log_2 N$
for $N \gg B$. 

For a vEB tree with the same height, the required memory transfers is only 
two. As shown in Figure \ref{fig:vEB-bfs}b, locating the key in leaf-node 12 requires only a transfer of nodes (1, 2, 3) 
followed by a transfer of nodes (10, 11, 12). Generally,  the memory transfer complexity for searching for a key in a tree of size $N$ is now reduced to $\frac{\log_2N}{\log_2B} = \log_B N$, simply by using 
an efficient tree layout so that nearby nodes are located in adjacent memory locations. 
If $B=1024$, searching a BFS tree for a key at a leaf requires 10x (or $\log_2 B$) more I/Os than searching a vEB tree with the same size $N$ where $N \gg B$. 

On commodity machines with multi-level memory, 
the vEB layout is even more efficient. 
So far the vEB layout is shown to have $\log_2B$ less I/Os for two-level memory. 
In a typical machine having three 
levels of cache (with cache line size of 64B), a RAM (with page size of 4KB) and a disk, 
searching a vEB tree can achieve up to 640x less I/Os than searching a BFS tree, assuming the node size is 4 bytes (Figure \ref{fig:vEB-bfs}c).
 
\paragraph{Dynamic data structures.}
In a standard \textit{linked-list} structure supporting traversals, insertions and deletions, 
the best-known cache-oblivious solution was $\mathcal{O}((\lg^2 N)/B)$ I/Os for updates
and $\mathcal{O}(K/B)$ for traversing $K$ elements in the list  \cite{Demaine:2002aa}.

The first cache-oblivious \textit{priority queue} was due to Arge et al. 
\cite{Arge:2002:CPQ:509907.509950} and it supports inserts and delete-min operations in $\mathcal{O}( {^1/_B} \log_{M/B} {^N/_B})$ I/Os.

The vEB layout in static cache-oblivious search tree has inspired many cache-oblivious \textit{dynamic search trees} such as 
cache-oblivious B-trees \cite{BenderDF05, BenderFFFKN07, BenderFGK05} and cache-oblivious binary trees \cite{BrodalFJ02}.
All of these search tree implementations have been proved having the optimal bounds of
$\mathcal{O}(\log_B N)$ in searches and require amortized $\mathcal{O}(\log_B N)$ I/Os for updates. 

However, vEB-based trees poorly support {\em concurrent} update operations. 
Inserting or deleting a node may result in 
relocating a large part of the tree in order to maintain 
the vEB layout (cf. Section \ref{subsec:staticveb}). Bender et al. 
\cite{BenderFGK05} discussed the problem
and provided important theoretical designs of concurrent vEB-based B-trees.
Nevertheless, we have found that the theoretical designs are not very efficient in practice 
due to the actual overhead of maintaining necessary pointers as well as their large memory footprint.

\subsubsection{New Relaxed Cache-oblivious Model}  \label{sec:COM-rco}

We observe that it is unnecessary to keep a vEB-based tree in a contiguous block of memory whose size is greater
than some upper bound. In fact, allocating a contiguous block of memory for a 
vEB-based tree does not guarantee a contiguous block of
\textit{physical memory}. Modern OSes and systems utilize
different sizes of continuous physical memory blocks,  for example, in the form
of pages and cache-lines. A contiguous block in virtual
memory might be translated into several blocks with gaps
in RAM; also, a page might be cached by several cache lines with gaps at any level of cache. 
This is one of the motivations for the new relaxed cache oblivious model proposed. 

We define {\em relaxed cache oblivious} algorithms to be cache-oblivious (CO)
algorithms with the restriction that an upper bound $\UB$ on the unknown memory
block size $B$ is known in advance. 
As long as an upper bound on all the block
sizes of multilevel memory is known, the new relaxed CO model maintains the key
feature of the original CO model \cite{Frigo:1999:CA:795665.796479}. 
First, temporal locality is exploited perfectly as there are no constraints on cache size
$M$ in the model. As a result, an optimal offline cache replacement policy can be
assumed. In practice, the Least Recently Used (LRU) policy with memory of size 
$(1+\epsilon)M$, where $\epsilon>0$, is nearly as good as the optimal replacement policy
with memory of size $M$ \cite{Sleator:1985:AEL:2786.2793}.
Second, analysis for a simple two-level memory
are applicable for an unknown multilevel memory (e.g., registers, L1/L2/L3 caches
and memory). Namely, an algorithm that is optimal in terms of data movement for a 
simple two-level memory is asymptotically optimal for an unknown multilevel memory. 
This feature enables algorithm designs that can utilize fine-grained data locality 
in the multilevel memory hierarchy of modern architectures. 

The upper bound on the contiguous block size can be obtained easily from 
any system (e.g., page-size or any values greater than that), which is platform-independent. In fact, the search performance in  
the new relaxed cache oblivious model is resilient to different upper bound values (cf. Lemma \ref{lem:dynamic_vEB_search} in Section \ref{sec:concurrentvEB}).

\subsubsection{New Concurrency-aware van Emde Boas Layout} \label{sec:concurrentvEB}

\begin{figure}[!t] 
\centering 
\scalebox{0.6}{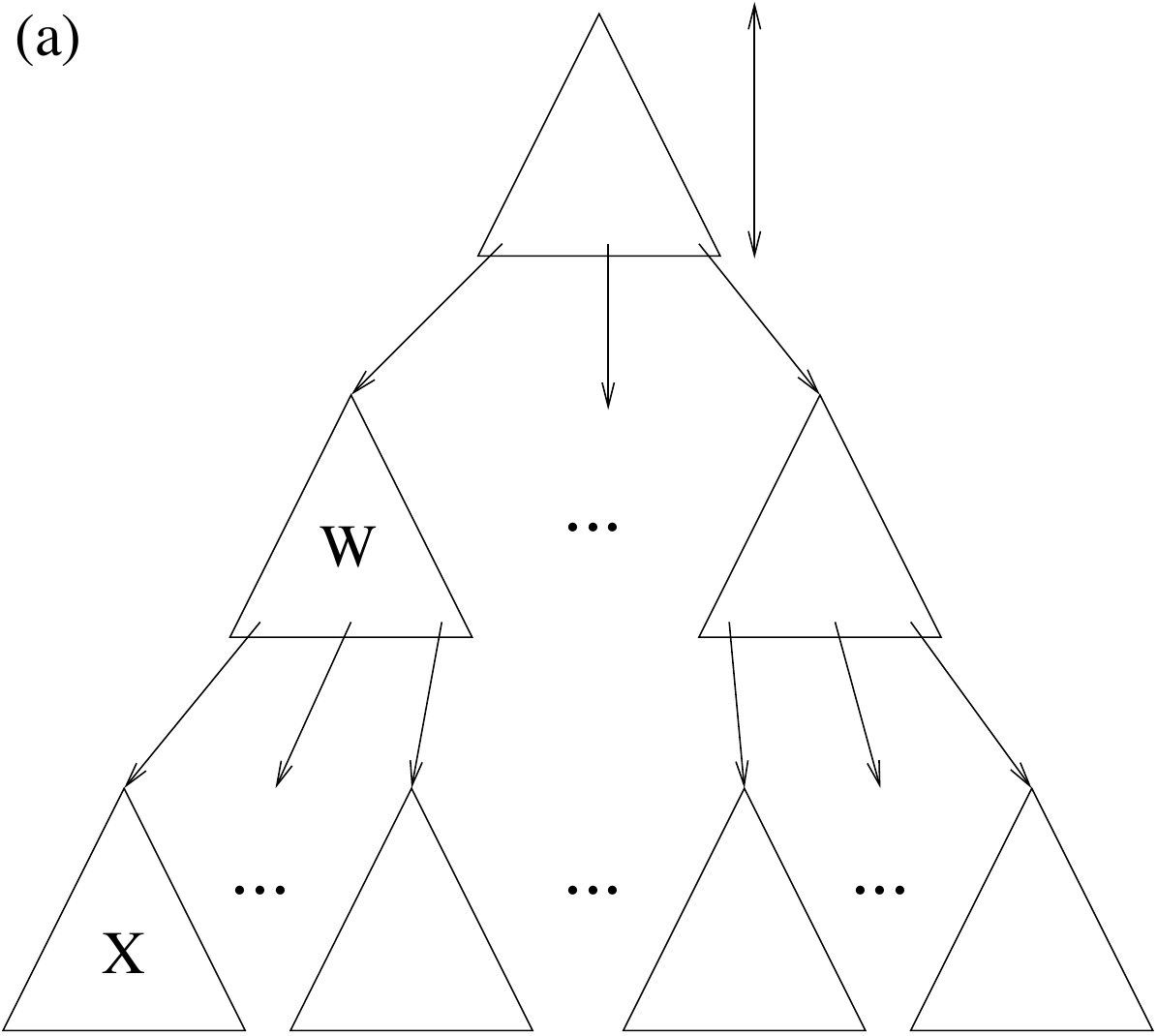}
\scalebox{0.6}{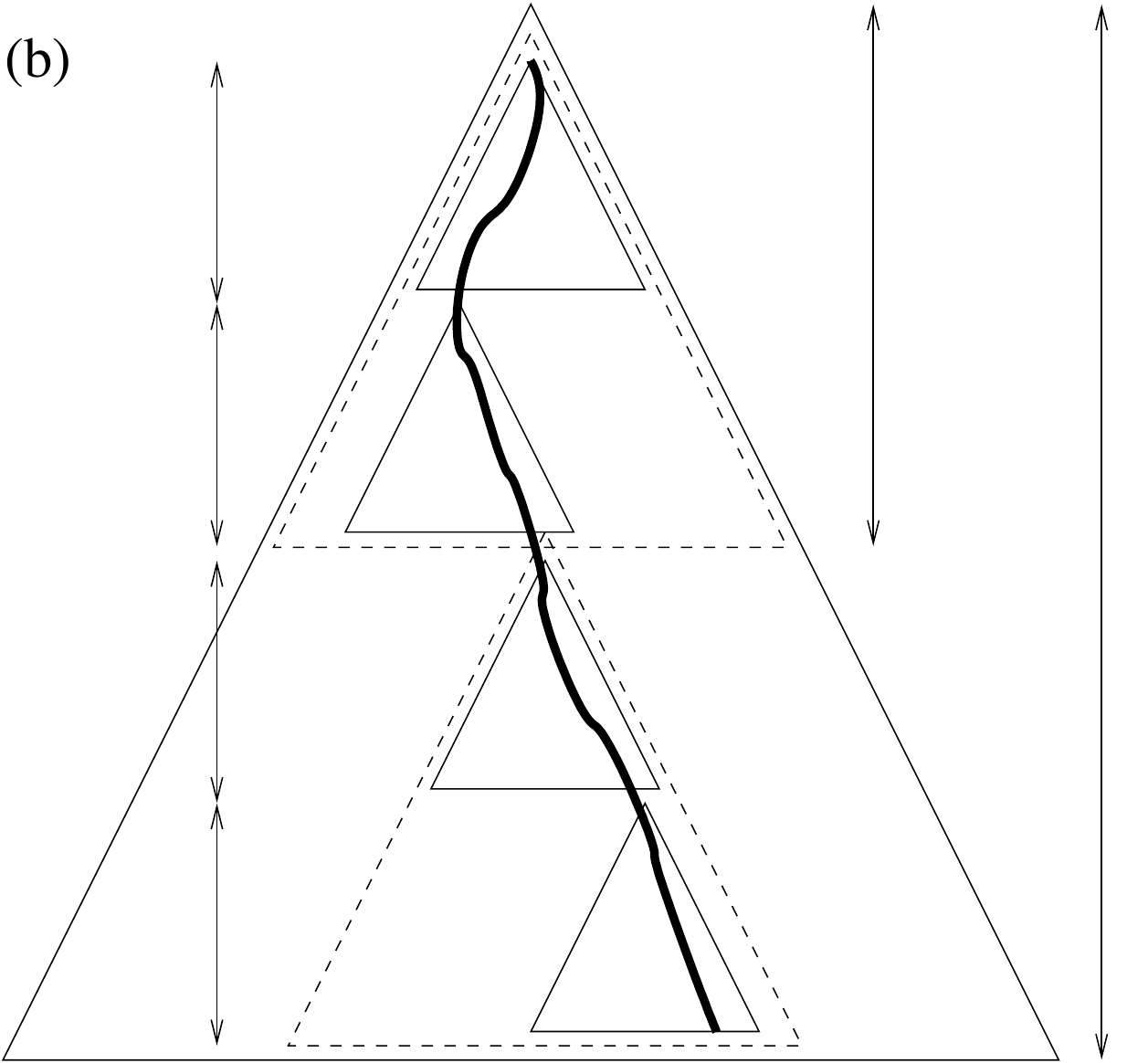}
\caption{{\em(a)} New concurrency-aware vEB layout. {\em(b)} 
Search using concurrency-aware vEB layout.}\label{fig:search_complexity}\label{fig:dynamicVEB}
\end{figure}

We propose improvements to the conventional van Emde Boas (vEB) layout
to support high performance and high concurrency, which results in new {\em concurrency-aware} dynamic vEB layout. 
We first define the following notations that will be used to elaborate on the improvements:
\begin{itemize}

\item $b_i$ (unknown): block size in terms of the number of nodes at level $i$ of the memory
hierarchy (like $B$ in the I/O model \cite{AggarwalV88}), which is unknown as in the cache-oblivious 
model \cite{Frigo:1999:CA:795665.796479}. When the specific level $i$ of the memory
hierarchy is irrelevant, we use notation $B$ instead of $b_i$ in order to be
consistent with the I/O model.

\item $\UB$ (known): the upper bound (in terms of the number of nodes) on the
block size $b_i$ of all levels $i$ of the memory hierarchy.

\item {\em $\Delta$Node}: the largest recursive subtree of a van Emde Boas-based 
search tree that contains at most $\UB$ nodes (cf. dashed triangles of height $2^L$ in
Figure \ref{fig:search_complexity}b). $\Delta$Node is a fixed-size tree-container
with the vEB layout.

\item "level of detail" $k$ is a partition of the tree into recursive subtrees of 
height at most $2^k$. 

\item Let $L$ be the level of detail of $\Delta$Node. Let $H$ be the height
of a $\Delta$Node, we have $H = 2^L$. For simplicity, we assume $H = \log_2
(\UB+1)$.

\item $N, T$: size and height of the whole tree in terms of basic nodes (not in
terms of $\Delta$Nodes).

\end{itemize}  

\subparagraph{Conventional van Emde Boas (vEB) layout.} \label{subsec:staticveb}

\begin{figure}[!t]
\centering \scalebox{0.7}{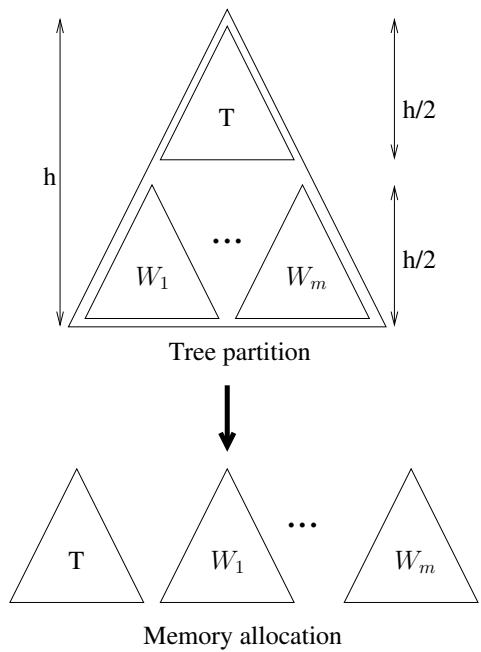}
\caption{Static van Emde Boas (vEB) layout: a tree of height $h$ is recursively split at height $h/2$. 
The top subtree $T$ of height $h/2$ and $m=2^{h/2}$ bottom subtrees $W_1;W_2; \ldots ;W_m$ 
of height $h/2$ are located in contiguous memory locations where T is located before  
$W_1;W_2;\ldots;W_m$.}\label{fig:vEB}
\end{figure}

The conventional van Emde Boas (vEB) layout has been introduced in
cache-oblivious data structures \cite{BenderDF05, BenderFFFKN07, BenderFGK05, BrodalFJ02,
Frigo:1999:CA:795665.796479}. Figure \ref{fig:vEB} illustrates the vEB layout.
Suppose we have a complete binary tree with height $h$. For simplicity, we
assume $h$ is a power of 2, i.e., $h=2^k, k \in \mathbb{N}$.
The tree is recursively laid out in the memory as follows. The tree is
conceptually split between nodes of height $h/2$ and $h/2+1$, resulting in a top
subtree $T$ and $m_1 = 2^{h/2}$ bottom subtrees $W_1, W_2, \cdots, W_{m_1}$ of
height $h/2$. The $(m_1 +1)$ top and bottom subtrees are then located in
contiguous memory locations where $T$ is located before $W_1, W_2, \cdots,
W_{m_1}$. Each of the subtrees of height $h/2$ is then laid out similarly to $(m_2 +
1)$ subtrees of height $h/4$, where $m_2 = 2^{h/4}$. The process continues until
each subtree contains only one node, i.e., the finest {\em level of detail}, 0.

The main feature of the vEB layout is that the cost of any search in this layout
is $O(\log_B N)$ memory transfers, where $N$ is the tree size and $B$ is the {\em
unknown} memory block size in the cache-oblivious model \cite{Frigo:1999:CA:795665.796479}. Namely, its search is cache-oblivious. 
The search cost is the
optimal and matches the search bound of B-trees that requires the memory block
size $B$ to be {\em known in advance}. Moreover, at any level of detail, each subtree
in the vEB layout is stored in a contiguous block of memory.

Although the conventional vEB layout is helpful for utilizing data locality, it poorly
supports concurrent update operations. Inserting (or deleting) a node at
position $i$ in the contiguous block storing the tree may restructure a large
part of the tree. For example, inserting
new nodes in the full subtree $W_1$ (a leaf subtree) in Figure \ref{fig:vEB} will  affect the other
subtrees $W_2, W_3, \cdots, W_m$ by rebalancing existing nodes between $W_1$ 
and the subtrees in order to have
space for new nodes. Even worse, we will need to allocate a new contiguous block of
memory for the whole tree if the previously allocated block of memory for
the tree runs out of space \cite{BrodalFJ02}. Note that we cannot use dynamic
node allocation via pointers since at {\em any} level of detail, each subtree in the vEB layout must be stored in a {\em
contiguous} block of memory.

\subparagraph{Concurrency-aware vEB layout.} \label{sec:relaxed-veb}

In order to make the vEB layout suitable for highly concurrent data structures
with update operations, we introduce a novel {\em concurrency-aware} dynamic vEB layout. Our key
idea is that if we know an upper bound $\UB$ on the unknown memory block size
$B$, we can support dynamic node allocation via pointers while maintaining the
optimal search cost of $O(\log_B N)$ memory transfers without knowing $B$ (cf.
Lemma \ref{lem:dynamic_vEB_search}). The assumption on known upper bound $\UB$ is supported 
by the fact that in practice it is unnecessary to keep the vEB layout in
 a contiguous block of memory whose size is greater than some upper bound.
         
Figure \ref{fig:dynamicVEB}a illustrates the new concurrency-aware vEB layout based on the
relaxed cache oblivious model. Let $L$ be the coarsest level of detail such that
every recursive subtree contains at most $\UB$  nodes. Namely, let $H$ and $S$ 
be the height and size of such a subtree then $H=2^L$ and $S=2^H - 1 < \UB$.
The tree is recursively
partitioned into level of detail $L$ where each subtree represented by a
triangle in Figure \ref{fig:dynamicVEB}a,  is stored in a contiguous memory block
of size $\UB$. Unlike the conventional vEB, the subtrees at level of detail $L$
are linked to each other using pointers, namely each subtree at level of detail
$k > L$ is not stored in a contiguous block of memory.  Intuitively, since $\UB$
is an upper bound on the unknown memory block size $B$, storing a subtree at
level of detail $k > L$ in a contiguous memory block of size greater than $\UB$,
does not reduce the number of memory transfers, provided there is perfect alignment. 
For example, in Figure
\ref{fig:dynamicVEB}a, traveling from a subtree $W$ at level of detail $L$, which
is stored in a contiguous memory block of size $\UB$, to its child subtree $X$ at
the same level of detail will result in at least two memory transfers: one for
$W$ and one for $X$. Therefore, it is unnecessary to store both $W$ and $X$ in a
contiguous memory block of size $2\UB$. As a result, the memory transfer cost for search operations in
the new concurrency-aware vEB layout is intuitively the same as that of the conventional vEB layout (cf. Lemma \ref{lem:dynamic_vEB_search}) while the concurrency-aware vEB supports
high concurrency with update operations.

\begin{lemma}
For any upper bound $\UB$ of the {\em unknown} memory block size $B$, a search in a complete binary tree with the new concurrency-aware vEB layout achieves the optimal memory transfer 
$O(\log_B N)$, where $N$ and $B$ are the tree size and 
the {\em unknown} memory block size in the cache-oblivious model
\cite{Frigo:1999:CA:795665.796479}, respectively.
\label{lem:search_mem} \label{lem:dynamic_vEB_search}
\end{lemma}
\begin{proof} (Sketch)
Figure \ref{fig:search_complexity}b illustrates the proof.  
Let $k$ be the coarsest level of detail such that every recursive subtree 
contains at most $B$ nodes. Since $B \leq \UB$, $k \leq L$, where $L$ is 
the coarsest level of detail at which every recursive subtree ($\Delta$Nodes) 
contains at most $\UB$ nodes. 
That means there are at most $2^{L-k}$ subtrees
along the search path in a $\Delta$Node and no subtree of depth $2^k$ is split 
due to the boundary of $\Delta$Nodes. Namely, triangles of height $2^k$ fit
within a dashed triangle of height $2^L$ in Figure \ref{fig:search_complexity}b. 

Because at any level of detail $i \leq L$ in the concurrency-aware vEB layout, a recursive subtree of depth $2^i$ is stored in a contiguous block 
of memory, each subtree of depth $2^k$ {\em within} a $\Delta$Node is stored in
at most 2 memory blocks of size $B$ (depending on the starting location of
the subtree in memory). Since every subtree of depth $2^k$ fits in a
$\Delta$Node (i.e.,
no subtree is stored across two $\Delta$Nodes), every subtree of depth $2^k$ is 
stored in at most 2 memory blocks of size $B$.

Since the tree has height $T$, $\lceil T / 2^k \rceil$ subtrees of depth $2^k$ 
are traversed in a search and thereby at most  $2  \lceil T / 2^k \rceil$ memory 
blocks are transferred. 

Since a subtree of height $2^{k+1}$ contains more than $B$ nodes, 
$2^{k+1} \geq \log_2 (B + 1)$, or $2^{k} \geq \frac{1}{2} \log_2 (B+ 1)$. 

We have $2^{T-1} \leq N \leq 2^T$ since the tree is a {\em complete} binary tree. 
This implies $ \log_2 N \leq T \leq \log_2 N +1$.  

Therefore, the number of memory blocks transferred in a search is 
$2  \lceil T / 2^k \rceil \leq 4 \lceil \frac{\log_2 N + 1}{\log_2 (B + 1)} \rceil 
= 4 \lceil \log_{B+1} N + \log_{B+1} 2\rceil$ $= O(\log_B N)$, where $N \geq 2$.
\end{proof}

A library of novel locality-aware and energy efficient concurrent search trees based on the new concurrency-aware vEB layout is presented in Section~\ref{sec:search-trees}


\subsection{Power Model for Computational Algorithms on Movidius Platform} \label{sec:myriad1-model}
The objective of this study is to build a power model that can estimate the power consumption of an algorithm on Myriad platform. In order to do so, the model considers both algorithmic and platform properties. Our model is inspired by Amdahl law analysis and the Roofline model of energy \cite{ 6877278, Choi:2013:RME:2510661.2511392}. Although the Roofline model also connects algorithmic and platform properties, it does not consider the number of cores nor memory contention when the number of cores varies. By estimating the power consumed by a system running a specific number of cores, our model can predict the number of cores required to achieve the least energy consumption. In this report, the model has been evaluated by micro-benchmarks and application kernels such as sparse/dense linear algebra kernels and graph kernels on Movidius embedded  platform (Myriad1).


\subsubsection{Energy Model Description}
In Deliverable D2.1, we presented our initial ideas on an energy model for Myriad1 platform. In this Deliverable, we have improved the model by considering both computational and data movement power. Another improvement is that we create the micro-benchmarks whose sizes are approximately 1 KB that fits to Myriad1 cache buffer. The measurements are, therefore, more accurate and do not need to include inter-operation cost like the power model in Deliverable D2.1. The details on how to build the model and its evaluation are described as follows.
\paragraph{Computational Power}
The computational power of a system is the required power to perform its computation using data from closest memory components such as register files. The computational power includes static power, active power and dynamic power of involved functional units. At first we aim to find out the corresponding values of static, active and dynamic power of each SHAVE core and each SHAVE arithmetic unit.  
The experimental results have shown that the power consumption of Movidius Myriad1 platform is ruled by the following model:
\begin{equation} \label{eq:Pfirst}
	P^{comp}=P^{sta}+\#\{\text{SHV}\} \times (P^{act}+P^{dyn}_{SHV})
\end{equation}
In that formula, the static power $P^{sta}$ is the needed power when the Myriad1 processor is on. The $P^{act}$ is the power consumed when a SHAVE core is on. Therefore, if benchmarks or programs use $n$ SHAVE cores, the active power needs to be multiplied with the number of used SHAVE cores.

The dynamic power $P^{dyn}_{SHV}$ of each SHAVE is the power consumed by all working operation units in one SHAVE. Operation units include arithmetic units (e.g., IAU, VAU, SAU and CMU) and load store unit (e.g., LSU1, LSU2). The experimental results show that different arithmetic operation units have different $P^{dyn}$ values.  


When running a benchmark with one more SHAVE, we can identify the sum of SHAVE $P^{act}$ and $P^{dyn}$ which is the power difference of the two runs (the run with one SHAVE core and the run with two SHAVEs).  Given the sum of $P^{act}$ and $P^{dyn}$, $P^{sta}$ is derived from the Equations~\ref{eq:Pfirst} . The experimental results show the average value of $P^{sta}$ from all micro-benchmarks:
\begin{equation} \label{eq:Pstat}
	P^{sta}=61.81 \text{ mW}
\end{equation}

Dynamic power $P^{dyn}_{SHV}$ of SHAVE running multiple units is computed by the following formula.
\begin{equation} \label{eq:Pdynshave}
	P^{dyn}_{SHV}=\sum_i{P^{dyn}_i(op)}
\end{equation}
That means $P^{dyn}_{SHV}$ is the sum of $P^{dyn}$ of all involved arithmetic units. E.g. 
$P^{dyn}_{(SauIau)}) = P^{dyn}_{(Sau)} + P^{dyn}_{(Iau)}$.

For each operation unit, we obtain the two parameters $P^{dyn}_{op}$ and $P^{act}$ by using the actual power consumption of the benchmark for individual units and multiple units. 

For example, $P^{dyn}_{Iau}$,  $P^{dyn}_{Sau}$ and $P^{act}$ are identified from Equations~\ref{eq:PdynIau}, ~\ref{eq:PdynSau}, ~\ref{eq:PdynSauIau}.
\begin{equation} \label{eq:PdynIau}
	P^{dyn}_{(Iau)} =P^{sta} + \#\{\text{SHV}\} \times (P^{act}+P^{dyn}_{(Iau)})
\end{equation}
\begin{equation} \label{eq:PdynSau}
	P^{dyn}_{(Sau)}= P^{sta} + \#\{\text{SHV}\} \times (P^{act}+P^{dyn}_{(Sau)})
\end{equation}
\begin{equation} \label{eq:PdynSauIau}
	P^{dyn}_{(SauIau)}= P^{sta} + \#\{\text{SHV}\} \times (P^{act}+P^{dyn}_{(Sau)}+P^{dyn}_{(Iau)})
\end{equation}

Then, the average value of  $P^{act}$ from all micro-benchmarks is derived as below.
\begin{equation} \label{eq:Pact}
	P^{act}=29.33 \text{ mW}
\end{equation}

Given the values of \textbf{$P^{sta}$}, \textbf{$P^{act}$} and \textbf{$P^{dyn}_{SHV}$} as Equations ~\ref{eq:Pstat},~\ref{eq:Pact} and~\ref{eq:Pdynshave} respectively, the computation power \textbf{$P^{comp}$} for Movidius Myriad1 can be estimated by applying the below formula:
\begin{equation} \label{eq:Psecond}
	P^{comp}=P^{sta}+ \#\{\text{SHV}\} \times \left(P^{act} + \sum_i{P^{dyn}_i(op)} \right)
\end{equation}

Table~\ref{table:Punits} lists the dynamic power of each operation unit. 
\begin{table}[h]
\begin{center}
\begin{tabular}{|l|l|}
\hline
\textbf{Operation Unit} & \textbf{$P^{dyn}_{op}$}\\ \hline
SauXor                  & 14.68          \\ \hline
SauMul                  & 17.69          \\ \hline
VauXor                  & 34.34         \\ \hline
VauMul                 & 51.98         \\ \hline
IauXor                  & 15.91          \\ \hline
IauMul                  & 18.48          \\ \hline
CmuCpss                 & 12.62          \\ \hline
CmuCpivr                & 18.84         \\ \hline
LsuLoad                & 29.87         \\ \hline
LsuStore                & 37.49          \\ \hline
\end{tabular}
\caption{\textbf{$P^{dyn}_{op}$} of SHAVE Operation Units}
\label{table:Punits}
\end{center}
\end{table}


\paragraph{Data Movement Power} \label{sec:Data-Movement-Power}

Since Myriad1 has two memory components within the chip (CMX and DDR), we model the data movement for both of the memory components. The data is either moved from CMX to registers or from DDR to registers before being processed by arithmetic units.

Data movement is performed by two units of the SHAVE core namely Load Store Unit (LSU). LSU loads and stores data from the memory components such as DDR and CMX to the register memory of SHAVE processor. Therefore, the consumed power of data movement also includes the static power {$P^{sta}$}, the active power {$P^{act}$}, LSU unit power {$P^{dyn}_{LSU}$} and the contention power {$P^{ctn}$} when there are more cores accessing memory than the number of available memory ports. Below is the formula to estimate the power consumption of data movement. 
\begin{equation} \label{eq:Pdata}
	P^{data}= P^{sta} + min(m,n) \times (P^{act}+P^{dyn}_{LSU}) + max(n-m, 0) \times P^{ctn}
\end{equation}
In the formula,  $n$ is the number of active cores running the program;  $m$ is the number representing memory ports or bandwidth available in the platform; the contention power $P^{ctn}$ is the power overhead occurring when SHAVE cores actively wait for accessing data because of the limited memory ports (or bandwidth) in the platform architecture. 

Myriad1 has two separate memory components: CMX and DDR. As described in section 2.2.1, each SHAVE core has its own CMX memory slice meaning that the number of data accessing ports $m$ equals to the number of active cores $n$ when transferring data between CMX memory and registers. Therefore, data movement for CMX does not cause contention power and the power consumption of CMX data movement {$P^{data}_{CMX}$} is:
\begin{equation}
	P^{data}_{CMX}= P^{sta} + min(n,n) \times (P^{act}+P^{dyn}_{LSU}) + max(0, 0) \times P^{ctn}
\end{equation}
\begin{equation} \label{eq:Pdata_CMX}
	P^{data}_{CMX}= P^{sta} + n \times (P^{act}+P^{dyn}_{LSU})
\end{equation}

The data movement between DDR and registers, however, has the contention power due to the limited DDR ports/bandwidth shared among all SHAVE cores.  
\begin{equation} \label{eq:Pdata}
	P^{data}_{DDR}= P^{sta} + min(m,n) \times (P^{act}+P^{dyn}_{LSU}) + max(n-m, 0) \times P^{ctn}
\end{equation}

\paragraph{A Power Model for Computational Algorithms}

Typical applications require both computation and data movement. The power consumption then needs to consider the power of both the computation and data movement. In our model, we use the concept of operational intensity proposed by Williams et.al.~\cite{Williams:2009:RIV:1498765.1498785} to characterize algorithms. An algorithm can be characterized by the amount of computational work $W$ and a number of data accesses $Q$. $W$ is the number of operations performed by a  program. $Q$ is the number of transferred bytes required during a program execution. Both $W$ and $Q$ define the operation intensity $I$ of the algorithm. 
\begin{equation} \label{eq:Pdata}
	I= \frac{W}{Q}
\end{equation}
For the power model, we assume that the core either perform computation or transfer data during its active state. The power consumption of a processor depends on the ratio of the computation to data movement during its execution. This portion can be calculated based on the time ratio of performing computation to transferring data. As the time needed to perform one operation is different from the time required to transfer one byte of data, we introduce a parameter {$\alpha$} to the model. The parameter {$\alpha$} is the property of the platform and its value depends on each platform architecture. The model considers both computation and data movement cost is derived below.
\begin{equation} \label{eq:Pfinal}
	P= P^{comp} \times (\frac{W}{\alpha \times Q + W}) + P^{data} \times (\frac{\alpha \times Q}{\alpha \times Q + W})
\end{equation}
After converting W and Q to I by using the Equation ~\ref{eq:Pfinal}, the final model is simplified as below:
\begin{equation} \label{eq:Pfinal1}
	P= P^{comp} \times (\frac{I}{\alpha+I}) + P^{data} \times (\frac{\alpha}{\alpha+I})
\end{equation}
The complete model with details to compute the power consumed by a given application/ algorithm is presented as below.
\begin{equation} 
\label{eq:Pfinal2} 
\begin{split}
	P&= P^{sta} \\
	&+ n \times (P^{act}+P^{dyn}_{SHV}) \times (\frac{I}{\alpha+I}) \\
	&+ (min(m,n) \times (P^{act}+P_{LSU}) + max(n-m, 0) \times P^{ctn}) \times (\frac{\alpha}{\alpha+I})
\end{split}
\end{equation}

Table~\ref{table:Parameters} summarizes the parameter list of the proposed model. 
\begin{table}[h]
\begin{center}
\begin{tabular}{|l|l|}
\hline
\textbf{Parameters} & \textbf{Explanation}\\ \hline
{$P^{sta}$}                 & Static power of a SHAVE core         \\ \hline
{$P^{act}$}                 & Active power of a SHAVE core        \\ \hline
{$P^{dyn}_{SHV}$}                  & Dynamic power of a SHAVE core        \\ \hline
{$P_{LSU}$}                  & Operation power of Load Store Unit          \\ \hline
{$P^{ctn}$}                  & Contention power of a SHAVE core          \\ \hline
$m$                 & Number of memory ports         \\ \hline
$n$                & Number of running SHAVE cores         \\ \hline
$I$              & Operation intensity of an algorithm         \\ \hline
{$\alpha$}                & Time ratio of data transfer to computation         \\ \hline
\end{tabular}
\caption{Model Parameter List}
\label{table:Parameters}
\end{center}
\end{table}

\subsubsection{Model Validation}
\paragraph{Experimental validation with micro-benchmarks}
Regarding the assembly code of micro-benchmarks, a fixed number of instructions is provided for all the experiments, meaning that each assembly code file contains five
instructions in an iteration and the iteration is infinitely repeated. This convention keeps the continuity and consistency of experiments while giving insights into measuring the power consumption with different SHAVE units, enabling the comparisons among SHAVE units. As compared to micro-benchmarks used in Deliverable D2.1, the number of instructions in one iteration for each mico-benchmark is re-calculated so that the whole iteration fits to L1 cache buffer size (1KB). 
 
The assembly code files used in the experimental evaluation contain code that executes the instruction decode and instruction fetch. Majority of tests use pseudo-realistic data. By using pseudo-realistic data we have as many non-zero values as possible and avoiding data value repetition at different offsets. Analyses of experimental results are performed based on an identified set of micro-benchmarks. Each micro-benchmark is executed with different numbers of SHAVE such as 1, 2, 4, 6 and 8 SHAVE cores. A sample code of micro-benchmark performing SauXor operation is given in Figure ~\ref{fig:MicrobenchmarkCode}. 

\begin{figure}[!t] \centering
\resizebox{0.7\columnwidth}{!}{ \includegraphics{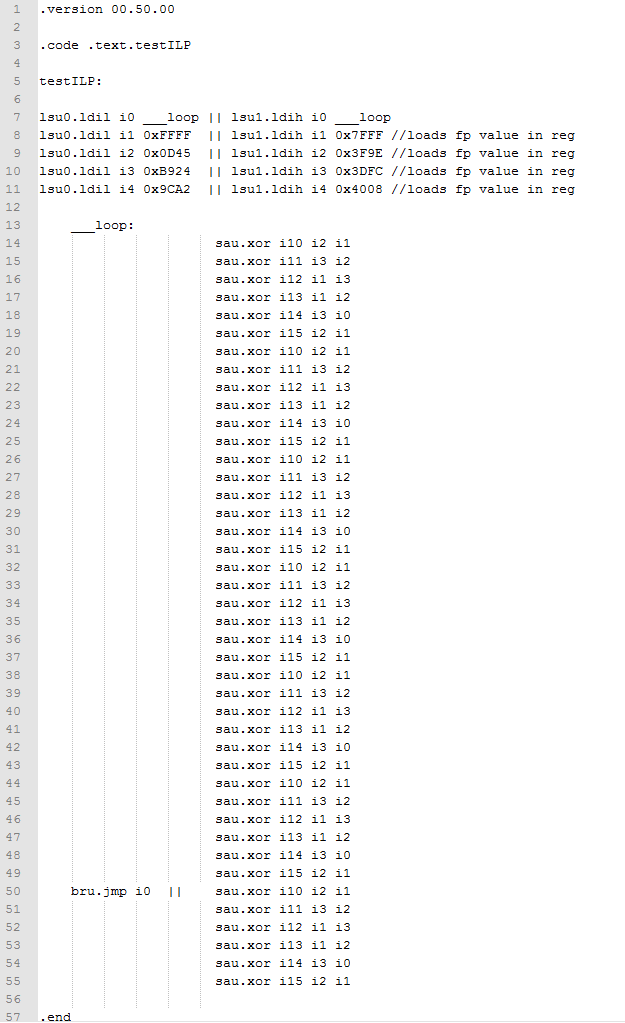}}
\caption{ An example of micro-benchmark is written in asm code to perform SauXor operations}
\label{fig:MicrobenchmarkCode}
\end{figure}

\subparagraph{Computation micro-benchmarks}

Applying this model to the computational micro-benchmarks for each of operation units, the model relative error is plotted in the Figure ~\ref{fig:OperationUnitError}. Having all parameter data from Equation ~\ref{eq:Pstat}, ~\ref{eq:Pact} and Table ~\ref{table:Punits}, the model-predicted values are computed. Relative errors are the percentage difference between the actual power consumption measured by device and the predicted values from the model. 

Under this model, the relative error which is the percentage difference between measured and estimated values varies from -5\% to 6\% which is an 11\% margin. This proves that the model can be applied for micro-benchmarks running a single unit or any combination of two or three units in parallel. This model shows the compositionality of power consumption not only for multiple SHAVE cores but also for multiple operation units within a SHAVE.

\begin{figure}[!t] \centering
\resizebox{0.6\columnwidth}{!}{ \includegraphics{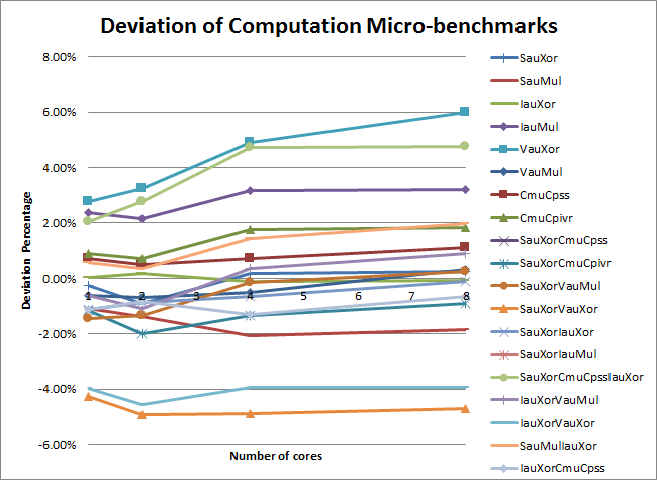}}
\caption{ Relative error of the model for computational micro-benchmarks}
\label{fig:OperationUnitError}
\end{figure}

\subparagraph{Data movement micro-benchmarks}
Data movement micro-benchmarks are the micro-benchmarks running LSU units to transfer data from CMX memory to core registers. The micro-benchmarks are designed to execute load/store operation only or with other functional units. When executed with other computation units, the instructions are in different pipelines to make both LSU and computation units executed in parallel manner.

Since each of the eight SHAVE cores in Myriad1 can access CMX data with its own port, this matches the Equation ~\ref{eq:Pdata_CMX} as presented in Section ~\ref{sec:Data-Movement-Power}.

 The experimental results of LSU also follow the power model shown in Figure ~\ref{fig:DataMovementError}. The relative error of data movement micro-benchmarks  varies from -4\% to 6\% which is a 10\% margin. This proves that the model can be applied to Myriad1 platform accessing CMX data.

\begin{figure}[!t] \centering
\resizebox{0.6\columnwidth}{!}{ \includegraphics{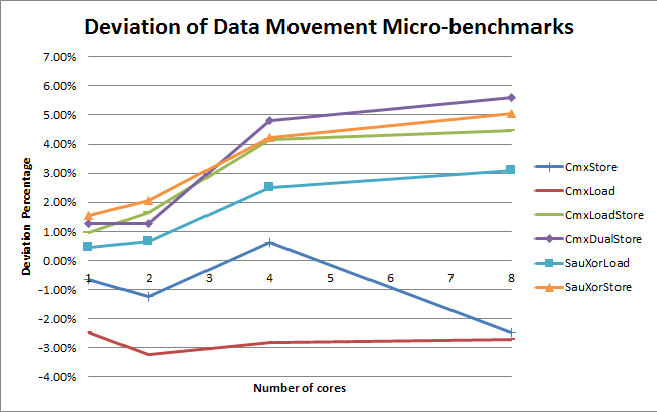}}
\caption{ Relative error of the model for data movement micro-benchmarks}
\label{fig:DataMovementError}
\end{figure}

\subparagraph{Intensity-based micro-benchmarks: the combination of computation and data movement}
Since any application requires both computation and data movement, we design micro-benchmarks which execute both computation and data movement units. Unlike the combination of multiple units in different instruction pipelines as computational micro-benchmarks, this set of micro-benchmarks contains sequential instructions to trigger functional units in sequential order. With this way, the assembly code does not use the pipeline optimization feature of Myriad1 platform and therefore, provide raw performance for a more precise analysis.

In order to validate the power model for any algorithms, the micro-benchmarks used in experiments at this phase also indicate different operation intensities. Operation intensity $I$ is retrieved from the assembly code by counting the number of arithmetic instructions such as Xor, Mul and the number of LSU instructions. The number of arithmetic instructions indicate the amount of work $W$ and LSU instructions multiplied by four (each LSU instruction load 32 bits which equals to four bytes) indicate the number of accessed data bytes $Q$. By changing the ratio of arithmetic instructions to LSU instructions, we created micro-benchmarks with different intensity values. A sample code of micro-benchmark with intensity $I=0.25$ is given in Figure ~\ref{fig:MicrobenchmarkIntensityCode}. 

\begin{figure}[!t] \centering
\resizebox{0.7\columnwidth}{!}{ \includegraphics{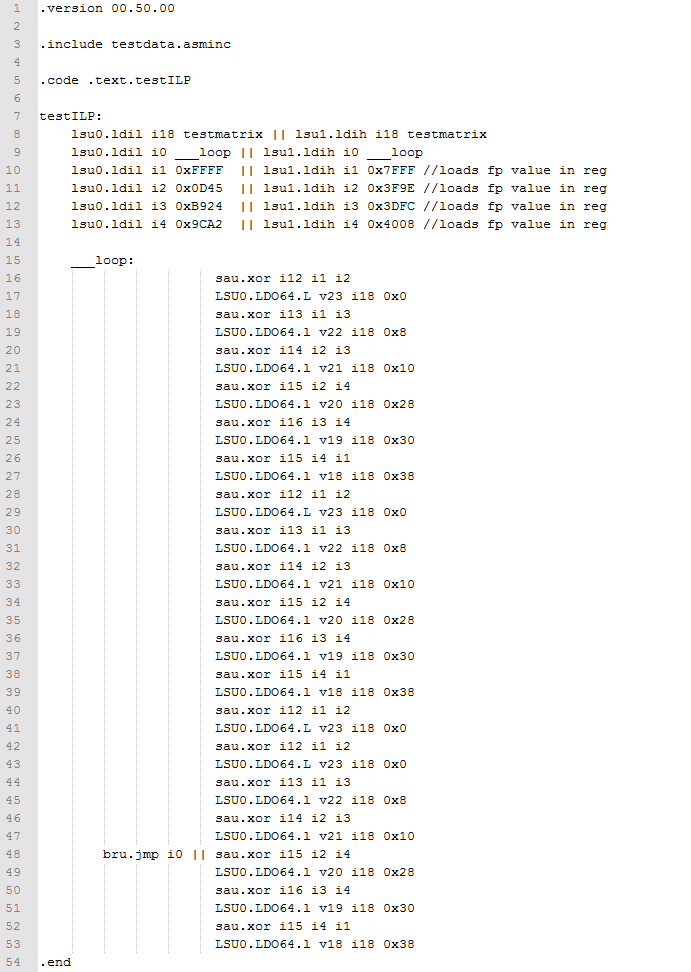}}
\caption{ An example of micro-benchmark is written in asm code with intensity $I=0.25$}
\label{fig:MicrobenchmarkIntensityCode}
\end{figure}

Intensity-based micro-benchmarks using CMX memory to store data has experimental results shown in Figure~\ref{fig:CMXError}. Since eight SHAVE cores have eight accessing ports to CMX memory, cores do not wait for memory fetching and there is no contention power. The power model of Myriad1 using CMX memory is the Equation \ref{eq:PCMX}.

\begin{equation} 
\label{eq:PCMX} 
\begin{split}
	P&= P^{sta} + n \times (P^{act}+P^{dyn}_{SHV}) \times (\frac{I}{\alpha+I}) \\
	&+ n \times (P^{act}+P_{LSU}) \times (\frac{\alpha}{\alpha+I})
\end{split}
\end{equation}

In the CMX power model in Equation \ref{eq:PCMX}, all parameters are known except $\alpha$. By using non-regression techniques in Matlab function \textit{lsqcurvefit} to find unknown parameters, $\alpha$ is identified as equal to one. This is reasonable since accessing four bytes of data in CMX requires average four cycles \cite{D4.1} 
which means one cycle per byte and an operation used in micro-benchmarks (e.g. SauXor) also requires one cycle to be executed.

We plot the model deviation error of Myriad using CMX in Figure \ref{fig:CMXError}. The relative error of intensity micro-benchmarks varies from -9\% to 4\% which is a 13\% margin. This proves that the model can be applied to estimate the power consumption of an algorithm on Myriad1 platform accessing CMX data.

\begin{figure}[!t] \centering
\resizebox{0.6\columnwidth}{!}{ \includegraphics{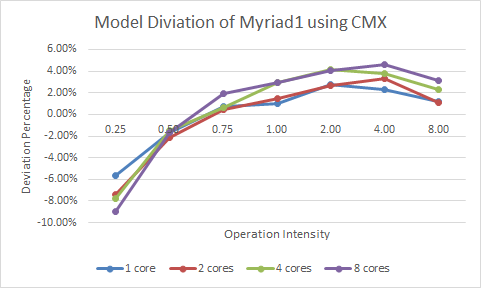}}
\caption{ Relative error of the model for intensity micro-benchmarks using CMX}
\label{fig:CMXError}
\end{figure}

Intensity-based micro-benchmarks using DDR memory to store data has the experimental results as shown in Figure~\ref{fig:DDRError}. Since eight SHAVE cores share the same data bus to DDR memory, contention power might happen when cores wait for memory fetching. The power model of Myriad1 using DDR memory is the Equation \ref{eq:PDDR}.

\begin{equation} 
\label{eq:PDDR} 
\begin{split}
	P&= P^{sta} + n \times (P^{act}+P^{dyn}_{SHV}) \times (\frac{I}{\alpha+I}) \\
	&+ (max(m,n) \times (P^{act}+P_{LSU}) + max(n-m, 0) \times P^{ctn}) \times (\frac{\alpha}{\alpha+I})
\end{split}
\end{equation}

In the DDR power model in Equation \ref{eq:PDDR}, there are unknown parameters such as $\alpha$, $m$ and $P^{ctn}$. By using non-linear regression techniques, $\alpha$ and $P^{ctn}$ are identified for each intensity value. 
Figure \ref{fig:MathlabDDRError} shows how well the predicted data fits to measured data with the found parameter values.
We plot the model deviation error of the model for intensity micro-benchmarks using DDR memory in Figure \ref{fig:DDRError}. The relative error of intensity micro-benchmarks varies from -16\% to 14\%. The model has high accuracy at intensity lower than 1 and higher deviation with micro-benchmarks running two cores. This requires more investigation. However, for each intensity, the deviation is less than 24\% margin.

\begin{figure}[!t] \centering
\resizebox{0.6\columnwidth}{!}{ \includegraphics{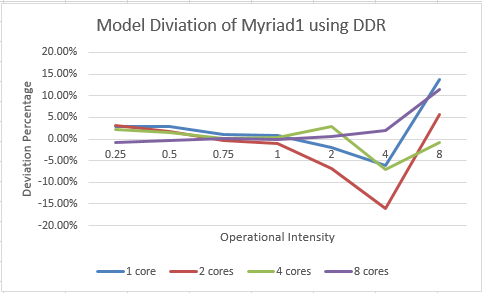}}
\caption{ Relative error of intensity micro-benchmarks using DDR}
\label{fig:DDRError}
\end{figure}

The found tunable parameters such as $\alpha$, $m$ and $P^{ctn}$ are summarized in the Table~\ref{table:TunableParameters}. The parameter values are derived from experimental results on micro-benchmarks using Matlab function $lsqcurvefit$.
\begin{table}[h]
\begin{center}
\begin{tabular}{|l|l|l|l|}
\hline
\textbf{Intensity} 	& \textbf{$P^{ctn}$} & \textbf{$\alpha$} & \textbf{$m$}	\\\hline
0.25               	& 0.1 &0.91 &1         \\ \hline
0.5                	& 0.1 &1.72 &1         \\ \hline
0.75               	& 0.1 &2.49	&1         \\ \hline
1               	& 1.97 &3.87	&1         \\ \hline
2               	& 5.25 &10	&1         \\ \hline
4               	& 0.1 &10	&1         \\ \hline
8               	& 0.1 &10	&2         \\ \hline
\end{tabular}
\caption{Values of Model Parameters derived from micro-benchmarks}
\label{table:TunableParameters}
\end{center}
\end{table}

\begin{figure}[!t] \centering
\resizebox{0.6\columnwidth}{!}{ \includegraphics{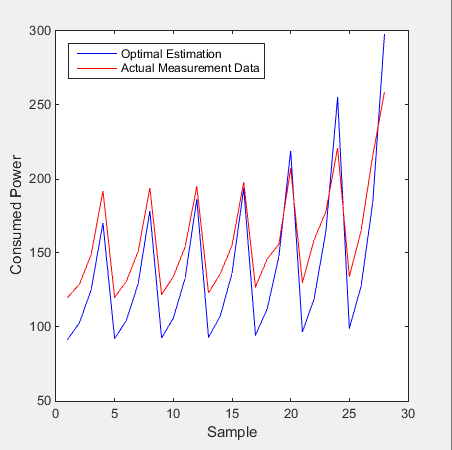}}
\caption{ Parameter fitting over the sample of intensity micro-benchmarks on DDR}
\label{fig:MathlabDDRError}
\end{figure}

\paragraph{Experiment with application benchmarks}
We want to analyze application benchmarks that represent for typical applications. Therefore, we follow two metrics to choose the applications for our investigation.
Since the energy consumption of an application depends on its operation intensity I, it is one main factor to choose application benchmarks. The energy consumption not only depends on the power but also program execution time. Therefore, the performance speed-up of an application is another main factor we consider to choose application benchmarks. 
The three benchmarks we choose are Matrix Multiplication, Sparse Matrix Vector Multiplication and Breadth First Search included in Berkeley Dwarfs ~\cite{Asa06} as shown in Figure ~\ref{fig:ApplicationCategories}. Matrix Multiplication is proved to have high operation intensity and high performance speed-up ~\cite{Ofenbeck:14} while Sparse Matrix Vector Multiplication has low intensity and  high speed-up due to its parallel scalability ~\cite{5348797}. Breadth First Search, on the other hand, has low intensity and saturated low scalability ~\cite{Cook:2013:HEC:2485922.2485949}.  Since applications with high intensity and low speed-up are typical sequential applications, which are not our target applications, so we do not include them in our validation.

\begin{figure}[!t] \centering
\resizebox{0.6\columnwidth}{!}{ \includegraphics{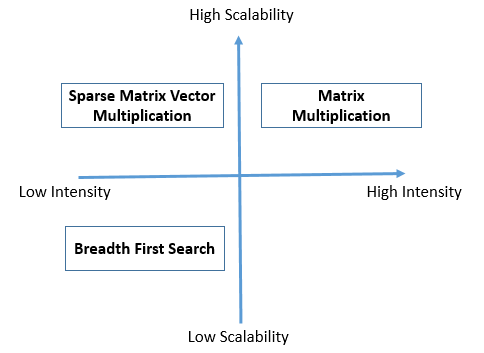}}
\caption{Application Categories}
\label{fig:ApplicationCategories}
\end{figure}

\subparagraph{Matrix multiplication}
Matrix multiplication (matmul) has been implemented based on the optimized assembly code for matmul using small blocks of data. Each block is the small submatrix of size $4\times4$. The matmul algorithm computes matrix C based on two input matrices A and B. All three matrices in this benchmarks are stored in DDR RAM. The operation intensity I of the matmul algorithm increase linearly with the matrix size.
We store each matrix element with float type which means four bytes. The number of operations and accessed data are calculated based on matrix size $n$ as: $W=2\times n^3$ and $Q=16\times n^2$. Intensity of matmul is also varied with matrix size as: $I = \frac{W}{Q} = \frac{n}{8}$  ~\cite{Ofenbeck:14}. In our experiments, we measured power consumption of matmul with varied matrix size. Figure ~\ref{fig:Matmul-8cores} is the predicted power consumption of matmul over intensity values using our power model. As shown in the diagram, power consumed by computation is the main contribution to the total power. With increased intensity, total power and computation power increase but power from accessing data decreases.

Apply the model by using the model parameters derived from micro-benchmarks, the deviation error of $matmul$ estimated power compared to measured data on 8 cores is from 24\% to 59\% as shown in Figure ~\ref{fig:MatmulDeviation}. The highest deviation occurs at size $16 \times 16$ because it is performed in short execution time and the multimeter device can not capture the real consumed power. This also happens for $matmul$ with matrix size less than $16 \times 16$. The deviation occurs because parameters have not considered the data-access patterns to access the same data or different data. Therefore, the model accuracy can be improved when we consider the data-access patterns and find tuning parameters such as $\alpha$, $P_{ctn}$ and $m$ based on actual measurements of the application.

\begin{figure}[!t] \centering
\resizebox{0.6\columnwidth}{!}{ \includegraphics{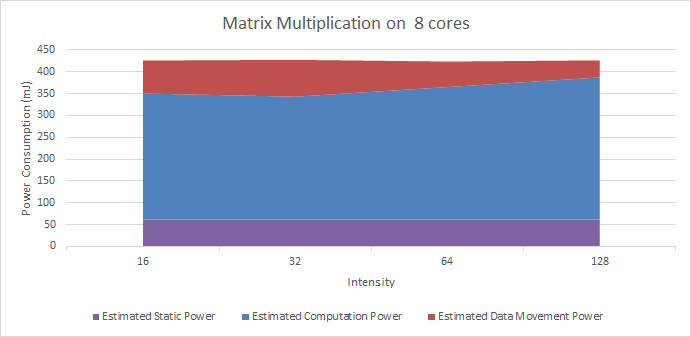}}
\caption{Power Analysis of Matrix Multiplication with 8 cores. Since the intensity of $matmul$ is also varied with matrix size as: $I = \frac{W}{Q} = \frac{n}{8}$, its consumed power is also varied by the intensity. The red and blue stacked lines are the estimated power consumed by data movement and computation which contribute to the total power respectively.}
\label{fig:Matmul-8cores}
\end{figure}

\begin{figure}[!t] \centering
\resizebox{0.6\columnwidth}{!}{ \includegraphics{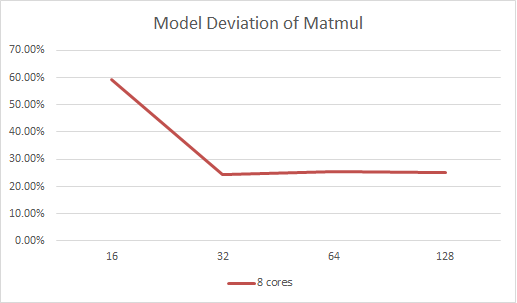}}
\caption{Deviation of estimated power from measured power for Matrix Multiplication}
\label{fig:MatmulDeviation}
\end{figure}

\subparagraph{Sparse matrix vector multiplication}
Sparse matrix vector multiplication (SpMV) is another kernel benchmark that we implement on Myriad1. All input matrix and vector of this benchmarks reside in DDR RAM. The operation intensity $I$ of SpMV algorithm does not depend on matrix size. The data layout of matrix in this implementation is compressed sparse row (CSR) format.
Each element of matrix and vector is also stored with float type. From our implementation, the number of operations and accessed data are calculated based on the size of one matrix dimension $n$ as: $W=10\times n$ and $Q=13 \times 4 \times n$. The intensity of SpMV does not depend on matrix size and is a fixed value: $I = \frac{W}{Q} = \frac{5}{26} = 0.19$. As the result, predicted power consumption of SpMV with varied matrix sizes is the same over matrix size as shown in Figure ~\ref{fig:SpMV-8cores}. 

Each element of matrix and vector is also stored with float type. From our implementation, the number of operations and accessed data is calculated based on the matrix size $n\times n$ as: $W=10\times n$ and $Q=13 \times 4 \times n$. Intensity of $SpMV$ is not dependent on matrix size and is a fixed value: $I = \frac{W}{Q} = \frac{5}{26} = 0.19$. Since the intensity value is small, power consumed by accessing data contributes significantly to the total power. The deviation error of $SpMV$ estimated power compared to measured power is from 15\% to 23\% which is an 8\% margin shown in Figure ~\ref{fig:SpMVDeviation}. Since we want to predict the trend of power consumption, an 8\% margin is acceptable for us to identify when to use the race-to-hold (RTH) condition \cite{D1.3}.

\begin{figure}[!t] \centering
\resizebox{0.6\columnwidth}{!}{ \includegraphics{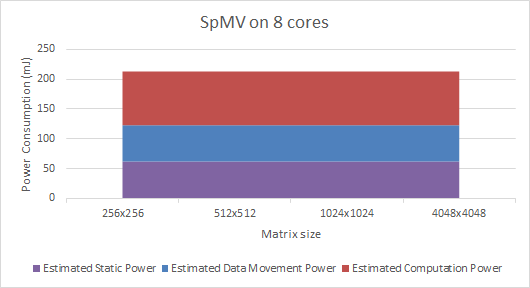}}
\caption{ Power analysis of Sparse Matrix Vector Multiplication}
\label{fig:SpMV-8cores}
\end{figure}

\remove { 
\begin{figure}[!t] \centering
\resizebox{0.6\columnwidth}{!}{ \includegraphics{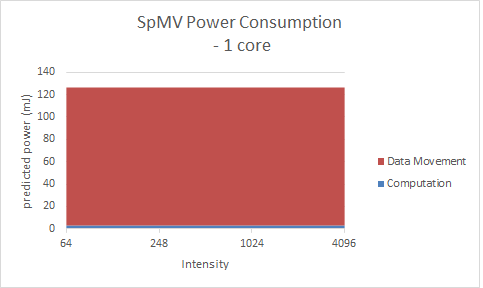}}
\caption{Power Analysis of Sparse Matrix Vector Multiplication with 8 cores. The red and blue stacked lines are the estimated power consumed by data movement and computation which contribute to the total power respectively.}
\label{fig:SpMV-1core}
\end{figure}
} 

\begin{figure}[!t] \centering
\resizebox{0.6\columnwidth}{!}{ \includegraphics{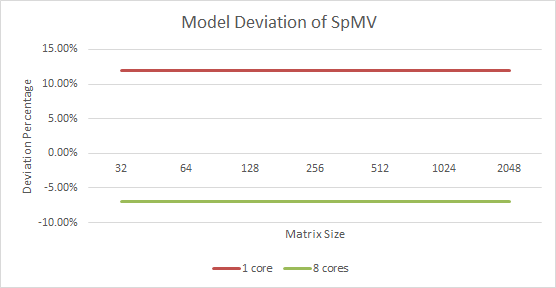}}
\caption{Deviation of estimated power from measured power of Sparse Matrix Vector Multiplication}
\label{fig:SpMVDeviation}
\end{figure}

\subparagraph{Breadth First Search.}
We also implement the Graph500 kernel, namely Breadth First Search (BFS), on Myriad1. BFS is the graph kernel to explore the vertices and edges of a directed graph from a starting vertex. We use the algorithm in the current Graph500 benchmark and port it to Myriad1. The graph is stored in DDR RAM. There are two properties that define the size of a graph: degree and scale. Scale identifies the number of nodes in the graph while degree is the average number of edges from a node. In our experiments, we mostly use the default scale of 16 from the Graph500 ~\cite{Beamer:2012:DBS:2388996.2389013}, which means the graph has $2^{16}$ vertices. The degree is varied from 14 to 17. The operation intensity $I$ of BFS algorithm does not depend on degree or scale. 
Each vertex data is stored with float type. From our implementation, the number of operations and accessed data are calculated based on the number of vertices $m$ and the number of edges $n$: $W=2 \times m + 5 \times n$ and $Q=8 \times m+ 16 \times n$. Intensity of BFS is a fixed value: $I = \frac{W}{Q} = 0.257$. As the result, predicted power consumption of BFS with varied degrees is the same over matrix size as shown in Figure ~\ref{fig:BFS-8cores}. 
Power consumed by accessing data contributes significantly to the total power. Figure ~\ref{fig:BFSDeviation} is the deviation error of the power model for BFS compared to measured data, which ranges from -14\% to 17\%. 

\begin{figure}[!t] \centering
\resizebox{0.6\columnwidth}{!}{ \includegraphics{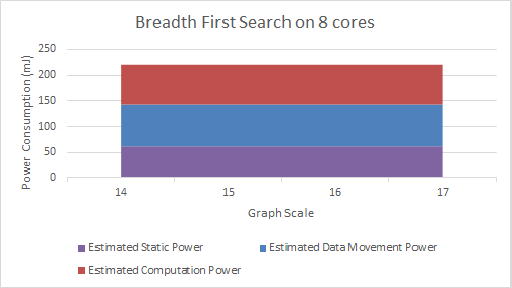}}
\caption{ Power analysis of Breadth First Search}
\label{fig:BFS-8cores}
\end{figure}

\remove{ 
\begin{figure}[!t] \centering
\resizebox{0.6\columnwidth}{!}{ \includegraphics{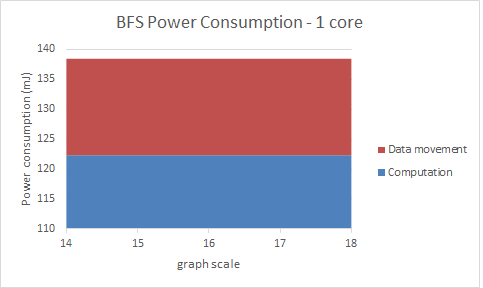}}
\caption{Power Analysis of Breadth First Search with 8 cores. The red and blue stacked lines are the estimated power consumed by data movement and computation which contributes to the total power respectively.}
\label{fig:BFS-1core}
\end{figure}
} 

\begin{figure}[!t] \centering
\resizebox{0.6\columnwidth}{!}{ \includegraphics{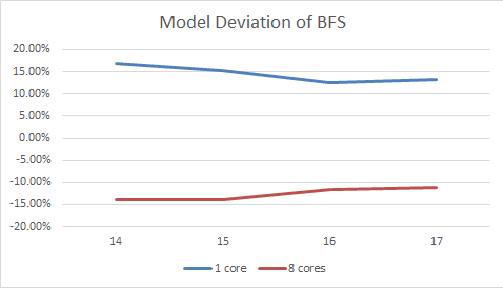}}
\caption{Deviation of estimated power from measured power of Breadth First Search}
\label{fig:BFSDeviation}
\end{figure}

\paragraph{Remarks and future works}

We have presented a power model for Myriad1 platform. In details, we have described the process to build the model and evaluate it. We have first validated the model with three sets of micro-benchmarks. We have also implemented three application kernels on Myriad1 platform, along with analyzing the amount of work $W$, the number of data accesses $Q$ and the intensity $I$ for each algorithm. Then we have applied our model to the implemented algorithms and evaluated its accuracy. 

In the future, we plan to use the model with EXCESS execution framework described in Deliverable D3.1 and D3.2. In that context, the model can use the feedback data of consumed power measured by the framework to find the accurate value of $\alpha$ and  $P^{ctn}$ and $m$ for a given platform. We also plan to improve the accuracy of the model by considering memory-access patterns of the implementation and instruction-pipeline parallelism. 



\subsection{Energy Model for Lock-Free Queues on CPU Platform}
We consider the problem of modeling the energy behavior of
lock-free concurrent queue data structures, more especially,
of lock-free queue implementations and parallel applications that use
them. Focusing on steady state behavior we decompose energy behavior
into throughput and power dissipation which can be modeled separately
and later recombined into several useful metrics, such as energy per
operation. Based on our models, instantiated from synthetic benchmark
data, and using only a small amount of additional application specific
information, energy and throughput predictions can be made for
parallel applications that use the respective data structure
implementation.

To model throughput we propose a generic model for lock-free queue
throughput behavior, based on a combination of the dequeuers' throughput
and enqueuers' throughput.
To model power dissipation we commonly split the contributions from the various
computer components into static, activation and dynamic parts, where
only the dynamic part depends on the actual instructions being
executed.
To instantiate the models a synthetic benchmark explores each queue
implementation over the dimensions of processor frequency and number
of threads.

Finally, we show how to make predictions of application throughput and
power dissipation for a parallel application using a lock-free queue
requiring only a limited amount of information about the
application work done between queue operations. Our case study on a
Mandelbrot application shows convincing prediction results.

\label{sec:cpu-model}

In this Section~\ref{sec:cpu-model}, we present a so-called white-box
model for performance and power of lock-free queue
implementations. This model is white-box, in the sense that the
phenomena that drive the evolution of the throughput and the power
dissipation are clearly identified, and their contribution to both
metrics is transparently stated. The Section~\ref{sec:cpu-model-inst}
will then explain how to set the parameters of the model with only a
few runs of the synthetic benchmark.

\subsubsection{Motivation and Preliminaries}

Lock-free implementations of data structures are a scalable approach for
designing concurrent data structures. Lock-free data structures offer
high concurrency and immunity to deadlocks and convoying, in
contrast to their blocking counterparts.
Concurrent FIFO queue data structures are fundamental data structures
that are key components in applications, algorithms, run-time and
operating systems. The producer/consumer pattern, \eg, is a common
approach to parallelizing applications where threads act as either
producers or consumers and synchronize and stream data items between
them using a shared collection.
A concurrent queue, \aka shared ``first-in, first-out'' or FIFO
buffer, is a shared collection of elements which supports at least the
basic operations \op{Enqueue} (adds an element) and \op{Dequeue} (removes
the oldest element). \op{Dequeue} returns the element removed or, if the
queue is empty, \nulle.
A large number of lock-free (and wait-free) queue implementations have
appeared in the literature,
\eg~
\cite{Val94,lf-queue-michael,TsiZ01b,MoirNSS:2005:elim-queue,%
DBLP:conf/opodis/HoffmanSS07,Gidenstam10:OPODIS} being some of the
most influential or most efficient results.
Each implementation of a lock-free queue has obviously its strong and weak
points so the impact on performance and energy when choosing one
particular implementation for any given situation may not be obvious.

As the number of known implementations of lock-free concurrent
queues is growing, it is of great interest to describe
a framework within which the different implementations can be ranked,
according to the parameters that characterize the situation.
A brute force approach could achieve this by running the
implementations on hand on the whole domain of study, gathering and comparing
measurements. This would yield high accuracy, but at a tremendous
cost, since the domain is likely to be large. Additionally, it would
only bring a limited understanding on the phenomena that drive the
behavior of the queue implementations.
Therefore, we propose generic models for predicting the behavior of
lock-free queues under steady state usage. The models are
instantiated for the queue implementations and machine on hand using
empirical data from a limited number of points in the domain.

The implementations can be ranked according to a plethora of metrics.
Traditionally, performance in terms of throughput has been the main metric.
Furthermore, the notion of energy efficiency has now extended into
every nook and cranny of Information Technology, at any scale, from
the Exascale machines that need huge improvements in terms of power
dissipation to be feasible~\cite{IESProadmap}, to the small electronic
devices where the battery lifetime is a critical issue.

We decompose the energy behavior of queues, and subsequently
applications, into two components: (i)~throughput and (ii)~power
dissipation. We model these components separately.
The predicted throughput and power dissipation can be recombined into
the energy-efficiency metric energy per queue operation, which is the
ratio between power dissipation and queue throughput. When modeling
an application, this metric can be extended to energy per unit of
application work. Further, plotting energy per operation or unit of
work according to throughput allows exploration of the Pareto-optimal
frontier of the energy$-$performance bi-criteria optimization problem
for the queues or the application.

Lock-free queue data structures generally offer disjoint-access parallelism: enqueuers and
dequeuers modify only their respective ends of the queue, and compete mostly with operations of
the same kind. Nonetheless, when the queue is close to empty, both ends point to the same
part of the queue, then enqueue and dequeue operations have to be synchronized, and every
operation impacts the behavior of any other.

Concerning the queue as a whole, a successful event can be seen as the dequeue of a
non-\nulle item, since this event implies that the item has been enqueued and
dequeued. Also, the throughput of the queue is naturally defined as the number of such
events per unit of time, which is a meaningful performance criterion for queues.

In this work, we focus on queues that are in a steady state, \ie such that the rate of each
operation attempt is constant. Then, the throughput \thr of the queue is the minimum
between the throughput of all dequeues \thrd, even those returning \nulle,
and throughput of enqueues \thre. Indeed, if $\thre > \thrd $, then the
queue grows and the throughput is determined by the dequeuers, which cannot obtain any
\nulle items; and if $\thre \leq \thrd $, then the queue is mostly empty and \nulle items
are dequeued, but the throughput is determined by the enqueuers.

Despite this decomposition, enqueuers' and dequeuers' throughput are still correlated when
the queue is mostly empty. In addition, the interactions between them are rather
asymmetric, as in broad terms, an enqueue can be delayed by any concurrent dequeue, while
for a dequeue, concurrent enqueues will cease to disturb it as they move away from the
dequeue end. 

Based on these facts, we decorrelate the throughput into several uncorrelated and basic
throughputs, and reconstitute the main throughput by combining them. Among the advantages
of this process, we earn a better understanding of the performance (as the basic
throughputs are meaningful), and we reduce the number of measurements needed to
instantiate the model on the whole domain of study.

The domain of study that we envision here can be viewed as the Cartesian product of four
sets: (i) number of threads accessing the queue, (ii) CPU frequencies, (iii) a range of dequeue access rates, (iv) a range of enqueue access rates.
The cardinality of the first two sets is at most a few tens, while the last two are
continuous sets that are not even bounded. Thanks to the removal of the
dependencies between throughputs, we are able to instantiate the model with only a few
data points, while the model covers the whole intervals.

Finally, this decomposition also eases the study of power dissipation, where we reuse
the same ideas as in the throughput estimation part.



\subsubsection{Framework}
\label{sec:pb-stat}

\paragraph{Synthetic Benchmark}

\subparagraph{Skeleton}





\begin{figure}[h!]
\begin{minipage}{.45\textwidth}
\begin{algorithmic}
\Procedure{Enqueuer}{}
\While{!done}
	\State \texttt{Parallel\_Work}()\;
	\State \texttt{Enqueue}()\;
\EndWhile
\EndProcedure
\end{algorithmic}
\end{minipage}\hfill%
\begin{minipage}{.45\textwidth}
\begin{algorithmic}
\Procedure{Dequeuer}{}
\While{!done}
	\State $\mathit{res} \leftarrow \texttt{Dequeue()}$\;
        \If{$\mathit{res} \neq \textsc{Null}$}
        	\State \texttt{Parallel\_Work}()\;
        \EndIf
\EndWhile
\EndProcedure
\end{algorithmic}
\end{minipage}%
\caption{Queue benchmark\label{alg:gen-ed}}
\end{figure}

\newcommand{\paus}{{\it pause}\xspace}


We run the synthetic benchmark composed of the two functions described in
Figure~\ref{alg:gen-ed}, starting with an empty queue. Half of the threads are assigned to be enqueuers while the
remaining ones are dequeuers. We disable logical cores (hyper-threading) and map different threads into
different cores, also the number of threads never exceeds the number of cores. In addition,
the mapping is done in the following way: when adding an enqueuer/dequeuer pair, they are
both mapped on the most filled but non-full socket.

The \pss (\FuncSty{Parallel\_Work}) shall be seen as a processing activity, pre-processing
for the enqueuers before they enqueue an item, and post-processing on an item from the
queue for the dequeuers. We assume that memory accesses in the \pss are negligible, and
represent the \pss as sequences of bunches of \paus{} instructions in the benchmark; we
note \pwe (resp. \pwd) the number of bunches of $90$ \paus{}s (which corresponds to $1000$
cycles) that compose the parallel work in the enqueuer (resp. dequeuer).

From a high-level perspective, \FuncSty{Enqueue} and \FuncSty{Dequeue} operations follow a
\rl pattern: a thread reads an access point to the \ds, works locally with this view of
the \ds, possibly performs memory management actions and prepares the new desired value as
an access point of the \ds. Finally, it atomically tries to perform the change through a call
to the \cas primitive. If it succeeds, \ie if the access point has not been changed by
another thread between the first read and the \cas, then it goes to the next \ps, otherwise
it repeats the process.

\subparagraph{Queue Implementations}



We study some of the most well-known and studied lock-free and linearizable
queues in the literature, as implemented in NOBLE~\cite{Sundell08}.
These queue algorithms are described in some detail in Section~\ref{sec:NOBLE}.
The aim of this work is still to predict the behavior of any lock-free queue algorithm and not
only the ones mentioned above. These algorithms are used to validate the model that we
present in the following sections.


When we speak about implementations of the queues, we actually refer to the
different implementations of enqueuing and dequeuing operations, along with
their memory management schemes.

\paragraph{General Power Model}
\label{sec:pow-mod}

The power is split into three elements: the {\it static} part is the cost of turning
the machine on, the {\it activation} part incorporates a fixed cost for each socket and each
core in use, and the {\it dynamic} part is a supplementary cost that depends on the
running application.

In accordance with the RAPL energy counters~\cite{DavidGoHaKhLe:2010:RAPL,%
  BrDoGaHoMu:2000:PAPI,Weaver:2012:MEP:2410139.2410475}, we further decompose each part
per-component, for memory, CPU, and {\it uncore} (denoted by a superscript M, C and U,
respectively):
\[ \pow{} = \sum_{X \in \{M,C,U\}} \left( \pstat{X} + \pact{X} + \pdyn{X}  \right). \]

We assume that we already know the platform characteristics, \ie all
static and active powers (they can be obtained as explained for
instance in D1.2~\cite{EXCESS:D1.2}), and we try to find the
application-specific dynamic powers.  In order to keep the formulas
readable, in the following, we denote by \pow{X} the dynamic power
\pdyn{X}.

\paragraph{Notations and Setting}

We denote by \nth the number of running threads that call the same
operation, and by \freq the clock frequency of the cores (we only consider the case where
all cores share the same clock frequency).

We recall that \pwe (resp. \pwd) is the amount of work in the \ps of an enqueuer
(resp. dequeuer), as the number of bunches of $90$ \paus{}s. For a given queue
implementation, we denote by \cwe (resp. \cwd) the amount of work in one try of the \rl of
the \FuncSty{Enqueue} (resp. \FuncSty{Dequeue}) operation. Associated with these amounts of
work, we define, for $\bx \in \{ \bd, \be \}$, the average execution time of the \ps (resp. the \rl
and a single try of the \rl) related to operation $\bx$ as \et{\psx} (resp. \et{\rlx} and
\et{\slx}).

In the same way, for $\bx \in \{ \bd, \be \}$, we denote by
\powx{C} (resp. \powxps{C} and \powxrl{C}) the dynamic CPU power dissipated by component $X$ in
(resp. the \ps related to and the \rl related to) operation $\bx$.

Finally, for $\bx \in \{ \bd, \be \}$, we denote by \ratx the ratio of the time that a thread spends in
the \rl, while it is associated with operation $\bx$.

In Sections~\ref{sec:thput} and~\ref{sec:power}, in order to keep expressions as simple as
possible, we define one unit of time as \second{\lambda}, where $\lambda$ is the execution
time of $90 \times f$ \paus{}s (as the \paus{} instructions are perfectly scalable with clock
frequency, $\lambda$ is constant). Throughput is expressed in number of operations per
unit of time, \ie per \seconds{\lambda}. Finally, we derive the power in Watts.

\newcommand{\terb}[1]{\ema{#1\,\text{TB}}}
\newcommand{\gigb}[1]{\ema{#1\,\text{GB}}}
\newcommand{\megb}[1]{\ema{#1\,\text{MB}}}
\newcommand{\kilb}[1]{\ema{#1\,\text{kB}}}
\newcommand{\megtps}[1]{\ema{#1\,\text{MTransfers}/\text{sec}}}

All experiments and their underlying predictions are done on \syscha (see
Section~\ref{sec:chalmers-system}), \ie a platform composed of a dual-socket
Intel\textsuperscript{\textregistered} Xeon\textsuperscript{\textregistered}
processor, with eight cores per socket.  The sizes of L3, L2 and L1 caches are
\megb{25}, \kilb{256} and \kilb{32}, respectively.

We run the implementations at the two extreme frequencies (excluding
Turbo mode) \ghz{1.2} and
\ghz{3.4}, for all possible even total numbers of threads, from 2 to 16, \ie for
$\nth \in \{1,\dots,8\}$.


\subsubsection{Throughput Estimation}
\label{sec:thput}

\paragraph{Throughput Decomposition Principles}
\label{sec:gen-thput}

\newcommand{\ercon}{inter-contention\xspace}
\newcommand{\racon}{intra-contention\xspace}

We recall that the throughput of the queue is defined as:
\[ \thr = \min \left( \thre, \thrd \right),\]
where \thre and \thrd are the enqueuers' and dequeuers' throughput, respectively.

As we are in steady state, one operation $o$ is performed every $\et{\psx} + \et{\rlx}$
unit of time by each thread, and \nth threads attempt to concurrently execute $o$, hence the general expression of
the throughput \thrx:
\[ \thrx = \frac{\nth}{\et{\psx} + \et{\rlx}}. \]

We have seen that the \pss of the benchmark are full of \paus{}s, thus the time \et{\psx} spent in a given \ps
is straightforwardly given by  $\et{\psx} = \pwx/\freq$.
The execution time of dequeue and enqueue operations is more problematic, for two main
reasons. {\it Primo}, because of the lock-free nature of the implementations. As the number of
retries is unknown, the time spent in the function call is not trivially
computable.
{\it Secundo}, when the activity on the queue is high, the threads compete for accessing a
shared data, and they stall before actually being able to access the data. We name this as
the \textit{expansion}, as it leads to an increase in the execution time of a single try
of the \rl.


The contention on the queue is twofold. At any time, and even if it could be negligible,
threads that perform the same operation disturb each other, since they try to access the
same shared data.
In addition, when the queue is mostly empty, enqueuers and dequeuers try to access the
same data, then interference occurs; enqueuers make dequeuers stall and {\it vice
  versa}.
We call the former case {\it \racon}, and the latter one {\it \ercon}.

As expected, we have noticed a marked difference between the execution time of a
dequeue operation returning \nulle and one that returns a queue item, \ie whether the
queue was empty or contained at least one item.
That is why we decompose \thrd into throughput of dequeue on empty queue \thrde
(that returns a \nulle item), and dequeue on non-empty queue \thrdne (that does not
return \nulle).

Further, the impact of \ercon on dequeue operations is negligible compared to the
impact of the queue being empty; therefore we ignore \ercon for dequeues.

In contrast, the queue being empty does not notably change the execution time of the
enqueue operation, while dequeue operations can impact the behavior of concurrent
enqueue operations greatly when the queue is close to empty. Hence, we split \thre into the
enqueue throughput \thrend when the queue is not inter-contended, and the enqueue
throughput \thred when the queue experiences the maximum possible \ercon.

These basic throughputs fulfill the two following inequalities: $\thrde \geq \thrdne$ and
$\thrend \geq \thred$.

Thanks to this separation into the four basic throughput cases \thrde, \thrdne, \thrend
and \thred, we earn a better understanding of the factors that influence the general
throughput, and we deinterlace their dependencies, which dramatically decreases the number
of points in the \ps sizes set where we need to take measurements for our modeling.
%
More precisely, by construction, \thrde and \thrdne do not indeed depend on \pwe, while
\thrend and \thred do not depend on \pwd. Nonetheless \thrd (resp. \thre) is defined as a
barycenter between \thrde and \thrdne (resp. \thred and \thrend), whose weights depend on
both \pwd and \pwe.

In Section~\ref{sec:bas-thput}, we describe the basic throughputs, we combine them in
Section~\ref{sec:comb-thputs}, then we explain how to instantiate the parameters of the
model in Section~\ref{sec:instantiation}, and finally exhibit results in
Section~\ref{sec:thput-res}.

\paragraph{Basic Throughputs}
\label{sec:bas-thput}

We aim in this section at estimating the throughput \thrxy of one of the basic operations
described in the previous subsection, where $\bx \in \{ \be, \bd\}$ and $\bb \in \{ \ub,
\lb\}$. We assume that \thrxy depends only on \pwx, in addition to the tacit dependencies
on the clock frequency, number of threads and queue implementation.
We denote by \cwxy the amount of work in a single try of the \rl related to operation
$\bx$ in case $\bb$ when the queue is not intra-contended.

\subparagraph{Low Intra-Contention}
\newcommand{\accns}{\ema{a}}
\newcommand{\accfs}{\ema{a'}}
\newcommand{\bccfs}{\ema{b'}}

We study in this section the low \racon case, \ie when (i) the threads do not suffer from
expansion due to threads that perform the same operation, and (ii) a success is obtained
with a single try of the \rl. As it appears in Figure~\ref{fig.sch-lc}, we
have a cyclic execution, and the length of the shortest cycle is $\et{\psx} +
\et{\slxy}$. Within each cycle, every thread performs exactly one successful operation,
thus the throughput is straightforward:
\begin{equation}
\thrxy = \frac{\nth}{\et{\psx} + \et{\slxy}} = \frac{\nth \freq}{\pwx + \cwxy}. \label{eq.thr-lc}
\end{equation}



\newcommand{\extth}[3]{%
\checkendxt(#1)
\edef\prev{\cachedata}
\pgfmathparse{\prev+#2}
\endxt(#1)={\pgfmathresult}
\checkendxt(#1)
\edef\nene{\cachedata}
\draw[#3] (\prev,\yt[#1]) -- (\nene,\yt[#1]);
}

\def\yt{{0.5,4.5,3,1.5}}
\def\morea{.8}
\def\moreb{.5}
\def\cle{2.4}
\def\ple{7.6}
\def\prevend{0}

\begin{figure}[t!b]
\begin{center}
\begin{tikzpicture} [scale=0.6, font=\small, thick, par/.style={blue,|-|}, csu/.style={green,|-|},%
    cfa/.style={red,|-|}, ini/.style={blue,-|}, finc/.style={green,|-}, finp/.style={blue,|-}]
\newarray\endxt
\expandarrayelementtrue
\readarray{endxt}{0&0&0&0}
\newarray\wpl

\draw[<->] (\cle,\yt[0]) -- ++ (\cle+\ple,0)  node[midway,fill=white] {Cycle};

\node [above, font=\small, text width=20, align=center] at (.5*\cle,\yt[1]) {Retry Loop};
\node [above, font=\small] at (\cle+.5*\ple,\yt[1]) {Parallel Work};

\extth{1}{\cle}{csu}

\extth{2}{3.8}{ini}
\extth{2}{\cle}{csu}

\extth{3}{7}{ini}
\extth{3}{\cle}{csu}

\extth{1}{\ple}{par}
\extth{1}{\cle}{csu}

\extth{2}{\ple}{par}
\extth{2}{1}{finc}

\extth{3}{5.4}{finp}

\extth{1}{2.4}{finp}

\draw[dotted] (\cle,\yt[0]-\moreb) -- ++ (0,\yt[1]-\yt[0]+\morea+\moreb);
\draw[dotted] (\cle+\ple+\cle,\yt[0]-\moreb) -- ++ (0,\yt[1]-\yt[0]+\morea+\moreb);

\end{tikzpicture}
\end{center}
\caption{Cyclic execution under low \racon\label{fig.sch-lc}}
\end{figure}







\subparagraph{High Intra-Contention}
\label{sec:high-intra}


As explained in Section~\ref{sec:gen-thput}, in this case, the direct evaluation of
the execution time of a \rl is more complex, but we have experimentally observed 
that the throughput is approximately linear with the expected number of
threads that are in the \rl at a given time. In addition, this expected number is almost
proportional to the amount of work in the \ps. As a result, a good approximation of the
throughput, in high \racon cases, is a function that is linear with the amount of work in
the \pwx.

\subparagraph{Frontier}

\def\cle{2.4}
\def\ple{4.8}
\def\prevend{0}

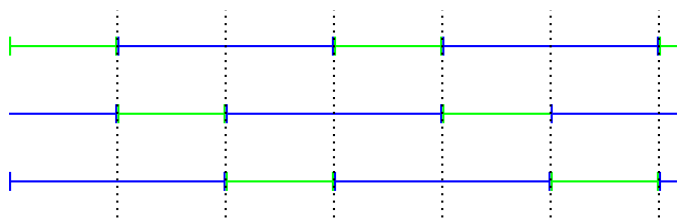
\begin{figure}[b!h]
\begin{center}
\begin{tikzpicture} [scale=0.6, font=\small, thick, par/.style={blue,|-|}, csu/.style={green,|-|},%
    cfa/.style={red,|-|}, ini/.style={blue,-|}, finc/.style={green,|-}, finp/.style={blue,|-}]
\newarray\endxt
\expandarrayelementtrue
\readarray{endxt}{0&0&0&0}
\newarray\wpl


\extth{1}{\cle}{csu}

\extth{2}{\cle}{ini}
\extth{2}{\cle}{csu}

\extth{3}{\ple}{par}
\extth{3}{\cle}{csu}

\extth{1}{\ple}{par}
\extth{1}{\cle}{csu}

\extth{2}{\ple}{par}
\extth{2}{\cle}{finc}

\extth{3}{\ple}{par}
\extth{3}{\cle}{csu}

\extth{1}{\ple}{par}
\extth{1}{.4}{finc}

\extth{2}{2.8}{finp}

\extth{3}{.4}{finp}

\draw[dotted]   (\cle,\yt[3]-\morea) -- ++ (0,\yt[1]-\yt[3]+\morea+\morea);
\draw[dotted] (2*\cle,\yt[3]-\morea) -- ++ (0,\yt[1]-\yt[3]+\morea+\morea);
\draw[dotted] (3*\cle,\yt[3]-\morea) -- ++ (0,\yt[1]-\yt[3]+\morea+\morea);
\draw[dotted] (4*\cle,\yt[3]-\morea) -- ++ (0,\yt[1]-\yt[3]+\morea+\morea);
\draw[dotted] (5*\cle,\yt[3]-\morea) -- ++ (0,\yt[1]-\yt[3]+\morea+\morea);
\draw[dotted] (6*\cle,\yt[3]-\morea) -- ++ (0,\yt[1]-\yt[3]+\morea+\morea);

\end{tikzpicture}
\end{center}
\caption{Intra-contention frontier\label{fig.sch-crit}}
\end{figure}

We now have to estimate whether the queue is highly intra-contended.
We recall that, generally speaking, a long parallel section leads to a low
intra-contended queue since threads are most of the time processing some
computations and are not trying to access the shared data. Reversely, when the
\ps is short, the ratio of time that threads spend in the \rl is higher, and
gets even higher because of both expansion and retries.

That being said, there exists a simple lower bound of the amount of work in the
\ps, such that there exists an execution where the threads are never failing in
their \rl.
We plot in Figure~\ref{fig.sch-crit} an ideal
execution with $\nth = 3$ threads and $\et{\psx}=(\nth-1) \times \et{\slxy}$. In this
execution, all threads always succeed at their first try in the \rl. Nevertheless, if we
shorten the \ps, then there is not enough parallel potential any more, and the
threads will start to fail: the queue leaves the low intra-contention state.





In practice, this lower bound ($\et{\psx}=(\nth-1) \times \et{\slxy}$) is actually a good
approximation for the critical point where the queue switches its state.
%

\paragraph{Combining Basic Throughputs}
\label{sec:comb-thputs}

We are given \pss sizes, and show how to link the throughput of the four basic operations,
with the dequeuers' and enqueuers' throughput. There are two possible states for the
queue: either it is mostly empty (\ie some \nulle items are dequeued), or it gets
larger and larger.

In the first case, some of the dequeues will occur on
an empty queue. In $1$ unit of time, \thre items are enqueued. These items are
dequeued in $\thre/\thrdne$ units of time (the queue is non-empty while they are dequeued),
which leads to a slack of $1-\thre/\thrdne$, where dequeues of \nulle items can take
place at a rate \thrde, hence the following throughput formula:
\begin{equation}
\thrd = \frac{\thre}{\thrdne} \times \thrdne +
        \left( 1 -\frac{\thre}{\thrdne} \right) \times \thrde. \label{eq:deq-comp}
\end{equation}

Concerning the enqueuers, we use the same assumption on \ercon as used on \racon in
Section~\ref{sec:high-intra}, saying that the throughput is linear with the expected number
of threads inside the \rl. Here, the expected number of threads inside the dequeue
operation is proportional to the ratio \ratd of the time spent by one dequeuer in its
dequeue operation. We do not know \et{\rld}, but we know that in average, to complete a
successful operation, a thread needs $\et{\psd}+\et{\rld}$ units of time, and among this
time it will spend \et{\psd} in the \ps. Therefore 
\begin{equation*}
\ratd = 1-\et{\psd}/(\et{\psd}+\et{\rld}) = 1 - \frac{\thrd \times \pwd}{\nth \times \freq}.
\end{equation*}
The minimum intra-contention is reached when this
ratio is $0$, while the maximum is obtained when it is $1$, thus:
\begin{equation}
\thre = \frac{\thrd \times \pwd}{\nth \times \freq} \times \thrend +
           \left( 1 - \frac{\thrd \times \pwd}{\nth \times \freq} \right) \times \thred.
        \label{eq:enq-comp}
\end{equation}

In the second case, enqueuers and dequeuers do not access to the same part of the queue,
thus \ercon does not take place, then $\thre=\thrend$, and all dequeues return a
non-\nulle item, hence $\thrd=\thrdne$.

\medskip

The discrimination of these two cases is trivial when enqueuers' and dequeuers' throughput
are given: the queue is in the first state (mostly empty) if and only if $\thre \leq
\thrd$.

Reversely, if we know the four basic throughputs
$\left( \thrend, \thred, \thrde, \thrdne\right)$, and aim at
reconstituting the dequeuers' and enqueuers' throughput
$(\thrd,\thre)$, several solutions could be consistent, \ie either a
growing queue, such that $\thre=\thrend$ and $\thrd=\thrdne$ and
fulfilling the inequality $\thre > \thrd$, or a mostly empty queue,
fulfilling Equations~\ref{eq:deq-comp} and~\ref{eq:enq-comp}.

\begin{theorem}
Given $\left( \thrend, \thred, \thrde, \thrdne\right)$, there exists a
consistent solution $(\thrd,\thre)$ with a growing queue if and only if
$\thrend>\thrdne$. In addition, this solution is unique and is such
that $\thre = \thrend$ and $\thrd = \thrdne$.
\label{th:grow}
\end{theorem}
\begin{proof}
\noindent$(\Rightarrow)$ If the queue is growing, then $\thre>\thrd$. Moreover, dequeues never occur on an empty
queue, hence $\thrd = \thrdne$, and there is no \ercon, thus $\thre = \thrend$.

\noindent$(\Leftarrow)$ Let us assume now that $\thrend>\thrdne$. $\thre = \thrend$ and
$\thrd = \thrdne$ is a valid solution, such that the queue is growing, since then
$\thre>\thrd$.

By construction, $\thre \leq \thrend$; if we had another solution such that the queue
grows and $\thre<\thrend$, it would mean that enqueues are inter-contended, which is
possible only when the queue is mostly empty. This is absurd, hence the uniqueness.
\end{proof}

\begin{theorem}
Given $\left( \thrend, \thred, \thrde, \thrdne\right)$, there exists a
consistent solution $(\thrd,\thre)$ with a mostly empty queue if and
only if
\begin{equation}
\frac{\thred}{\thrdne} \leq 1 - \frac{\pwd}{\nth \times \freq}
\left( \thrend - \thred \right).
\label{eq:in-comp}
\end{equation}
In addition, this solution is unique and is given by Equations~\ref{eq:enq-comp}
and~\ref{eq:deq-comp}.
\label{th:empty}
\end{theorem}
\begin{proof}
\noindent$(\Rightarrow)$ Let a solution $(\thrd,\thre)$ with a mostly empty queue. By construction, the
throughputs follow Equations~\ref{eq:enq-comp} and~\ref{eq:deq-comp}. As \thre is an
increasing function according to \thrd (because $\thrend\geq\thred$), we derive
\[ \thre \geq \frac{\thrdne \times \pwd}{\nth \times \freq} \times \thrend +
           \left( 1 - \frac{\thrdne \times \pwd}{\nth \times \freq} \right) \times \thred.\]

The queue is mostly empty, thus the dequeues of non-\nulle items have to be faster than
the enqueues, which translates into $\thrdne \geq \thre$. The two inequalities combined
show the implication.

\noindent$(\Leftarrow)$ Let us assume now that Inequality~\ref{eq:in-comp} is fulfilled.
Equation~\ref{eq:deq-comp} can be rewritten into
\[ \thre = \frac{\thrd-\thrde}{1-\frac{\thrde}{\thrdne}}.\]
Let us consider now ${\thre}'$ and ${\thre}''$ two functions of ${\thrd}'$ that fulfill the
following system of equations:
\[
\left\{ \begin{array}{l}
{\thre}'\left( {\thrd}' \right) = \frac{{\thrd}'-\thrde}{1-\frac{\thrde}{\thrdne}}\\
{\thre}'' \left( {\thrd}' \right) = \frac{{\thrd}' \times \pwd}{\nth \times \freq} \times \thrend +
           \left( 1 - \frac{{\thrd}' \times \pwd}{\nth \times \freq} \right) \times \thred.
\end{array} \right.
\]
We have ${\thre}'\left( \thrde \right) = 0$ and ${\thre}'\left( \thrdne \right) =
\thrdne$. According to Inequality~\ref{eq:in-comp}, we know also that ${\thre}''\left(
\thrde \right) \leq \thrde$. In addition, ${\thre}''$ is a linearly increasing function of
${\thrd}'$ and ${\thre}'$ a linearly decreasing function of ${\thrd}'$. This shows that
there exists a unique \thrd such that ${\thre}'\left( \thrd \right) = {\thre}''\left(
\thrd \right)$, and if we define \thre as $\thre = {\thre}'\left( \thrd \right) = {\thre}''\left(
\thrd \right)$, the pair $\left(\thrd,\thre\right)$ is such that
\[
\left\{ \begin{array}{l}
\thrdne \leq \thrd \leq \thrde\\
\thred \leq \thre \leq \thrend\\
\thre \leq \thrd
\end{array} \right. .
\]
This implies that $\left(\thrd,\thre\right)$ is a solution with an empty queue, and we have shown that this
solution is unique.
\end{proof}


\begin{corollary}
Given $\left( \thrend, \thred, \thrde, \thrdne\right)$, there exists at least one solution
$(\thrd,\thre)$.
\end{corollary}
\begin{proof}
We show that if the inequality of Theorem~\ref{th:grow} is not fulfilled, \ie if
$\thrend \leq \thrdne$, then the inequality of Theorem~\ref{th:empty} is true.
We have indeed
\begin{align*}
\left( 1 - \frac{\pwd}{\nth \times \freq}
\left( \thrend - \thred \right)\right) \thrdne  - \thred &=
 \left( 1 - \frac{\pwd \times \thrend}{\nth \times \freq}\right) \thrdne
- \left( 1 -  \frac{\pwd \times \thrdne}{\nth \times \freq}\right) \thred \\
&\geq \left( 1 - \frac{\pwd \times \thrend}{\nth \times \freq}\right) \thrdne
- \left( 1 -  \frac{\pwd \times \thrdne}{\nth \times \freq}\right) \thrend \\
&\geq \thrdne - \thrend \\
 \left( 1 - \frac{\pwd}{\nth \times \freq}
\left( \thrend - \thred \right)\right) \thrdne - \thred &\geq 0,
\end{align*}
which proves the Corollary.
\end{proof}

One can notice that if $\thrend > \thrdne$ and Inequality~\ref{eq:in-comp} are fulfilled
and the queue could be either mostly empty or growing. In this case, we choose, for each
operation, the mean of the two solutions, in order to minimize the discontinuities.


\subsubsection{Power Estimation}
\label{sec:power}

We recall that we are interested only in the dynamic powers as we assume that static and
activation powers are known.

\paragraph{CPU Power}

Firstly, as we map each thread on a dedicated core, there is no interference between the CPU power
of different cores, so we can compute the dynamic power as
\begin{equation}
\pow{C} = \nth \times \powe{C} + \nth \times \powd{C}.
\label{eq:dec-ed}
\end{equation}

Secondly, we assume that we can segment time and consider that, given a thread performing
operation $\bx \in \{ \be, \bd\}$, the power dissipated in the \rl and the power dissipated in the \ps are
independent. There only remains to weight the previous powers by the time spent in each of
these regions:
\begin{equation}
\powx{C} =  \ratx \times \powxrl{C} + (1 - \ratx) \times \powxps{C}.
\label{eq:dec-psrl}
\end{equation}

As shown in Section~\ref{sec:comb-thputs}, the ratio can be obtained through
\begin{equation}
\ratx = 1 - \frac{\thrx \times \pwx}{\nth \times \freq}.
\label{eq:ratx}
\end{equation}

Altogether, we obtain the final formula for dynamic CPU power
\begin{equation}
\pow{C} = \nth \left( \sum_{\bx \in \{ \be,\bd \}} \powxrl{C} +  
\frac{\thrx \times \pwx \times \left( \powxps{C} - \powxrl{C}  \right)}{\nth \times \freq}     \right)
\label{eq:cpu-pow}
\end{equation}

\paragraph{Memory and Uncore Power}


We have noticed in~\cite{EXCESS:D1.2} that the dynamic memory power is proportional to
the intensity (number of units of memory accessed per unit of time) of main memory
accesses and remote accesses, when the threads read separate places of the memory.

Here, the \ds does not directly involve the main memory since we keep its size reasonably
bounded (if the queue reaches the maximum size, we suspend the measurements, empty the
queue, and resume), hence the power dissipation in memory is only due to remote accesses,
which only appears as the threads are spread across sockets (\ie when $\nth>4$).

Moreover, as the \pss are full of \paus{}s, communications can only take place in the \rl,
and there is no dynamic memory power dissipated in the \pss.
Concerning the \rls, we make the following assumption: the amount of data
accessed per second in a \rl depends on the implementation, but given an implementation,
once a thread is in the \rl, it will always try to access the same amount of data per
second. When the queue is highly intra-contended, if a thread fails then it will retry and
will access the data in the same way as in the previous try; and if there is expansion, then
the thread will still try to access the data for the whole time it is in the \rl.

In addition, the dequeuers (and the same line of reasoning holds for
the enqueuers) tries here to access the same data. Therefore either
memory requests are batched together when sent outside the socket, or
the Home Agent (cache coherency mechanism/memory controller) keeps
track of the previous requests. This implies that the number of
threads attempting to access the data does not impact the dynamic
memory power greatly when the rate of requests is high.

All things considered, as a thread working on operation $\bx$ spends a fraction \ratx of its time
inside its \rl, we obtain that the dynamic memory power dissipated in the \rl is
proportional to $\ratx$ (times the amount of data accessed per unit of time in the \rl,
which is a constant). Hence
\begin{equation}
\pow{M} = \rate \times \codae{M} + \ratd \times \codad{M},
\label{eq:pow-mem}
\end{equation}
where \codae{M} and \codad{M} are constants.


%




The dynamic uncore power is computed exactly in the same way as the dynamic memory power.

\subsection{White-box Methodology for Instantiating the Energy Model of Queues on CPU Platform} 
\label{sec:cpu-model-inst}

In this section, we show how to obtain the parameters of the model, so
that predictions can further be made. Those parameters depend on the
architecture where the application is running, thus we need some
measurements on a few runs of the synthetic benchmark to discover
them. This methodology is white-box, since we know which runs to use,
so that the calibration of the model is achieved with a minimum
possible set of runs.

\subsubsection{Instantiating the Throughput Model}
\label{sec:instantiation}

We recall that, for all $\bx \in \{ \be, \bd\}$ and $\bb  \in \{ \ub,
\lb\}$, \thrxy depends only on \pwx, while \thre and \thrd depend on both \pwd and
\pwe. We denote now by $\thrd(\pwd,\pwe)$ (resp. $\thre(\pwd,\pwe)$) the dequeuers'
(resp. enqueuers') throughput as the amount of work in the \ps of the dequeuers is \pwd and
enqueuers' one is \pwe. The estimate of a value is denoted by a hat on top, while
the measured value does not wear the hat.

Let $\psma=1$, $\pmid=20$ and $\pbig=1000$ be three distinctive amounts of work, that
corresponds to different states of the execution. If $\pwx=\pbig$, we can neglect the
impact of operation $\bx$ on the queue, $\pwx=\pmid$ is a low \racon case since the
non-expanded critical sections are experimentally less than $2$ units of time, and
$\pwx=\psma$ corresponds to a highly inter- or \racon case. We note the we cannot use a
$0$ size as amount of work since it leads to undesirable results due to the back-to-back
effect (a thread does not allow other threads to access the queue for several consecutive
iterations).

\paragraph{Low Intra-Contention}


The basic throughputs that are not intra-contended can be spawned
from \cwxy (critical section size of operation $\bx$ in case $\bb$),
where $\bx \in \{ \be, \bd\}$ and $\bb \in \{ \ub,
\lb\}$, which we try
to estimate here.
We pick four points where the basic throughputs are easy to approximate.  We have
$\thrd(\pmid,\psma) < \thre(\pmid,\psma)$, as the amounts of
work in the \rls are in practice less than $10$. For the same reason, at this point, we are in low \racon
from the dequeuers' point of view. Altogether,
\[ \thrd(\pmid,\psma) = \thrdne(\pmid) = \frac{\nth \times \freq}{\pmid + \cwdne}, \text{ hence} \]
\begin{equation*}
\wh{\cwdne} = \frac{\nth \times \freq}{\thrd(\pmid,\psma)}- \pmid.
\end{equation*}

Then, according to Equation~\ref{eq:deq-comp}, we have
\begin{align*}
\frac{\nth \freq}{\pmid + \wh{\cwde}} &= \thrde(\pmid)\\
\frac{\nth \freq}{\pmid + \wh{\cwde}}
&= \frac{\thrd(\pmid,\pbig) - \thre(\pmid,\pbig)}%
{1- \frac{\left(\pmid + \wh{\cwdne}\right) \times \thre(\pmid,\pbig)}{\nth \times \freq}},\\
\end{align*}
from which we can extract \wh{\cwde} since we know already \wh{\cwdne}.

In the same way, we can compute \wh{\cwend} then \wh{\cwed}, by using $(\pbig,\pmid)$ and
$(\psma,\pmid)$.




\paragraph{High Intra-Contention}

We aim here at estimating \thrxy on a high \racon point. $\psma=1$ and $\pmid=20$ are such
that $\thrd(\psma,\pmid) \geq \thre(\psma,\pmid)$. According to Equation~\ref{eq:deq-comp}, we have
\[ \thrd(\psma,\pmid) = \thre(\psma,\pmid) +
        \left( 1 -\frac{\thre(\psma,\pmid)}{\wh{\thrdne}(\psma)} \right) \times \wh{\thrde}(\psma).
\]
In addition, if $\thrd(\psma,\psma) \geq \thre(\psma,\psma)$, then
\[ \thrd(\psma,\psma) = \thre(\psma,\psma) +
        \left( 1 -\frac{\thre(\psma,\psma)}{\wh{\thrdne}(\psma)} \right) \times \wh{\thrde}(\psma),
\]
otherwise, $\thrd(\psma,\psma) = \wh{\thrdne}(\psma)$. In both cases, we can find the two unknowns
$\wh{\thrdne}(\psma)$ and $\wh{\thrde}(\psma)$ thanks to the two equations.

This last point is also used in the same way for enqueuers: if $\thrd(\psma,\psma) \geq
\thre(\psma,\psma)$, then
\begin{equation*}
\thre(\psma,\psma) = \frac{\thrd(\psma,\psma) \times \psma}{\nth \times \freq} \times \wh{\thrend}(\psma)
+ \left( 1 - \frac{\thrd(\psma,\psma) \times \psma}{\nth \times \freq} \right) \times \wh{\thred}(\psma),
\end{equation*}

otherwise, $\thre(\psma,\psma) = \wh{\thrend}(\psma)$.

Like previously, we have $\thrd(\pmid,\psma) < \thre(\pmid,\psma)$, hence
$\wh{\thrend}(\psma)=\thre(\pmid,\psma)$. This implies that in any cases we can compute
$\wh{\thrend}(\psma)$, but we do not have access to $\wh{\thred}(\psma)$ if
$\thrd(\psma,\psma) < \thre(\psma,\psma)$. In this case, the bottleneck of the queue is
likely to be the dequeuers, hence we set the value $\wh{\thred}(\psma) =
\wh{\thrend}(\psma)$ by default.

All \wh{\thrxy} are then obtained by joining $\wh{\thrxy}(\psma)$ to the leftmost point
of the low \racon part:

\begin{equation*}
\wh{\thrxy}(\pwx) = \left\{ \begin{array}{ll}
\frac{\frac{f}{\wh{\cwxy}}-\wh{\thrxy}(\psma)}{(\nth - 1 ) \wh{\cwxy} - \psma} \times (\pwx - \psma) + \wh{\thrxy}(\psma) &\quad \text{if }\pwx \leq (\nth - 1 ) \wh{\cwxy}\\
\vspace*{.05cm}&\\
\frac{\nth \times \freq}{\pwx +\wh{\cwxy}} &\quad \text{otherwise}.
\end{array}\right.
\end{equation*}

Finally, dequeuers' and enqueuers' throughput are reconstituted as explained in
Section~\ref{sec:comb-thputs}: if Equation~\ref{eq:in-comp} is fullfilled, then they are
computed through Equations~\ref{eq:deq-comp} and~\ref{eq:enq-comp} that can be rewritten as:
\begin{equation*}
\left\{ \begin{array}{l}
\wh{\thrd}(\pwd,\pwe) = \frac{\wh{\thrde}(\pwd) + \wh{\thred}(\pwe) \left( 1- \frac{\wh{\thrde}(\pwd)}{\wh{\thrdne}(\pwd)}\right)}%
{1-\frac{\pwd}{\nth \freq}\left( \wh{\thrend}(\pwe) - \wh{\thred}(\pwe) \right) \left( 1- \frac{\wh{\thrde}(\pwd)}{\wh{\thrdne}(\pwd)} \right)}\\
\vspace*{.05cm}\\
\wh{\thre}(\pwd,\pwe) = \frac{\wh{\thrd}(\pwd,\pwe) \times \pwd}{\nth \times \freq} \times \wh{\thrend}(\pwe) 
+\left( 1 - \frac{\wh{\thrd}(\pwd,\pwe) \times \pwd}{\nth \times \freq} \right) \times \wh{\thred}(\pwe).\\
\end{array}\right.
\end{equation*}

Otherwise, $\wh{\thrd}(\pwd,\pwe) = \wh{\thrdne}(\pwd)$ and $\wh{\thre}(\pwd,\pwe) =
\wh{\thrend}(\pwe)$.


\subsubsection{Instantiating the Power Model}

We use once again $\psma=1$, $\pmid=20$ and $\pbig=1000$ as three distinctive amounts of
work, that allows easy approximations for the power dissipation expressions.

We have seen that if $X \in \{ M,U \}$, then $\pow{X} = \ratd \times \codad{X} + \rate
\times \codae{X}$, which can be approximated at $(\pwd,\pwe) = (\pbig,\psma)$ by
$\pow{X}(\pbig,\psma) = \rate(\psma) \times \codae{X}$, since \ratd is then nearly $0$. It
implies that
\[ \wh{\codae{X}} = \frac{\pow{X}(\pbig,\psma)}{1 - \frac{\thre(\pbig,\psma) \times \psma}{\nth \times \freq}}. \]
We obtain \wh{\codad{X}} similarly at $(\pwd,\pwe) = (\psma,\pbig)$.

Concerning the dynamic CPU power, we firstly estimate the power dissipated in the \pss.
According to the implementation, the CPU power dissipated by the \ps of enqueuers and
dequeuers is the same for both, and this power does not depend on the amount of
work. These restrictions are not a loss of generality, since the aim here is to study the
queue implementations. It can then be estimated by using $(\pbig,\pbig)$,
where the ratios \ratx can be considered as $0$, which leads to
\[  \wh{\powxps{C}} = \frac{\pow{C}(\pbig,\pbig)}{ 2 \nth}. \]

We reuse the point $(\pbig,\psma)$, where \ratd is very close to $0$, to derive that
\[ \pow{C} = \nth \left( \rate(\psma) \times \wh{\powerl{C}} + (1 - \rate(\psma)) \wh{\poweps{C}} \right) + \nth \wh{\powdps{C}},
\]
which is equivalent to
\begin{equation*}
\wh{\powerl{C}} =\frac{\pow{C}(\pbig,\psma)}{\nth \left( 1 - \frac{\thre(\pbig,\psma) \psma}{\nth \times \freq} \right)}
 - \left( \frac{2}{1 - \frac{\thre(\pbig,\psma) \psma}{\nth \times \freq}} -1 \right) \wh{\powxps{C}}
\end{equation*}

Once again, we obtain \wh{\powdrl{C}} with the same line of reasoning at $(\pwd,\pwe) = (\psma,\pbig)$.

Finally, \wh{\pow{M}} and \wh{\pow{U}} (resp. \wh{\pow{C}}) are computed by using
Equation~\ref{eq:pow-mem} (resp. Equations~\ref{eq:dec-ed} and~\ref{eq:dec-psrl}), and
the estimates of the ratios that are issued from Section~\ref{sec:thput}
\[ \wh{\ratx} = 1 - \frac{\wh{\thrx} \times \pwx}{\nth \times \freq}. \]

\subsubsection{Summary}

To summarize, we have built a model that needs to be calibrated
(instantiated) before the execution of an application that would use
the queue.  The calibration phase starts with the run of synthetic
benchmarks on the following set of points and the measurements of the
dequeuers' and enqueuers' throughputs:
\begin{align*}
(\pwd,\pwe) \in \big\{ &(\pmid,\psma),
(\pmid,\pbig), (\psma,\pmid), (\pbig,\pmid),\\
& (\psma,\psma),
(\pbig,\psma), (\pbig,\pbig), (\psma,\pbig) \big\}.
\end{align*}

The calibration phase continues with the extraction of the parameters
that rule the four basic throughputs (hence the dequeuers' and
enqueuers' throughputs) on the whole domain.

After this calibration phase, we are able, given any parallel section
values $(\pwd,\pwe)$, to estimate the throughput and the energy
consumption of the application, through a few basic arithmetic
operations, as explained in the previous subsections.


\newpage
\section{Programming Abstractions and Libraries} \label{sec:libraries}
In this section, we describe our studies on programming abstractions and libraries such as concurrent search trees and concurrent lock-free queues.

\subsection{Concurrent Search Trees} \label{sec:search-trees}
In this section, we present libraries of concurrent search trees and their performance and energy analysis.

\subsubsection{Energy-efficient Concurrent Search Trees} \label{sec:delta-trees}

\paragraph{DeltaTree ($\Delta$Tree)} \label{sec:DeltaTree}
Concurrent trees are fundamental data structures that are widely used in different contexts such as load-balancing \cite{DellaS00, HaPT07, ShavitA96} and searching \cite{Afek:2012:CPC:2427873.2427875, BronsonCCO10, Brown:2011:NKS:2183536.2183551, Crain:2012:SBS:2145816.2145837, DiceSS2006, EllenFRB10}.
Most of the existing highly-concurrent search trees are not considering the fine-grained
data locality. The non-blocking concurrent search trees 
\cite{Brown:2011:NKS:2183536.2183551, EllenFRB10} and Software Transactional
Memory (STM) search trees
\cite{Afek:2012:CPC:2427873.2427875, BronsonCCO10,
Crain:2012:SBS:2145816.2145837, DiceSS2006} have been regarded
as the state-of-the-art concurrent search trees. They have been proven
to be scalable and highly-concurrent. 
However these trees are not designed for fine-grained data locality. 
Prominent concurrent search trees which are often included in several benchmark 
distributions such as the concurrent red-black
tree \cite{DiceSS2006} by Oracle Labs and the concurrent AVL tree
developed by Stanford \cite{BronsonCCO10} are not designed for data locality either. It is challenging to devise search trees that are portable, highly concurrent and fine-grained locality-aware. A platform-customized locality-aware search trees \cite{KimCSSNKLBD10, Sewall:2011aa} are not portable while there are big interests of concurrent data structures for unconventional platforms~\cite{Ha:2010aa, Ha:2012aa}. Concurrency control techniques such as transactional memory~\cite{Herlihy:1993aa,Ha:2009aa} and 
multi-word synchronization~\cite{Ha:2005aa,Ha:2003aa,Larsson:2004aa} do not take into account fine-grained locality while fine-grained locality-aware techniques such as van Emde Boas layout \cite{Prokop99,vanEmdeBoas:1975:POF:1382429.1382477} poorly support concurrency. 

Based on the new concurrency-aware vEB (cf. Section~\ref{sec:concurrentvEB}), we implement $\Delta$Tree \cite{deltatreeTR2013}, a portable locality-aware unbalanced concurrent search tree. Figure \ref{fig:treeuniverse} illustrates a $\Delta$Tree $U$ which is composed by a group of subtrees ($\Delta$Nodes).
A $\Delta$Node's internal nodes are put together in cache-oblivious
fashion using the concurrency-aware vEB layout (cf. Section \ref{sec:relaxed-veb}). The search operation for $\Delta$Tree
is wait-free. We are aware that the $\Delta$Tree has two major 
shortcomings, namely being an unbalanced tree and having a
poor memory utilization because of the inter-node pointers usage.

In order to address the $\Delta$Tree shortcomings, we implement the Balanced $\Delta$Tree (b$\Delta$Tree).
b$\Delta$Tree is devised by improving the structure and the 
algorithm of $\Delta$Tree to support the concurrent, Btree-like bottom-up insertions, which
ensures balanced tree.
Another major change is that in b$\Delta$Tree, the fat-pointer $\Delta$Nodes are 
replaced with pointer-less $\Delta$Nodes. The removal of
$\Delta$Nodes internal pointers made way for 200\% more nodes to fit into the tree (in 64-bit x86 systems, a pointer
costs 8 bytes of memory while an unsigned variable requires only 4 bytes of memory). The 
concurrent search operations supported by b$\Delta$Tree are still without locks and waits, but
they are no longer wait-free. Based on our study, a more-compact tree results in less data transfers.
And less data transfers leads to better performance and energy efficiency.

Finally, we implement the Heterogeneous $\Delta$Tree (h$\Delta$Tree), which is 
a better performing, more-compact tree than b$\Delta$Tree. h$\Delta$Tree is devised by changing
the leaf-oriented (external) tree layout of the leaf $\Delta$Nodes into an internal tree layout. 
This improvement allows 100\% more nodes to fit inside the tree, which further
improves the tree's overall operation performance.

Based on experimental insights, our $\Delta$Trees are different from previous theoretical designs of concurrent 
cache-oblivious (CO) trees such as the concurrent packed-memory CO tree and concurrent exponential CO tree \cite{BenderFGK05}. 
The concurrent packed-memory CO tree gives a good amortized memory transfer cost of $\Theta(\log_B N + (\log^2 N/B))$
for tree updates, assuming that operations occur {\em sequentially}. However, the proposed data representation 
requires each node to have the parent-child pointers. 
Besides the complication in re-arranging those pointers, we have found that eliminating pointers 
from the node to minimize memory footprint is significantly 
beneficial for cache-oblivious tree in practice (cf. improvement from $\Delta$Tree to b$\Delta$Tree in Section \ref{sec:bDeltaTree}). 
In the $\Delta$Tree experimental evaluation (cf. Section \ref{sec:evaluation}), b$\Delta$Tree
is 100\% faster than $\Delta$Tree in searching, which is attained by simply
removing pointers from the tree node.

In the concurrent exponential CO tree by Bender et al. \cite{BenderFGK05}, expected memory transfer cost
for search and update operations is $\mathcal{O}(\log_B N + (\log_\alpha \text{lg} N))$, assuming that all processors are {\em synchronous}. Cormen et al. \cite[pp. 212]{CormenSRL01}, however,  
wrote that although the underlying exponential tree algorithm \cite{548472} is an important theoretical
breakthrough, it is complicated and unlikely to compete with similar sorting algorithms. 
In fact, nodes in the exponential tree grow exponentially in size, which not only complicates maintaining inter-node pointers but also exponentially increases the tree's memory footprint in practice. 
In contrast, the memory footprint of $\Delta$Tree with the fixed size $\Delta$Nodes gradually expands on-demand when the tree grows. Thanks to the fixed size $\Delta$Nodes, $\Delta$Tree exploits further locality by utilizing a "map" and an efficient inter-node connection (cf. Section \ref{sec:interdesc}). 

\begin{figure}[!t] \centering \includegraphics[width=0.8\textwidth]{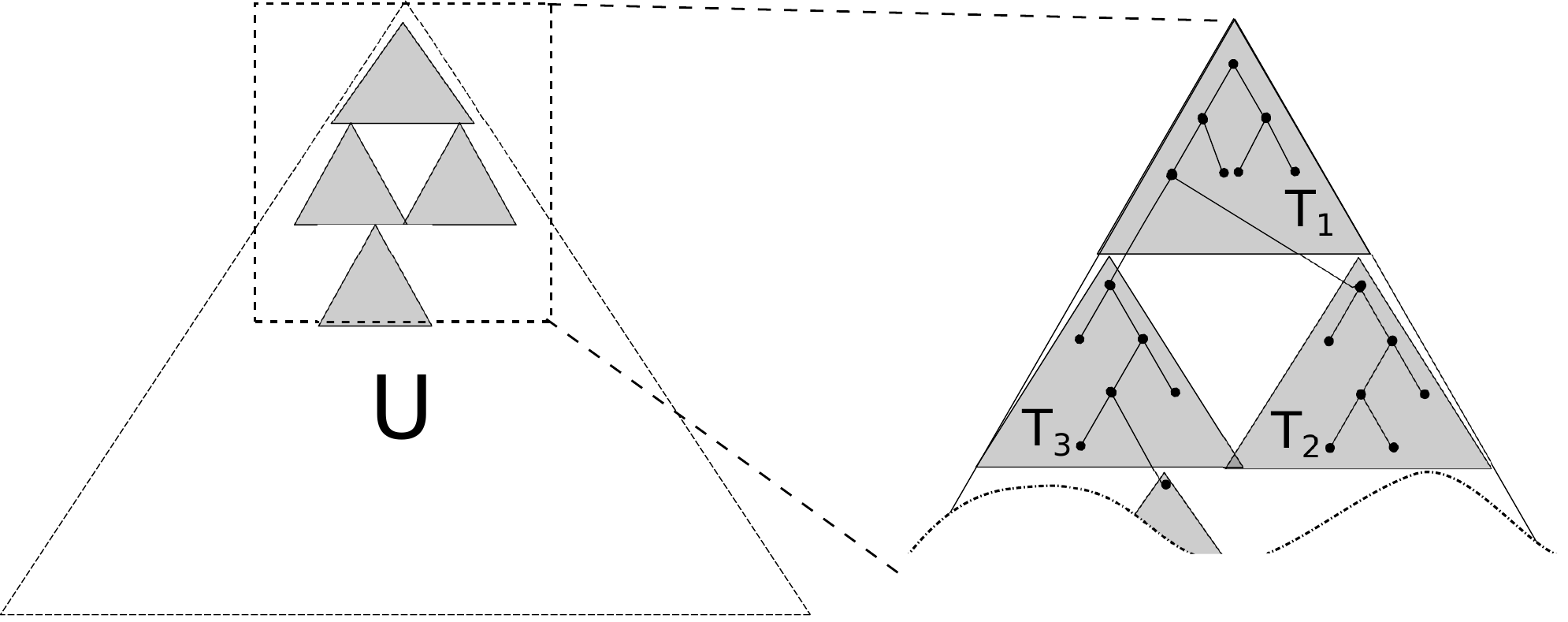}
\caption{Depiction of a $\Delta$Tree $U$. Triangles $T_x$ represent the $\Delta$Nodes.}
\label{fig:treeuniverse}
\end{figure}

\begin{figure}[!t] 
\centering 
\begin{algorithmic}[1] 
\Start{\textsc{node} $n$} \label{lst:line:nodestruct}
	\START{member fields:}
\State $tid \in \mathbb{N}$, if $> 0$ indicates the node is \textit{root} of a $\Delta$Node with an id of $tid$ ($T_{tid}$) 
\State $value \in
	\mathbb{N}$, the node value, default is \textbf{empty} 
\State $mark \in \{true,false\}$, a value of \textbf{true} indicates a logically deleted node
\State $left, right \in \mathbb{N}$, left and right child pointers 
\State $isleaf \in {true,false}$, indicates whether the node is a leaf of a $\Delta$Node,  \par 
	\hskip\algorithmicindent default is \textbf{true} 	\label{lst:line:leafdefault}
\END \End

\Statex
\Start{\textsc{$\Delta$Node} $T$} \label{lst:line:deltanodestruct}
\START{member fields:}
\State $\mathit{nodes}$, a group of pre-allocated \textsc{node} $n$ $\{n_1,n_2,\ldots,n_{\UB}\}$
\State $\mathit{buffer}$, $\Delta$Node's buffer (pre-allocated array of $b_1,b_2,\ldots,b_{\mathit{\#threads}}$)
\State $\mathit{meta}$, $\Delta$Node's metadata (single \textsc{$\Delta$NodeMeta})

\END \End

\Statex
\Start{\textsc{$\Delta$NodeMeta} $M$} \label{lst:line:trianglestruct}
\START{member fields:}
\State $\mathit{locked}$, indicates whether a $\Delta$Node is locked 
\State $\mathit{opcount}$, a counter for the active update operations 
\State $\mathit{root}$, pointer to the root node of the $\Delta$Node ($T_x.n_1$)
\State $\mathit{mirror}$, pointer to the root node of the $\Delta$Node's mirror ($T_{x'}.n_1$)
\END \End

\Statex
\Start{\textsc{universe} $U$} 		\label{lst:line:universe}
	\START{member fields:}
\State $root$, pointer to the $root$ of the topmost $\Delta$Node ($T_1.root$) 
\END \End
\end{algorithmic}
\caption{$\Delta$Tree's data structures.} \label{lst:datastruct}
\end{figure}

The $\Delta$Tree consists of $|U|$ $\Delta$Nodes of fixed size $\UB$. Each of
the $\Delta$Node contains a \textit{leaf-oriented} binary search tree (BST) $T_i, i=1,
\dots,|U|$. The $\Delta$Tree $U$ provides the following operations: \textsc{insertNode($v, U$)}, 
which adds value $v$ to the set $U$, \textsc{deleteNode($v, U$)} 
for removing a value $v$ from the set, and
\textsc{searchNode($v, U$)}, for determining whether value $v$ exists in the
set. We use the term \textit{update} operation for either insert or delete
operation. We assume that duplicate values are not allowed inside the set and a
special value, for example $0$, is reserved as an indicator of an \textsc{Empty}
value.

\subparagraph{Data structures.}

The implementation of $\Delta$Tree utilizes the data structure in
Figure \ref{lst:datastruct}. The topmost level of $\Delta$Tree is represented by
a struct \textsc{universe} (line \ref{lst:line:universe}) that contains 
a pointer (\textbf{root}) 
to the root node of the first $\Delta$Node ($T_1.n_1$).
A $\Delta$Tree is formed by a group of $\Delta$Nodes. 

$\Delta$Tree's $\Delta$Nodes are represented by the struct
\textsc{$\Delta$Node} (line \ref{lst:line:deltanodestruct}). A $\Delta$Node consists
of a collection of $\UB$ \textbf{nodes} and a \textbf{buffer} array. 
$\Delta$Node's \textbf{buffer} array length is equal 
to the number of operating threads. Each $\Delta$Node is accompanied by
a metadata \textsc{$\Delta$NodeMeta} that holds lock and counters variable.

Struct \textsc{$\Delta$NodeMeta} (line \ref{lst:line:trianglestruct})
acts as metadata for every $\Delta$Node.
This structure consists of a field \textbf{opcount}, which is a counter that indicates 
the number of insert/delete threads that
are currently operating within that $\Delta$Node; and field \textbf{locked} that indicates 
whether a $\Delta$Node is currently locked by a maintenance operation.
When \textbf{locked} is set as \textit{true}, no insert/delete threads 
are allowed to get into a $\Delta$Node. 
Lastly, \textsc{$\Delta$NodeMeta} contains a \textbf{root} pointer 
that points to the first  \textsc{node}
of $\Delta$NodeMeta's accompanying $\Delta$Node, and the pointer \textbf{mirror} that points 
to the \textit{root} of the $\Delta$Node's mirror (cf. Section \ref{sec:mirroring}).

Each \textsc{node} structure (line \ref{lst:line:nodestruct}) contains field \textbf{value}, 
which holds a value for guiding the search, or a data value if it resides in a leaf-node.
Field \textbf{mark} indicates a logically deleted value, if set to \textit{true}. 
A \textit{true} value of \textbf{isleaf} indicates the node is a leaf
node, and \textit{false} otherwise. Field \textbf{tid} is a unique
identifier of a corresponding $\Delta$Node and it is used to let a thread know whether itself 
has moved
between $\Delta$Nodes. A \textbf{tid} is only defined in the root node of a $\Delta$Node.

\subparagraph{$\Delta$Tree functions.}

\begin{figure}[!t] 
\centering \begin{algorithmic}[1] 

\Function{searchNode}{$v, U$}
\State $p \gets U.root$
\While{(TRUE)} 	\label{lst:line:search-while}
	\State $lastnode \gets p$	\Comment{atomic assignment} \label{lst:line:lastnode-p}
        \If{($p.value < v$)}		\label{lst:line:searchless}
            \State $p \gets p.left$
        \Else					\label{lst:line:searchelse}
            \State $p \gets p.right$ 
        \EndIf
        \If{($!p$ \textbf{or} $lastnode.isleaf =$ TRUE)}	\label{lst:line:searchifleaf}
        		\State \textbf{break} 					 \label{lst:line:search-end}
	\EndIf
\EndWhile
\If{($lastnode.value = v$)} 		\label{lst:line:linsearch3}
	\If{($lastnode.mark =$ FALSE)}	 \Comment{lastnode is not deleted}	\label{lst:line:linsearch1}
		\State\Return TRUE 
	\Else 
		\State\Return FALSE 
	\EndIf 
\Else \State \{Search the last visited $\Delta$Node's \textit{buffer} for $v$\}	\label{lst:line:searchbuffer}
	\If{(\{$v$ is found\})}							\label{lst:line:linsearch2}
		\State\Return TRUE 
	\Else 
		\State\Return FALSE 
	\EndIf 
\EndIf 
\EndFunction
\end{algorithmic}
\caption{$\Delta$Tree's \textit{wait-free} search algorithm.}\label{lst:nodeSearch}
\end{figure}

$\Delta$Tree provides basic functions such as search, insert and delete functions. 
We refer insert and delete operations
as the \textit{update} operations.

Function \textsc{searchNode($v,U$)} (cf. Figure \ref{lst:nodeSearch}), is going to walk 
over the $\Delta$Tree to
find whether the value $v$ exists in $U$ (Figure \ref{lst:nodeSearch}, lines \ref{lst:line:lastnode-p}--\ref{lst:line:search-end}). 
The  \textsc{searchNode($v,U$)} function returns \textbf{true} whenever $v$ has been
found, or \textbf{false} otherwise (Figure \ref{lst:nodeSearch}, line \ref{lst:line:linsearch1}). 
This operation is
guaranteed to be wait-free (cf. Lemma \ref{lem:waitfreesearch}).

\begin{lemma} \label{lem:waitfreesearch}
$\Delta$Tree search operation is \textit{wait-free}.
\end{lemma}
\begin{proof}(Sketch) The proof can be served based on these observations on Figure \ref{lst:nodeSearch}:
\begin{enumerate}
\item \textsc{SearchNode} and invoked \textsc{SearchBuffer} (line \ref{lst:line:searchbuffer}) do not wait for any locks.
\item The number of iterations in the \textit{while} loop (line \ref{lst:line:search-while})
	 is bounded by the \textit{height}  $T$ of the tree, or $\mathcal{O}(T)$. The loop is always ended
	 whenever a leaf node or an empty node is found in line \ref{lst:line:search-end}.
\item \textsc{SearchBuffer} time complexity is bounded by the buffer size, which is a constant.
\end{enumerate}
Therefore the  \textsc{SearchNode} time is bounded by $\mathcal{O}(T)$, where $T$ is the height of the tree.
\end{proof}

\begin{figure}

\centering
\footnotesize
\begin{algorithmic}[1]

\Function{insertNode}{$v, p$}	\label{lst2:line:insertfunc}
\State $p \gets U.root$

\State \textit{BEGIN}:

\If{(\{Entering new $\Delta$Node $T_x$\})}                 			                                                 
   \State  Decrement previous $\Delta$Node's opcount ($T'_x.opcount$)				
   \State  Wait if $\Delta$Node is currently maintained 		\label{lst2:line:waitinsert}
\EndIf
    
\If{Not at the tree's last level node}                                                    
        \If{($v < p.value$)}                                                                              
            \If{Is at leaf ($p.isleaf =$ TRUE)}  									                              
                \If{(CAS($p.left.value$, \textbf{empty}, $v$) = \textbf{empty})}       \label{lst2:line:ins-lp1} 	\label{lst2:line:growins1}    
                	   \State Do insert to the left   	                          \label{lst2:line:growins1-end}                         
                   \State Decrement $T_x.opcount$ 
                \Else
                    \State Re-try insert from $p$ 		\label{lst2:line:retryinsleft}                    
                \EndIf
            \Else
                \State Go to left child ($p \gets p.left$) ) and restart from \textit{BEGIN}     
            \EndIf
        \ElsIf {($v > p.value$)}                                                                      
            \If{(Is at leaf ($p.isleaf =$ TRUE)}                                 
                \If{(CAS($p.left.value$, $\textbf{empty}$, $p.value$) = \textbf{empty})}  \label{lst2:line:ins-lp2}  \label{lst2:line:growins2} 
                	   \State Do insert to the right                    
	   	   \State Decrement $T_x.opcount$ 
                \Else
                    \State Re-try insert from $p$  \label{lst2:line:retryinsright}             
                \EndIf
            \Else
                \State  Go to right child ($p \gets p.right$) and restart from \textit{BEGIN}             	\label{lst2:line:right-end}
             \EndIf
        \ElsIf {($v = p.value$)}                                                                                            
            \If{Is at leaf ($p.isleaf =$ TRUE)}                                   
                \If{$v$ is not deleted ($p.mark =$ FALSE)}								
                    \State Decrement $T_x.opcount$                            \Comment{$v$ already exist} \label{lst2:line:valexist}
                \Else
                	   \State Do insert right
                \EndIf
            \Else
                \State Go to right child ($p \gets p.right$) and restart from \textit{BEGIN}       
            \EndIf
        \EndIf
\Else			
		\If{$v$ is already in $T_x.\mathit{buffer}$}
			\State Decrement $T_x.opcount$
		\Else
			\State Put $v$ inside $T_x.\mathit{buffer}$ \label{lst2:line:insertbuffer}		\Comment{buffered insert}	
			\If{(\textsc{TAS}($T_x.locked$))} \Comment{Acquire maintenance lock} \label{lst2:line:lockbuf}
			\State Decrement $T_x.opcount$ \label{lst2:line:dec-opcount}
			\State Wait for all updates to finish ($T_x.opcount$=0)  \label{lst2:line:spinwait-main}
			\State do \textsc{rebalance}($T_x$) or \textsc{expand($p$)}  \label{lst2:line:expand}
			\EndIf 
		\EndIf
\EndIf
\EndFunction
\algstore{bkbreak}
\end{algorithmic}
\end{figure}
\clearpage
\begin{figure}
\footnotesize
\begin{algorithmic}[1]

\algrestore{bkbreak}

\Function{deleteNode}{$v, p$}		\label{lst2:line:deletefunc}				

\State $p \gets U.root$
\State \textit{BEGIN}:
\If{(\{Entering new $\Delta$Node $T_x$\})}                 			                                                 
   \State  Decrement previous $\Delta$Node's opcount ($T'_x.opcount$)				
   \State  Wait if $\Delta$Node is currently maintained 		\label{lst2:line:waitdelete}
\EndIf

\If{Is at leaf ($p.isleaf =$ TRUE)) \textbf{or} at the tree's last level node}
        \If{($p.value = v$)}
            \If{(\textsc{CAS}($p.mark$, FALSE, TRUE) != FALSE))}              	\label{lst2:line:markdel}       
                	\State Decrement  $T_x.opcount$	                                      \Comment{$v$ is already deleted}
            \Else	
             	\If{$p$ is still the leaf node} 	\label{lst2:line:markdel-check}
			\If{(\textsc{TAS}($T_x.locked$))} \Comment{Acquire the maintenance lock} \label{lst2:line:lockbuf-del}
				\State Decrement $T_x.opcount$	\label{lst2:line:dec-opcount-del}
				\State Wait for all updates to finish ($T_x.opcount$=0) \label{lst2:line:spinwait-main-del} 
                			\State Do node merge if needed \label{lst2:line:merge}          
			\EndIf       
		\Else		 
			\State Re-try delete from current node $p$ \label{lst2:line:retrydel}
		\EndIf
            \EndIf
        \Else
        		\State Search ($T_{x}.\mathit{buffer}$) for $v$
	   	\If{$v$ is found in $T_{x}.\mathit{buffer.idx}$}
                 	\State Remove $v$ from $T_{x}.\mathit{buffer.idx}$	\label{lst2:line:bufdel} 	\Comment{buffered delete}
	    	\EndIf
		\State Decrement $T_x.opcount$                                            
    	\EndIf
\Else 
	\If{($v < p.value$)}
        		\State Go to left child ($p \gets p.left$)
	\Else
        		\State Go to right child  ($p \gets p.right$)
	\EndIf	
	\State Restart from \textit{BEGIN}
\EndIf
\EndFunction

\end{algorithmic}
\caption{$\Delta$Tree's concurrent update (\textit{insert} and \textit{delete}) algorithms.}
\label{lst:pseudo-ops-simple}
\end{figure}

Function \textsc{insertNode($v, U$)} (cf. Figure \ref{lst:pseudo-ops-simple}, line \ref{lst2:line:insertfunc}) 
inserts value $v$ at a leaf of $\Delta$Tree, 
provided $v$ does not exist in the
tree (Figure \ref{lst:pseudo-ops-simple}, line \ref{lst2:line:valexist}). 
Following the nature of a leaf-oriented tree, a successful insert operation 
replaces a leaf with a subtree of three nodes \cite{EllenFRB10} (cf. Figure
\ref{fig:treetransform}a and in Figure \ref{lst:pseudo-ops-simple}, 
lines \ref{lst2:line:growins1} and \ref{lst2:line:growins2}). Because the $\Delta$Tree
structure is pre-allocated, an insert operation needs only to do a "logical" replacement, or 
only replacing the value of a leaf and its children.

\begin{figure}[!t] \centering \includegraphics[scale=0.8]{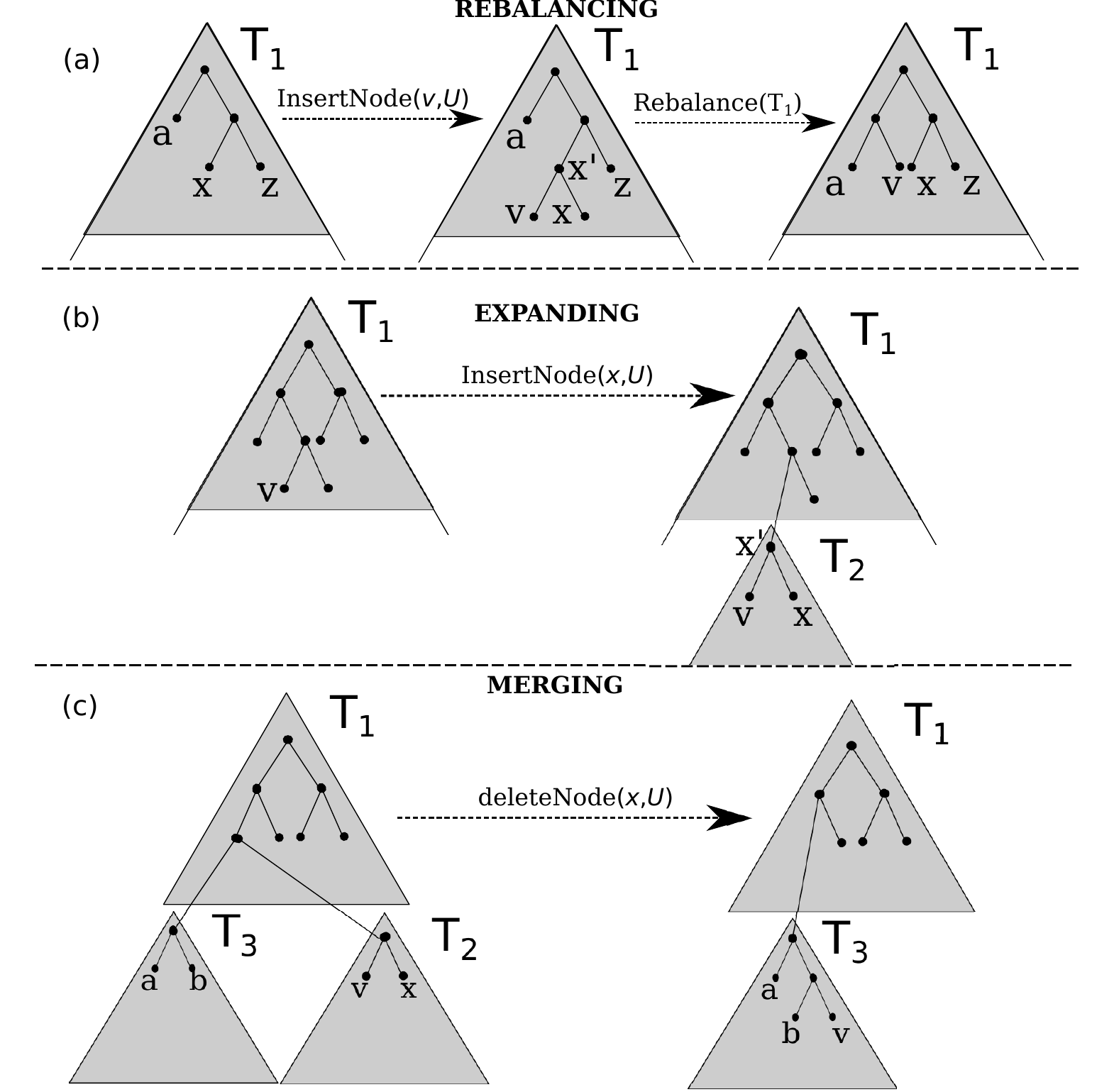}
\caption{(a)\textit{Rebalancing}, (b)\textit{Expanding}, and (c)\textit{Merging} operations on $\Delta$Tree.}
\label{fig:treetransform}
\end{figure}

The function \textsc{deleteNode($v, U$)} (cf. Figure \ref{lst:pseudo-ops-simple}, line \ref{lst2:line:deletefunc}) 
{\em marks} the leaf that contains value $v$ as deleted. \textsc{deleteNode($v, U$)} fails if $v$ does not exist in the tree or 
the leaf containing $v$ is already mark as deleted. 

To avoid conflicting insert and delete operations, 
\textsc{insertNode} and \textsc{deleteNode} operations
are using the single word CAS (Compare and Swap)
and \textit{leaf-checking} to coordinate between the
operations (cf. Figure \ref{lst:pseudo-ops-simple}, lines \ref{lst2:line:markdel} and 
\ref{lst2:line:markdel-check} for delete, and \ref{lst2:line:ins-lp1}
and \ref{lst2:line:ins-lp2} for insert).

\subparagraph{Maintenance functions.} \label{sec:maintenance-func}

Other than the basic functions, $\Delta$Tree has tree maintenance functions 
that are invoked in certain cases of inserting and deleting a node from the tree. 
These functions, namely rebalance, expand and merge, are the unique
feature of $\Delta$Tree that acts as a garbage collectors for marked nodes 
and as a safeguard to maintain $\Delta$Tree's small height.

Function \textsc{rebalance($T_v.root$)} (cf.  Figure \ref{lst:pseudo-ops-simple} 
line \ref{lst2:line:expand}) is responsible for rebalancing a
$\Delta$Node after an insertion.
Figure \ref{fig:treetransform}a illustrates the rebalance operation. If a
new node $v$ is to be inserted at the last level $H$ of $\Delta$Node $T_1$, the
$\Delta$Node is rebalanced to a complete BST. All
leaves' heights are then set to $\lfloor\log N\rfloor + 1$,
where $N$ is the number of leaves (e.g., $a,v,x,z$ in Figure \ref{fig:treetransform}a). 
In the same Figure \ref{fig:treetransform}a, after rebalancing, tree $T_1$ height reduces from 4 to 3.

We also define the \textsc{expand($l$)} function (cf.  Figure \ref{lst:pseudo-ops-simple} 
line \ref{lst2:line:expand}) that is responsible for
creating a new $\Delta$Node and connecting it to the parent of a leaf node $l$ 
(cf. Figure \ref{fig:treetransform}b).
Expanding is triggered only if 1) after \textsc{insertNode($v, U$)}, leaf that contains 
$v$
will be at the last level of a $\Delta$Node; and 2) rebalancing will no longer
reduce the current height of $\Delta$Node $T_i$. 
For example if an expand takes place at node $l$ which is located at the lowest
level of a $\Delta$Node, or $depth(l) = H$, 
the parent of $l$ swaps one of its child pointers that previously points to
$l$ into the root of the newly created $\Delta$Node's $root$ (cf. Figure \ref{fig:treetransform}b).

Function \textsc{merge($T_x.root$)} (cf.  Figure \ref{lst:pseudo-ops-simple} 
line \ref{lst2:line:merge}) is defined to merge $T_x$ with its sibling. 
For example in Figure \ref{fig:treetransform}c, $T_2$ is
merged into $T_3$. Then
the pointer of $T_3$'s grandparent that previously points
to the parent of both $T_3$ and $T_2$ is replaced by a pointer to $T_3$.
Merge operation is invoked provided that a particular $\Delta$Node, where 
a deletion has taken place, is filled less than $2^{^t/_2}$ of its capacity 
(where $t = \log(\UB)$) and both
nodes of that $\Delta$Node and its siblings are going to fit into a single $\Delta$Node. 

Before any maintenance operation starts, the threads responsible 
are required to do Test and Set (TAS) on $\Delta$Node's \textit{locked} 
field (cf. Figure \ref{lst:pseudo-ops-simple}, lines \ref{lst2:line:lockbuf} and  \ref{lst2:line:lockbuf-del}). Advanced locking techniques \cite{HaPT07_JSS, KarlinLMO91, LimA94} can also be used. If a $\Delta$Node's \textit{locked} has been set to TRUE,
new incoming update operations are forced to wait at the tip of that $\Delta$Node
so that they will not interfere with the maintenance operation  
(cf. Figure \ref{lst:pseudo-ops-simple}, line \ref{lst2:line:waitdelete}
and \ref{lst2:line:waitinsert}) . 

For efficiency, a $\Delta$Node's maintenance operation starts 
after a batch of concurrent update operations within that $\Delta$Node. The TAS
also serves as a mechanism to make the concurrent update operations compete for 
a maintenance lock (\textit{locked} variable). The losing threads will
leave their maintenance burden 
to a thread who successfully acquired the $\Delta$Node's \textit{locked} variable.

After the maintenance operation finished, $\Delta$Node's \textit{locked} 
is released by the maintenance thread. 
Therefore, waiting update threads can continue with their respective operation 
within the newly-maintained $\Delta$Node.

Note that although \textsc{rebalance}
and \textsc{merge} execution are sequential within a $\Delta$Node, invoking 
these functions involves at most $2\times\UB$ nodes, where a constant $\UB \ll N$.

\subparagraph{Mirroring.} \label{sec:mirroring}
Whenever a $\Delta$Node is undergoing a
maintenance operation (balancing, expanding, or merging), a mirroring operation
also takes place to facilitate the wait-free search. 
Mirroring works by maintaining the original $\Delta$Node and writing 
the results into the mirror $\Delta$Node. 

After a maintenance operation finishes, the pointer to the root of the maintained $\Delta$Node 
is switched to the root of the $\Delta$Node's mirror. In result, the mirror, or maintained $\Delta$Node, 
is now become part 
of $\Delta$Tree, replacing the "old" $\Delta$Node. As the $\Delta$Node's buffer is enclosed within a $\Delta$Node, 
it is also switched at the same time $\Delta$Node was switched. 
As a result, the original $\Delta$Node and its helping buffer served as the latest snapshot, which
enables wait-free search for $\Delta$Tree. 

\paragraph{Balanced $\Delta$Tree}\label{sec:bDeltaTree} 

$\Delta$Tree served as a proof of concept
of a concurrency-aware vEB-based search trees, although it has major weaknesses that 
can affect its empirical performance.
First, the 
\textit{left} and \textit{right} pointers are occupying larger memory 
spaces compared to the data values. For example, in a $\Delta$Node 
with 127 nodes, the set of pointers will occupy 2032 bytes ($127 \times 16$ bytes) 
of memory in a 64-bit operating system. 
Therefore, the amount of space used by pointers alone 
is four times the amount of space used by the data values (i.e., $127 \times 4$ bytes = 508 bytes, 
assuming node's $value$ is a 4 bytes integer).
This is inefficient, because every block transfer between any levels of memory carries
a bigger portion of pointers instead of important data.
Second, inserting a sequence of increasing or decreasing 
numbers into $\Delta$Tree will results a linked-list of $\Delta$Node (i.e., $\Delta$Tree's
height is equal to $N$, or the number of elements). This is a direct implication 
of the unbalanced form of $\Delta$Tree.  

Balanced $\Delta$Tree (b$\Delta$Tree) is implemented to address the $\Delta$Tree's weaknesses. 
b$\Delta$Tree is developed by devising three major
strategies, namely replacing internal $\Delta$Node pointers with map, crafting the
efficient inter-node connection, and ensures that the tree is always balanced.

Improving $\Delta$Tree to become a pointer-less concurrent cache-oblivious 
is difficult. For example, Brodal, et al. \cite{BrodalFJ02}
have discussed about pointer-less CO BST and their
implementation involves heavy calculation to determine the children positions of a node.
In contrast to our $\Delta$Tree that consists of group of $\Delta$Nodes, 
they are treating their tree as $N$-sized vEB-layout tree. Thus, not only we need 
to address the node parent-child pointers, we must also consider the inter-$\Delta$Node pointers. 

As a matter of fact, having fixed-sized $\Delta$Nodes ($\UB \ll N$) is 
actually advantageous. A single pre-calculated
map of all $\UB$ nodes can be generated once and used by all operations 
throughout the 
whole tree. The inter-$\Delta$Node connection still need pointers
though, but instead of $2\times\UB$ pointers, we will show that 
a significant reduction of pointers ($^1/_2\times\UB$) is possible.

\subparagraph{"Map" instead of pointers.}  \label{sec:mapdesc}
\begin{figure}[!t] \centering 
\begin{algorithmic}[1]

\Start{\textbf{Map}} \label{lst:line:map}
	\START{member fields:}
\State $left \in
	\mathbb{N}$, interval of the \textit{left} child pointer address
\State $right \in
	\mathbb{N}$, interval of the \textit{right} child pointer address 
\END \End

\Statex
\State \textbf{Map} $map[\UB]$

\Statex 

\Function{right}{p, base}  \label{lst:line:right}
    \State $nodesize \gets  \textsc{sizeOf}(\textbf{single node})$
    \State $idx \gets (p - base)/nodesize$
    
    \If {($map[idx].right$ != 0)}
        \State \Return $base + map[idx].right$
    \Else
        \State \Return 0
    \EndIf
\EndFunction

\Statex 

\Function{left}{p, base} \label{lst:line:left}
    \State $nodesize \gets  \textsc{sizeOf}(\textbf{single node})$
    \State $idx \gets (p - base)/nodesize$
    
    \If {($map[idx].left$ != 0)}
        \State \Return $base + map[idx].left$
    \Else
        \State \Return 0
    \EndIf
\EndFunction
\end{algorithmic}
\caption{\textit{Mapping} functions.}
\label{lst:map-func}
\end{figure}

The Balanced $\Delta$Tree is implemented 
by completely eliminating (\textit{left} and \textit{right}) 
pointers within a $\Delta$Node. Instead, these pointers are replaced with the corresponding
\textsc{left} and \textsc{right} functions (Figure \ref{lst:map-func}, lines  \ref{lst:line:right}
and \ref{lst:line:left}). \textsc{left} and \textsc{right} functions, given an arbitrary node $v$ and
the root memory address of its container $\Delta$Node, will return the the addresses 
of the left and right child nodes of $v$. These functions will return $0$ if $v$ is the deepest node
within a $\Delta$Node. The \textsc{left} and \textsc{right} functions requires 
a \textit{map} array with $\UB$ length to determine a node's children address (Figure  \ref{lst:map-func}, line  \ref{lst:line:map}).

$\Delta$Node's \textit{map} is a read-only record of the nodes' memory address differences.
The address differences of every $\Delta$Node's nodes and its left and right children 
are pre-calculated and recorded when the Balanced $\Delta$Tree is initialized. 
Because $\Delta$Nodes are using the same fixed-size concurrency-aware vEB layout, 
only one map is needed for all of the available $\Delta$Nodes.  
Therefore, since only one \textit{map} with size
$\UB$ is used for all traversing, memory footprint for the Balanced $\Delta$Tree's 
$\Delta$Node operations can be kept minimum (namely, the read-only map is always cached). 
Please note that the \textit{map} only supports internal (within $\Delta$Node) node traversal. 

\subparagraph{Inter-$\Delta$Node connection.}  \label{sec:interdesc}

To facilitate the traversing from one $\Delta$Node to another, the inter-$\Delta$Node connection
mechanism is used in addition to the \textit{map}.
To use the inter-$\Delta$Node connection, first we logically assign a color to each $\Delta$Node's
node edge. 
Each node has only two outgoing edges, which are the left edge and the right edge. 
These left and right edges are assigned color 0 and 1, respectively.
Now any paths traversed from the $root$ of a $\Delta$Node to reach 
any internal node will produce a bit-sequence of colors. This bit representation
can be translated into an array index that contains a pointer to a child $\Delta$Node. 
The maximum size of the bit representation is the height of $\Delta$Node or $\log(\UB)$ bits.
We are using a leaf-oriented tree and allocate a pointer array with
$^1/_2\times\UB$ length. Pseudocode in Figure \ref {lst:scannode-func}
explains how the inter-$\Delta$Node connection works in a pointer-less search function.

The removal of the internal nodes' pointers results in 200\% more
node counts in a $\Delta$Node compared to the $\Delta$Tree with the same $\UB$ size.
A $\Delta$Node's struct is now consist of internal nodes, which is just an array of keys, 
plus an array of $^1/_2\times\UB$ pointers for the inter-$\Delta$Node connection.

\begin{figure}[!t] \centering 
\begin{algorithmic}[1] 

\Function{pointerLessSearch}{key, $\Delta$Node, maxDepth}
    
  \While{$\Delta$Node is not leaf}
   	\State $bits \gets 0$
    	\State $depth \gets 0$
   	\State $p \gets {\Delta}Node.root$
   	\State $base \gets p$
    	\State $link \gets {\Delta}Node.link$
    
    	\While{($p \And p.value$ != EMPTY)}  \Comment{continue until leaf node} 
        		\State $depth \gets depth + 1$ \Comment{increment depth}
        		\State $bits \gets bits << 1$ \Comment{shift one bit to the left in each level}
        		\If {($key < p.value$)}
             		\State $p \gets$ \textsc{left}($p, base$)
        		\Else
            		\State $p \gets$ \textsc{right}($p, base$)
            		\State $bits \gets bits+1$ \Comment{right child color is 1}
        		\EndIf
    	\EndWhile
	\LComment{pad the $bits$ to get the index of the child $\Delta$Node's pointer:}
    	\State $bits \gets bits << (maxDepth - depth) - 1$ 
    	\LComment{follow nextRight if highKey is less than searched value:}
    	\If{(${\Delta}Node.highKey <= key$)}
     		\State ${\Delta}Node \gets {\Delta}Node.nextRight$
   	\Else
     		\State ${\Delta}Node \gets link[bits]$ \Comment{jump to child $\Delta$Node}
   	\EndIf
  \EndWhile
  \State \Return ${\Delta}Node$
\EndFunction

\end{algorithmic}
\caption{Search within pointer-less $\Delta$Node. This function will return the {\em leaf} $\Delta$Node containing
the searched key. From there, a simple binary search using \textsc{left} and \textsc{right} functions
is adequate to pinpoint the key location. The search operations are utilizing both the 
$nextRight$ pointers and  $highKey$
variables to handle concurrent $\Delta$Node splitting~\cite{Lehman:1981:ELC:319628.319663}.}
\label{lst:scannode-func}
\end{figure}

\subparagraph{Concurrent and balanced tree.}

To enable both of the concurrency and balanced tree out of the $\Delta$Tree, 
new variables are added into the $\Delta$Nodes and
changes are made to the $\Delta$Tree's update and maintenance operations.
We improved $\Delta$Tree's update algorithms by adopting 
the algorithms from the concurrent lock-based 
B-link trees \cite{Lehman:1981:ELC:319628.319663} 
by Lehman and Yao.

The insert operations in b$\Delta$Trees are now done in a bottom-up fashion 
to ensure that the tree is always balanced.
Meanwhile, the search operations are done in a top-down, left-to-right fashion.

The modification to how the insert and search operations work results in the omission
of both \textsc{merge($T_x.root$)} and \textsc{expand($l$)} functions 
(cf. Section \ref{sec:maintenance-func}). 
Instead, the new \textsc{split($T_v$)} functions is introduced to aid the bottom-up
tree building.

\textsc{split($T_v$)} functions splits the $\Delta$Node $T_v$ into $T'_v$ and 
a right sibling $T_x$. The internal nodes that were previously occupying 
$T_v$ are distributed evenly between $T'_v$ and $T_x$. 
If $T_v$ is the topmost $\Delta$Node, a new parent $\Delta$Node $T_p$ is created as well, taking
both $T'_v$ and $T_x$ as its children. 
The split operation concludes with an insertion of the lowest value from $T_x$ nodes
into $T_p$.
The $\Delta$Node split is triggered only when a $\Delta$Node overflows after an insertion, 
which means it is at least filled with ($2^{^t/_2}$) internal nodes, where $t=\log(\UB)$.

Beside the addition of split operation, 
two additional variables are added inside the struct \textsc{$\Delta$NodeMeta}. 
These additional variables 
are $nextRight$ pointer, which points to the right sibling $\Delta$Node; and $highKey$
variable that contains the upper-bound value of the $\Delta$Node. Starting with NULL values,
both of these new variables are populated when the $\Delta$Node splits. For example, 
after $T_v$ splits, the $T_v.nextRight$ is going to point to $T_x$, 
while $T_v.highKey$ is set to the minimum value of all the
nodes in $T_x$. These new variables
help the concurrent search with respect to the node splits \cite{Lehman:1981:ELC:319628.319663}.

With the addition of $nextRight$ and $highKey$ variables, 
Balanced $\Delta$Tree's search operations do not need locks and waits
(cf. Figure \ref{lst:scannode-func}). However, 
the search cannot be regarded as wait-free. 
According to Lehman and Yao \cite{Lehman:1981:ELC:319628.319663}, 
in the worst case, which is extremely unlikely, an infinite loop caused by 
continuous $\Delta$Node splits might happen.

To address the memory waste, 
the same rebalancing procedure as in Figure \ref{fig:treetransform}a is employed by b$\Delta$Tree. 
The rebalancing also helps
to clean-up the nodes marked for deletion, keeping $\Delta$Nodes always in a good shape.
A mirroring whenever the tree splits and rebalances is also retained to make the tree
traversals not require any locks and waits.

\paragraph{Heterogeneous $\Delta$Tree}\label{sec:hDeltaTree}
 
In Balanced $\Delta$Tree, the leaf-oriented (or external tree) layout is 
adopted for $\Delta$Nodes in order to facilitate the inter-$\Delta$Node connection mechanism
using $^1/_2\times\UB$ pointers.
However, it is not necessary for the leaf $\Delta$Nodes to have leaf-oriented 
layout since they do not have any successor $\Delta$Nodes.
 
Based on the above observation, we devise a heterogeneous balanced 
$\Delta$Tree by changing the layout of the balanced $\Delta$Tree's leaf $\Delta$Nodes.
This new layout uses internal tree instead of the external tree for the leaf $\Delta$Nodes. 
With this layout change, 100\% more key nodes are now fit into any leaf $\Delta$Node,
if compared to the previously non-leaf $\Delta$Nodes with the same $\UB$ limit.
To save more memory spaces, we also remove the array of pointers for intra-$\Delta$Node connection
in the leaf $\Delta$Nodes metadata. 

With this improved version of $\Delta$Tree, or heterogeneous $\Delta$Tree, we find that the efficiency of 
searches is improved (cf. Section \ref{sec:perfeval}). Compared to $\Delta$Tree and balanced $\Delta$Tree, the heterogeneous $\Delta$Tree
delivers lower number of cache misses and more efficient branching.

\subsubsection{Libraries of Concurrent Search Trees}
We have developed concurrent search tree libraries that contain the following components:
\begin{enumerate}
  \item Non-blocking binary search tree (NBBST). The Non-blocking binary search tree (NBBST) library contains the non-blocking binary search tree of Ellen \cite{EllenFRB10}.
  \item STM-based search trees. STM-based search trees library contains the Software Transactional Memory (STM)-based AVL tree (AVLtree), red-black tree (RBtree), and speculation
friendly tree (SFtree) from the Synchrobench benchmark \cite{synchrobench}.
  \item Concurrent B-tree (CBtree). Concurrent B-tree (CBTree) library contains an optimized Lehman and Yao B-link tree \cite{Lehman:1981:ELC:319628.319663}. 
B-link tree is a highly-concurrent B-tree variant 
and it is still being used as a backend in popular database systems such as 
PostgreSQL\footnote{\url{https://github.com/postgres/postgres/blob/master/src/backend/access/nbtree/README}}.
  \item Static cache-oblivious tree using the static vEB layout (VTMtree). The Concurrent static vEB binary search tree (VTMTree) library contains a concurrent version of the static vEB binary search tree  \cite{BrodalFJ02} developed using GNU C Compiler v4.9.1's STM.
  \item DeltaTree ($\Delta$Tree). The tree families include Delta tree, Balanced Delta tree and Heterogeneous Delta tree are described in Section ~\ref{sec:delta-trees}.

\end{enumerate}

\paragraph{Obtaining and compilation.} The libraries are provided in a separate directory for easy access and maintenance. The repository address 
is http://gitlab.excess-project.eu/ibrahim/tree-libraries. A makefile for each of the libraries is also provided to aid compilations. 
The libraries have been tested on Linux and Mac OS X platforms.

\paragraph{Running and outputs.} By default, the provided makefile will build the standalone benchmark 
version of the libraries which will accept these following parameters:

\noindent\fbox{%
    \parbox{0.95\textwidth}{%
   \texttt{-r <NUM>    : Allowable range for each element (0..NUM)\\
-u <0..100> : Update ratio. 0 = Only search; 100 = Only updates\\
-i <NUM>    : Initial tree size (initial pre-filled element count)\\
-t <NUM>    : DeltaNode ($\UB$) size (ONLY USED IN DELTATREE FAMILIES)\\
-n <NUM>    : Number of benchmark threads\\
-s <NUM>    : Random seed. (0 = using time as seed, Default)
}}}
\\
\\
The benchmark outputs are formatted in this sequence: 

\noindent\fbox{%
    \parbox{0.95\textwidth}{%
    \texttt{
0: range, insert ratio, delete ratio, \#threads, \#attempted insert, \#attempted delete, \#attempted search, \#effective insert, \#effective delete, \#effective search, time (in msec.)
}}}

NOTE: {\tt 0:} characters are just unique token for easy tagging (e.g., for using {\tt grep}).

\noindent\fbox{%
    \parbox{0.95\textwidth}{%
\ttfamily
\$ ./DeltaTree -h \newline
DeltaTree v0.1\newline
===============\newline
Use -h switch for help.\newline
\newline
Accepted parameters\newline
-r <NUM>    : Range size\newline
-u <0..100> : Update ratio. 0 = Only search; 100 = Only updates\newline
-i <NUM>    : Initial tree size (inital pre-filled element count)\newline
-t <NUM>    : DeltaNode size\newline
-n <NUM>    : Number of threads\newline
-s <NUM>    : Random seed. 0 = using time as seed\newline
-d <0..1>   : Density (in float)\newline
-v <0 or 1> : Valgrind mode (less stats). 0 = False; 1 = True\newline
-h          : This help\newline
\newline
Benchmark output format: \newline
"0: range, insert ratio, delete ratio, \#threads, attempted insert, attempted delete, attempted search, effective insert, effective delete, effective search, time (in msec)"
}}

\noindent\fbox{%
    \parbox{0.95\textwidth}{%
\ttfamily
\$ ./DeltaTree -r 5000000 -u 10 -i 1024000 -n 10 -s 0\newline
DeltaTree v0.1\newline
===============\newline
Use -h switch for help.\newline
\newline
Parameters:\newline
- Range size r:		 5000000\newline
- DeltaNode size t:	 127\newline
- Update rate u:	 10\% \newline
- Number of threads n:	 10\newline
- Initial tree size i:	 1024000\newline
- Random seed s:	 0\newline
- Density d:		 0.500000\newline
- Valgrind mode v:	 0\newline
\newline
Finished building initial DeltaTree\newline
The node size is: 25 bytes\newline
Now pre-filling 1024000 random elements...\newline
...Done!\newline
\newline
Finished init a DeltaTree using DeltaNode size 127, with initial 1024000 members\newline
\#TS: 1421050928, 511389\newline
Starting benchmark...\newline
Pinning to core 0... Success\newline
Pinning to core 3... Success\newline
Pinning to core 1... Success\newline
Pinning to core 8... Success\newline
Pinning to core 9... Success\newline
Pinning to core 10... Success\newline
Pinning to core 2... Success\newline
Pinning to core 11... Success\newline
Pinning to core 4... Success\newline
Pinning to core 12... Success\newline
\newline
0: 5000000, 5.00, 5.00, 10, 249410, 248857, 4501733, 195052, 53720, 1000568, 476\newline
\newline
Active (alloc'd) triangle:258187(266398), Min Depth:12, Max Depth:30 \newline
Node Count:1165332, Node Count(MAX): 1217838, Rebalance (Insert) Done: 234, Rebalance (Delete) Done: 0, Merging Done: 1\newline
Insert Count:195052, Delete Count:53720, Failed Insert:54358, Failed Delete:195137 \newline
Entering top: 0, Waiting at the top:0
}}

NOTE: {\tt \#TS:} is the benchmark start timestamp. 

\paragraph{Pluggable library.} To use any component as a library, each library provides a (.h) header file and a simple, 
uniform API in C. These available and callable APIs are:

\noindent\fbox{%
    \parbox{0.95\textwidth}{%
\textsc{Structure:}\\
\\
\texttt{<libname>\_t} : Structure variable declaration.\\
\\
\textsc{Functions:}\\
\\
\texttt{<libname>\_t* <libname>\_alloc()} :	Function to allocate the defined structure, returns the allocated (empty) structure.\\
\\
\texttt{void* <libname>\_free(<libname>\_t* map)} :	Function to release all memory used by the structure, returns NULL on success.\\
\\
\texttt{int <libname>\_insert(<libname>\_t* map, void* key, void* data)} : Function to insert a key and a linked pointer (data), returns 1 on success and 0 otherwise.\\
\\
\texttt{int <libname>\_contains(<libname>\_t* map, void* key)} : Function to check whether a key is available in the structure, returns 1 if yes and 0 otherwise.\\
\\
\texttt{void *<libname>\_get(<libname>\_t* map, void* key)} : Function to get the linked data given its key, returns the pointer of the data of the corresponding key and 0 if the 
key is not found.\\
\\
\texttt{int <libname>\_delete(<libname>\_t* map, void* key)} : Function to delete an element that matches the given key, returns 1 on success and 0 otherwise. 
}}

As an example, the concurrent B-tree library provides the \texttt{cbtree.h} file that can be linked into
any C source code and provides the callable \texttt{cbtree\_t* cbtree\_alloc()} function. It is also possible
to use the MAP selector header (\texttt{map\_select.h}) plus defining which tree type to use 
so that  MAP\_\textless{operator}\textgreater functions are used instead as specific tree function as the below example:

\begin{lstlisting}[frame=single, language=C]
#define MAP_USE_CBTREE 
#include "map_select.h"

int main(void)
{
	long numData = 10;
	long i;
	char *str;
	puts("Starting...");
	MAP_T* cbtreePtr = MAP_ALLOC(void, void); 
	assert(cbtreePtr);
	for (i = 0; i < numData; i++) { 
		str = calloc(1, sizeof(char)); *str = ÕaÕ+(i%254);
		MAP INSERT(cbtreePtr, i+1, str); 
	}
	for (i = 0; i < numData; i++) {
		printf("%ld: %c\n", i+1, 
			*((char*)MAP_FIND(cbtreePtr, i+1))); 
	}
	for (i = 0; i < numData; i++) {
		printf("%ld: %d\n", i+1, 
			MAP_CONTAINS(cbtreePtr, i+1));
	}
	for (i = 0; i < numData; i++) {
		MAP_REMOVE(cbtreePtr, i+1);
	}
	for (i = 0; i < numData; i++) {
		printf("%ld: %d\n", i+1, 
			MAP_CONTAINS(cbtreePtr, i+1));
	}
	MAP_FREE(cbtreePtr)
	puts("Done."); 
	return 0;
}

\end{lstlisting}

\paragraph{Intel PCM integration.} All of the libraries provide support for Intel PCM measurement. 
To enable Intel PCM measurement metrics,
the compiler must be invoked using \texttt{-DUSE\_PCM} parameter during the libraries's compilation 
and all the Intel PCM compiled object files must be linked to the output executables.
  





\subsubsection{Performance and Energy Analysis of Concurrent Search Trees} \label{sec:evaluation}
We have experimentally analyzed the performance and energy efficiency of $\Delta$Tree (Section \ref{sec:DeltaTree}), 
balanced $\Delta$Tree (b$\Delta$Tree) (Section \ref{sec:bDeltaTree}), and heterogeneous $\Delta$Tree
(h$\Delta$Tree) (Section \ref{sec:hDeltaTree}) in comparison with that of the other prominent concurrent search trees in the libraries (cf. Section \ref{sec:benchsetup}). The experimental evaluation was conducted on both Intel high performance computing (HPC) and ARM embedded platforms. 
We have also evaluated how $\Delta$Trees would perform 
in the worst-case and average-case setups against the highly optimized B-tree  and the widely used GCC's std::set (cf. Section \ref{sec:worstCase}).
The experimental evaluation was conducted on both Intel high performance computing (HPC) platforms and an ARM embedded platform. 

\paragraph{Testbed choices} \label{sec:benchsetup}

Pthreads were used for threading and all running threads were pinned to
the available physical cores using pthread\_setaffinity\_np. 
GCC 4.9.1 was used with -O2 for all program compilations. All of the tests
are repeated at least $10\times$ to guarantee consistent results.

Several decisions have been made for choosing VTMtree's size, $\Delta$Nodes' $\UB$
and CBTree's block size. Since the VTMtree's size was fixed, 
we set it to $2^{23}$ so that VTMtree needed not to expand. For fair
comparisons, $\Delta$Trees' $\UB$ and CBTree's 
order were set to their respective values so that each $\Delta$Node and CBTree's node fitted within a memory page (4KB).
Therefore, CBTree's order was set to 336.

\subparagraph{Performance benchmark setup}

Performance indicators (in operations/second) were calculated 
using the number of ($\mathit{rep}=\num[group-separator={,}]{5000000}$) operations 
divided by the maximum time for the threads to finish the whole operations.
Combination of
update rate $u=\{0, 20, 50\}$ and number of thread 
$\mathit{nr}=\{1, 2, \ldots, 16\}$ were used for each run.
Update rate of 0 equals to 100\% search, 
while 50 update rate equals to 50\% insert and delete operations out of $\mathit{rep}$ operations.
All involved operations used random
values of $v \in (0, \mathit{init} \times 2], v \in \mathbb{N}$. 
All the trees were initialized with a number $init$ of nodes to simulate trees 
that partially fit into the last level cache (LLC) of the evaluating platform. 
The $init$ values were $\num[group-separator={,}]{4194303}$ and $\num[group-separator={,}]{2097151}$ for the Intel HPC platform and ARM embedded platform, respectively. 

The Intel HPC platform was equipped with 
dual Intel Xeon E5-2670 CPUs, 
with total of 16 cores (no hyperthreading). The platform had 32GB of RAM,  
2MB (8$\times$256KB) L2 cache, and shared 20MB L3 cache for each CPU. 
Linux with kernel version 2.6.32-358 was used in this platform.

The embedded ARM platform was a Samsung Exynos 5410
octa-core ARM board with ARM's big.LITTLE CPU. This board was equipped with four Cortex-A15 1.6Ghz cores 
and four Cortex-A7 1.2Ghz cores in a single CPU and 2GB LPDDR3 RAM. The A15 cores
had a 2MB shared L2 cache while the A7 cores had a 512KB shared L2 cache.
Although the CPU had a total of 8
cores, only 4 cores could be active at a time because of its design limitation. Linux with kernel version 3.4.98 was used
in this board. As the GCC compiler available for this board did not support the transactional memory
extension, we had to remove VTMtree from the benchmarks for this platform.

\subparagraph{Energy benchmark setup}

To assess the energy consumption of the trees, energy indicators 
were subsequently collected during 100\% search and 50\% update benchmarks.
The ARM platform was equipped with a built-in "off-chip" power measurement system that 
was able to measure the energy for the A15 cores, A7 cores, and memory in real-time. 
For the Intel HPC platform, a server with two Intel Xeon E5-2690 v2
of 20 cores in total was used. The Intel PCM library using built-in CPU counters was used to measure 
the energy for each CPU and DRAM. 
  
The benchmarks run only with the minimum and maximum
available physical cores. The total energy
consumed by all CPUs and memory (in Joules) was collected and divided by the number of
operations. The collected measurements did not include the initialization cost of trees.

\paragraph{Performance results} \label{sec:perfeval}

\begin{figure}[!t] \centering
\resizebox{0.9\columnwidth}{!}{\scriptsize\input{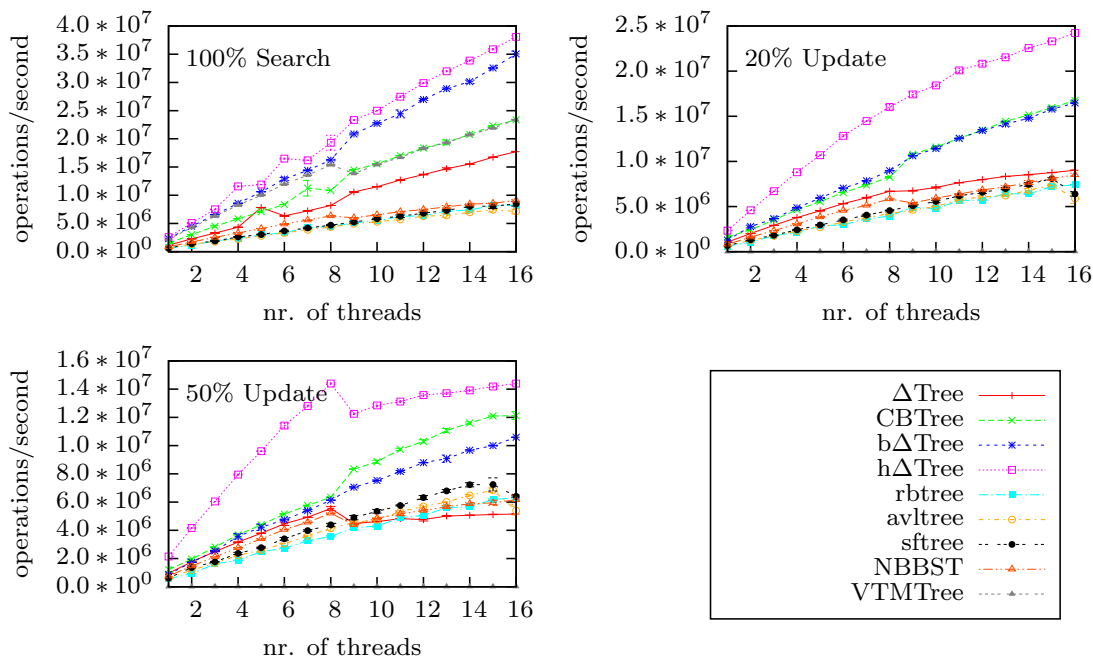}}
\caption{Performance comparison of the tested trees with 
\num[group-separator={,}]{4194303} initial members on an Intel HPC platform with two 8-core chips. On a single chip, the heterogeneous $\Delta$Tree (h$\Delta$Tree) is 140\% faster than the concurrent B-tree (CBTree) in the 50\%-update benchmark with 8 threads. The h$\Delta$Tree performance decreases when the number of threads goes from 8 to 9 because of the cache-coherence issue between two chips (cf. Section \ref{sec:perfeval}).}
\label{fig:perf-4194303}
\end{figure}

\begin{figure}[!t] \centering
\resizebox{0.9\columnwidth}{!}{\scriptsize\input{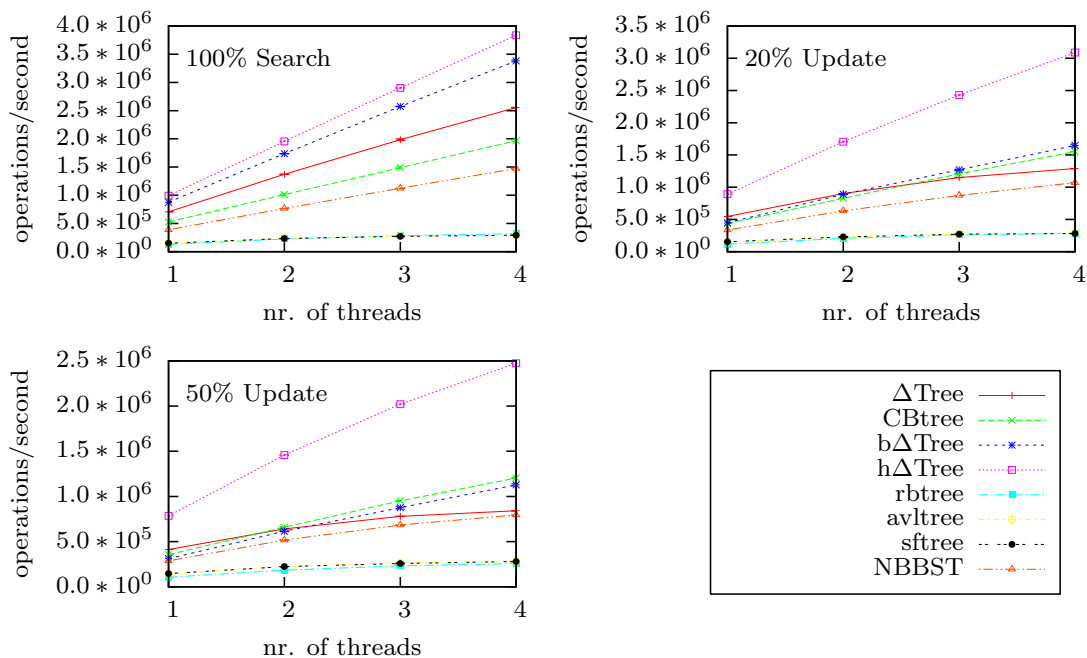}}
\caption{Performance comparison of the tested trees with 
\num[group-separator={,}]{2097151} initial members on an embedded ARM platform. The heterogeneous $\Delta$Tree (h$\Delta$Tree) is 100\% faster than the concurrent B-tree (CBTree) in the 50\%-update benchmark with 4 threads.}
\label{fig:perf-2097151}
\end{figure}

Among the new trees proposed in this paper, b$\Delta$Tree was up to 100\% faster 
than $\Delta$Tree for 100\%
search and was up to 20\% faster in 50\% updates. The h$\Delta$Tree
was up to 5\% faster than b$\Delta$Tree in 100\%
searching. However, in 20\% and 50\% updates,
h$\Delta$Tree was faster by up to 140\% than the b$\Delta$Tree (cf.
Figure \ref{fig:perf-4194303}). 

In comparison with the other trees, Figure 
\ref{fig:perf-2097151} and \ref{fig:perf-4194303} show that the heterogeneous $\Delta$Tree (h$\Delta$Tree)
was the fastest. The h$\Delta$Tree was up to 140\% faster than CBTree
in the 50\% update/8-thread setting on the Intel HPC platform, and  
was up to 100\% faster on the embedded ARM platform
in the 20\% update.  The b$\Delta$Tree's and CBTree's performance was trailing behind h$\Delta$Tree
in all test-cases.

\paragraph{Energy consumption results} \label{sec:energyEval}

\begin{figure}[!t] \centering
\resizebox{0.9\columnwidth}{!}{ \includegraphics{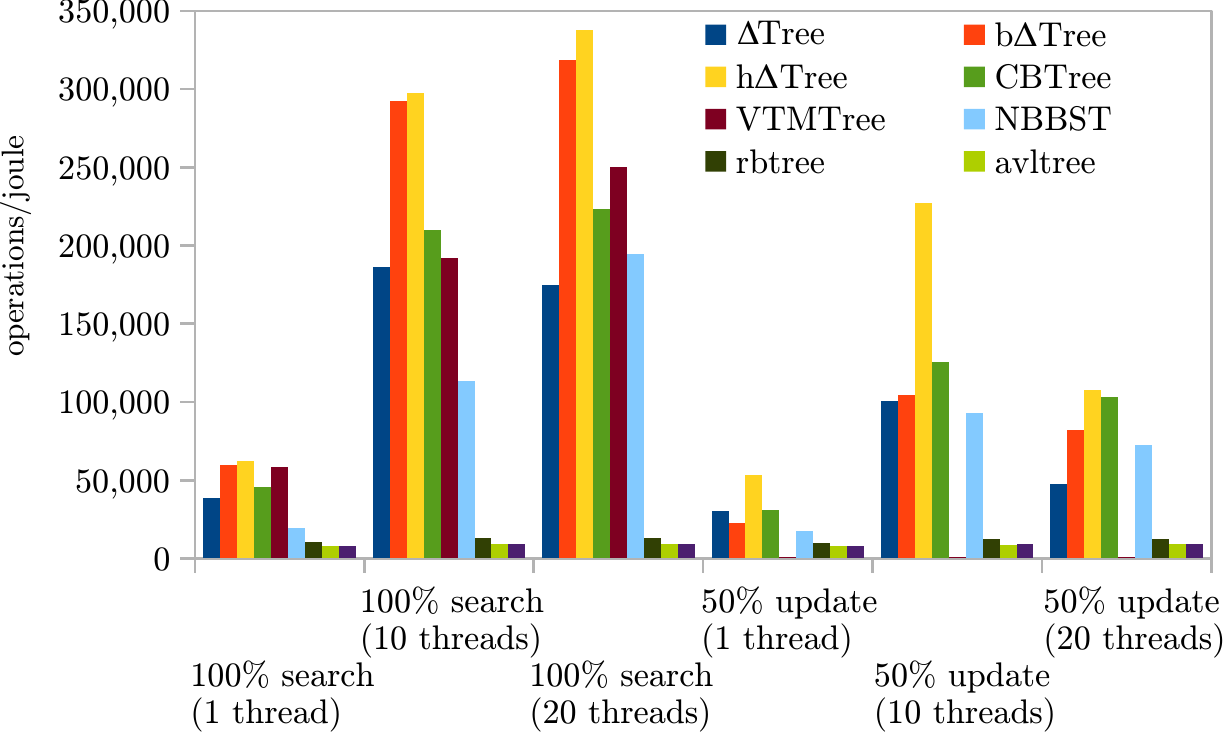}}
\caption{Energy efficiency comparison on an Intel HPC platform with two 10-core chips. On a single chip, the heterogeneous $\Delta$Tree (h$\Delta$Tree) was 80\% more energy efficient than the concurrent B-tree (CBTree) in the 50\%-update benchmark with 10 threads. The energy efficiency decreases for some trees in the 50\%-update benchmark with 20 threads (i.e., with 2 chips) because of the cache-coherence issue between two chips (cf. Section \ref{sec:perfeval}). (\textit{Standard errors} $\leq 0.4\%$).}
\label{fig:X86-power}
\end{figure}

\begin{figure}[!t] \centering
\resizebox{0.9\columnwidth}{!}{ \includegraphics{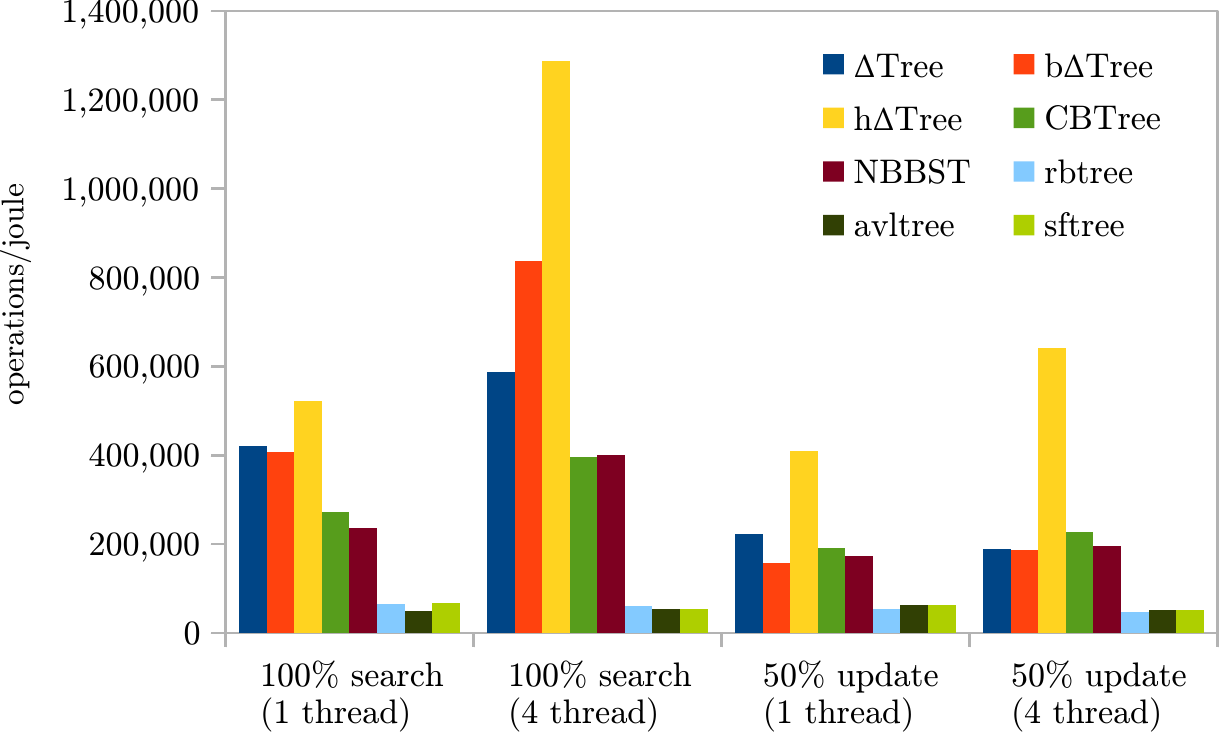}}
\caption{Energy efficiency comparison on an embedded ARM platform. The heterogeneous $\Delta$Tree (h$\Delta$Tree) was 220\% more energy efficient than the concurrent B-tree (CBTree) in the 100\%-search benchmark with 4 threads. (\textit{Standard errors} $\leq 0.27\%$)}
\label{fig:arm-power}
\end{figure}

Among the new trees proposed, balanced $\Delta$Tree (b$\Delta$Tree) was up to 85\% more energy efficient
than $\Delta$Tree (cf. Figure \ref{fig:X86-power}, 100\% search/20-thread setting).
On the other hand, heterogeneous $\Delta$Tree (h$\Delta$Tree) was up to 125\% more energy efficient
than b$\Delta$Tree (cf. Figure \ref{fig:X86-power}, 50\% update/10-thread setting). The results indicated that the 
increased locality by pointer-less $\Delta$Nodes in b$\Delta$Tree and
heterogeneous $\Delta$Nodes with 100\% more leaf nodes in h$\Delta$Tree significantly lowered the energy consumption of the trees.
 
In comparison with the other trees (e.g., CBTree) on the HPC platform (Figure
\ref{fig:X86-power}), h$\Delta$Tree was up to 33\% more energy
efficient in the 100\% search benchmark and 80\% more energy
efficient in the 50\% update/10-thread setting. 

Figure \ref{fig:arm-power} demonstrated the energy efficiency of concurrency-aware vEB-based trees on the embedded ARM platform. 
The most energy efficient tree was h$\Delta$Tree which was 
220\% more energy efficient than CBTree and NBBST
in the 100\% search/4-thread setting. In the 50\% update/4-thread setting, 
h$\Delta$Tree was the only one that managed to 
beat CBTree by 180\%.

\begin{figure}[!t] \centering
\resizebox{0.9\columnwidth}{!}{ \includegraphics{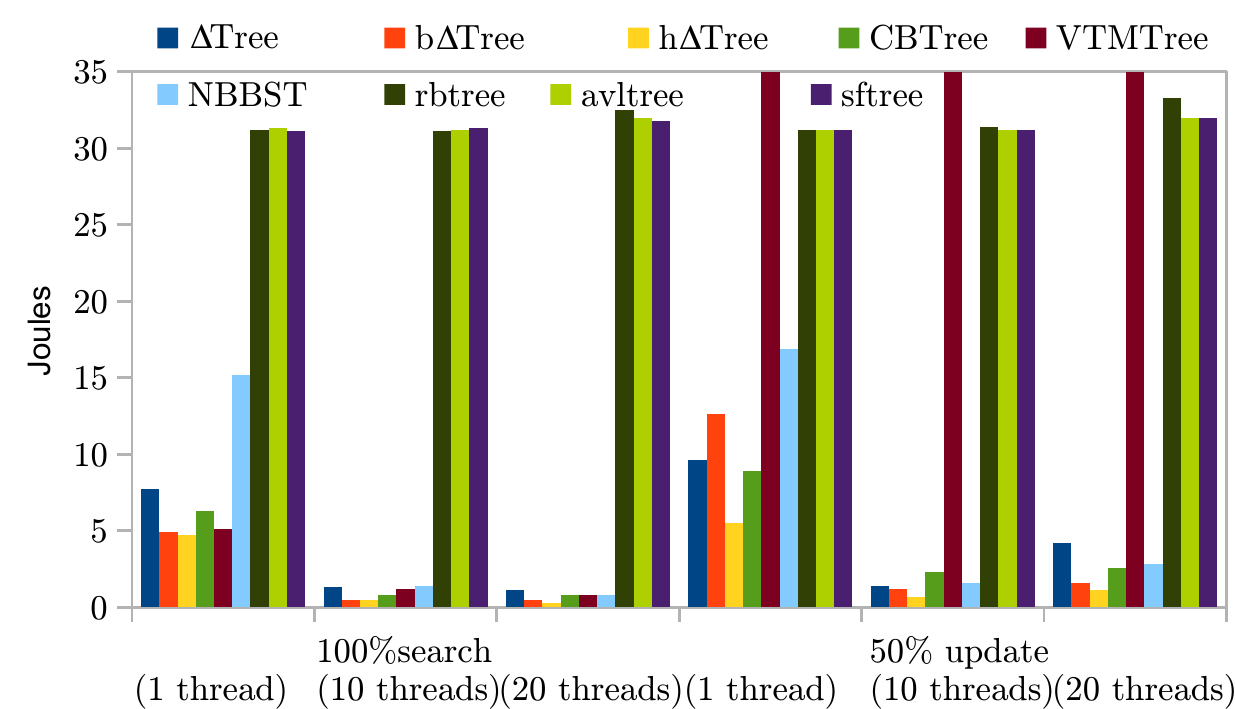}}
\caption{Memory (DRAM) energy consumption on Intel HPC platform. 
Measured using Intel PCM.}
\label{fig:mem-energy}
\end{figure}

\begin{figure}[!t] \centering
\resizebox{0.9\columnwidth}{!}{ \includegraphics{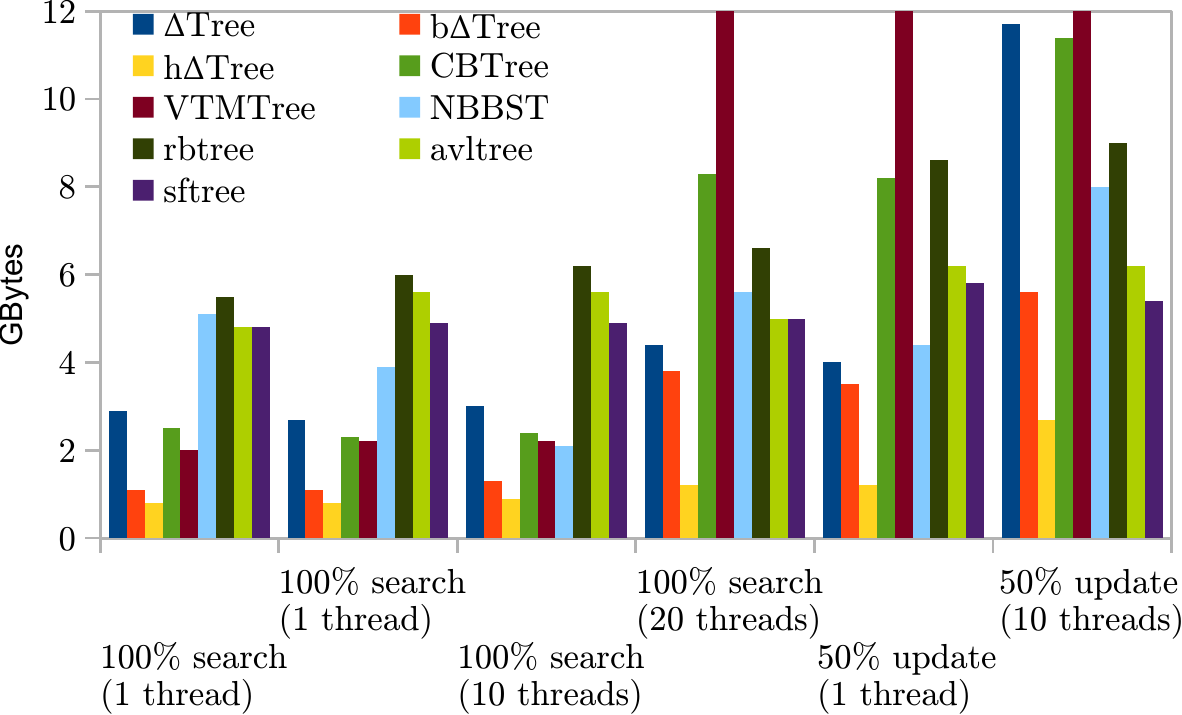}}
\caption{Amount of data transferred between RAM and CPU for Read + Write operations. Measured using Intel PCM.
Data transfer goes up considerably for some trees in 50\% update
on 20 threads (w/ 2 CPUs) because of the cache-coherence mechanism.}
\label{fig:mem-rw}
\end{figure}

A detailed breakdown of the energy-affecting factors was conducted on the Intel HPC platform.
Figure \ref{fig:mem-energy} shows the DRAM-only energy consumption 
whereas Figure \ref{fig:mem-rw} shows the amount of memory transfers between RAM and CPU for read-write operations.

Figure \ref{fig:mem-rw} shows that h$\Delta$Tree was the 
most efficient in terms of memory transfers between
the RAM and CPU compared to the other trees. 
As a result, h$\Delta$Tree had the lowest DRAM energy 
consumption (cf. Figure \ref{fig:mem-energy}).  

\paragraph{Energy and performance}
The experimental results on energy efficiency (Figure
\ref{fig:X86-power} and \ref{fig:arm-power})  and on performance (Figure \ref{fig:perf-4194303} and \ref{fig:perf-2097151}) 
showed that the usage of concurrency-aware vEB layouts 
was able to reduce the $\Delta$Trees' energy consumption and at the same time increase their performance.

If we closely compare side-by-side all of the results on the Intel HPC platform, a strong 
relationship between energy efficiency, performance and data transfer
can be concluded for this platform. As an example, 
in the 100\% search benchmark where concurrency-aware vEB trees are expected
to perform best,  the heterogeneous $\Delta$Tree was 50\% faster 
and 40\% more energy efficient
than CBTree. Figure \ref{fig:mem-rw} tells us that in this scenario, CBTree was transferring
almost $2\times$ more data compared to h$\Delta$Tree. 

However, there are other cases where the energy benefit exceeds 
the performance benefit, which we strongly suspect were attributed
to having a locality-aware data structure.
On the ARM platform, in the 100\% search/4 threads benchmark, 
h$\Delta$Tree was 220\% more energy efficient than 
CBTree while its performance was only 100\% faster.

Figures \ref{fig:perf-4194303} and \ref{fig:X86-power} also show that a less than 2x performance benefit
sometimes results in a 3x to 4x energy benefit for the STM-based trees on the HPC platform. The 
DRAM energy consumption of the STM-based trees
(e.g., AVLtree, RBtree, SFtree and VTMTree) was up to 10x higher than that of the other
trees (cf. Figure \ref{fig:mem-energy}), which partly explains the phenomenon.

Finally, the above evaluation relies on the Intel CPU performance counters or the Intel PCM and ARM
platform off-chip energy counters. Using more robust prediction
models for performance \cite{Chen:2014:EYC:2597479.2597500}
and energy \cite{McCullough:2011:EEM:2002181.2002193} might enable us to 
obtain more
insights that are useful in designing energy-efficient and highly concurrent data structures.

\paragraph{Multi-CPU coherency issue} \label{sec:cache-coher}
Figure \ref{fig:perf-4194303} shows some "dips" in performance
for several trees after passing the 8-thread mark on the Intel HPC platform. 
The same behavior was reflected in energy efficiency results (cf. 
Figure \ref{fig:X86-power}). In this chart, the same "dips" were also shown in the 50\% update/20 thread setting. 

Our experimental analysis revealed that these "dips" were caused by the cache coherence protocol in a multi-CPU system. The Intel HPC platform consists of 2 CPUs with 8-cores each. 
Therefore, if a tree
was exploiting data-locality in the cache of a single CPU, cache-coherence protocol did not involve the bus that connects CPUs and memory.  
Starting from 9 threads (or from 11 threads in the energy benchmarks), 2 CPUs
were used and therefore the cache-coherence protocol must use the bus to transfer data. 
A closer look at Figure \ref{fig:mem-rw} explains that 
the dips were also related to the increase of RAM-CPU data transfers, which
strengthens the cache-coherence theory.
Moreover, there were no cache-coherence issues 
in the ARM single CPU platform. Therefore, the "dips" did not appear in any results on the ARM platform.

\paragraph{Comparison with GCC std::set} \label{sec:worstCase}

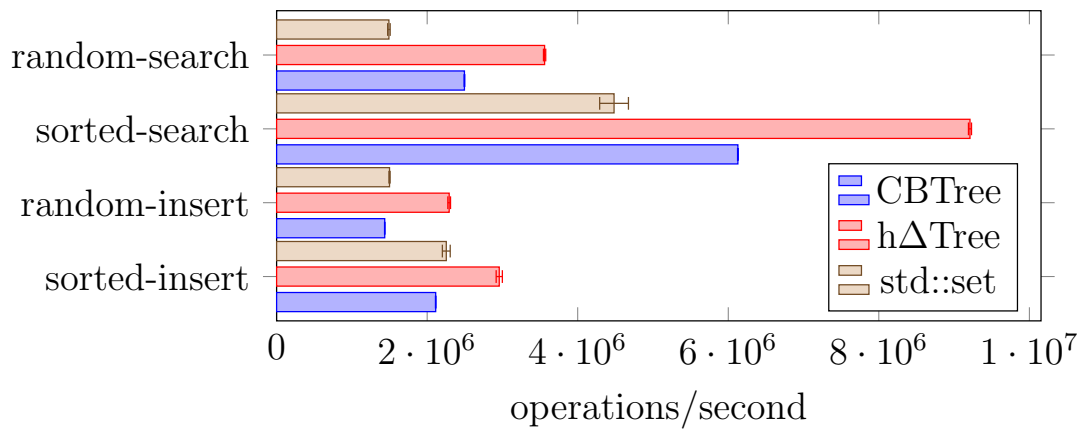
\begin{figure}[!t] \centering
\resizebox{0.9\columnwidth}{!}{ 
\begin{tikzpicture}
  \begin{axis}
    [
    width=10cm, height=5cm,
    enlarge y limits=0.2,
    scaled x ticks=false,
    xbar,
    xmin=0,
    xlabel={operations/second},
    symbolic y coords={sorted-insert,random-insert,sorted-search,random-search},
    bar width=6pt,
    ytick=data,
    legend pos=south east
    ]
    
\addplot+  [
	error bars/.cd, 
	x dir=both, 
	x explicit
]
table [x error=error, row sep=crcr] 
{
x        	  y             error \\
2113334.23	sorted-insert 	2888.0704221559\\
1436425.97	random-insert 		2303.2845403101\\
6127972.16	sorted-search 	4352.881447301\\
2495301.85	random-search		4215.9011296083\\
};

\addplot+ [
              	error bars/.cd, 
		x dir=both, 
		x explicit
		]
table [x error=error, row sep=crcr] 
{
x        	y                error \\
2957809.80	sorted-insert  	41063.0593075216 \\
2292632.53	random-insert  		17969.8233327317 \\
9211329.64	sorted-search 	20427.5043253486 \\
3560204.26	random-search 		15008.821675643 \\
};

\addplot+ [
              	error bars/.cd, 
		x dir=both, 
		x explicit
		]
table [x error=error, row sep=crcr] 
{
x        	y                error \\
2254212.20	sorted-insert  	52142.1560510366 \\
1500175.49	random-insert		7576.6687994266\\
4482888.90	sorted-search	191164.489145048 \\
1492020.76	random-search		16806.863869775 \\
};

\legend{CBTree, h$\Delta$Tree, std::set}  
\end{axis}
\end{tikzpicture}
}
\caption{Random and sorted insertions/searches of \num[group-separator={,}]{5000000} values on h$\Delta$Tree, CBTree and std::set\textless{int}\textgreater.}
\label{fig:worstcase}
\end{figure}

For $\Delta$Tree, inserting a sequence of increasing (sorted) numbers 
will result in a linked-list of $\Delta$Node.
The bottom-up insertions adopted by balanced $\Delta$Tree and subsequently 
heterogeneous $\Delta$Tree was supposed to prevent this (worst-case) from happening.

Therefore, we tested CBTree and h$\Delta$Tree for random and sorted number insertions. 
GCC standard library std::set\textless{int}\textgreater was also included as a baseline. For the sorted
inserts, 
starting with a blank tree, an increasing sequence of \num[group-separator={,}]{5000000} 
numbers starting from 1 were inserted into the tree using a single thread. The
same amount of random numbers was used instead for the random inserts. 
These tests were done on the same HPC platform described in Section \ref{sec:benchsetup}.
Since std::set is currently not supporting
concurrent operations, a concurrent insert test was deliberately left out.

The result in Figure \ref{fig:worstcase} showed that h$\Delta$Tree was 38\% faster than 
CBTree and 30\% faster than std::set in the sorted inserts. 
For random inserts, the h$\Delta$Tree was about 50\%
faster compared to CBTree and std::set. Further evaluation with \textit{perf} revealed
that in h$\Delta$Tree, random inserts triggers more branch misses (6.82\%) compared
to sorted inserts (1.04\%).

We also measured the time required to search all of the inserted values 
using the same way as they were inserted (random or sorted). 
In Figure \ref{fig:worstcase}, h$\Delta$Tree was up to 105\% faster in the random search
compared to std::set. This lead was increased to 130\% 
when searching the sorted sequence of inserted numbers.

The above results imply two things: 1) h$\Delta$Tree performed better in both of the worst and average
case of tree search/update compared to optimized B-tree and the widely used std::set; 
and 2) random value inserts and searches were not actually
benefit the B-tree based tree. Therefore, the benchmark setups
in Section \ref{sec:benchsetup} is assured to give a fair comparison of all the tested trees
performance. The reader is referred to \cite{HaTUTGRWA14, deltatreeTR2013} for more details about $\Delta$Trees.


\subsection{Concurrent Lock-Free Queues}
This section exposes the quality of the prediction of the model on
lock-free queues. We have used a library of well-known lock-free
implementations, and have predicted their throughput and energy
consumption thanks to the model explained in
Section~\ref{sec:cpu-model}, instantiated with the process exhibited
in Section~\ref{sec:cpu-model-inst}. In Section~\ref{sec:impl-desc},
we first give a brief description of the queue implementations, then
we show the experimental study in Section~\ref{sec:impl-xp}, that we
apply both on synthetic benchmarks and a more realistic application.

\subsubsection{Description of the Implementations}
\label{sec:impl-desc}

\readarray{alg}{MS&Val&TZ&Gid&caca&Hof&Moi}
\label{sec:NOBLE}
\paragraph{NOBLE}
Most of the implementations that we use are part of the NOBLE
library~\cite{Sundell02,Sundell08}.  The NOBLE library offers support
for non-blocking multi-process synchronization in shared memory
systems. NOBLE has been designed in order to: i) provide a collection
of shared data objects in a form which allows them to be used by
non-experts, ii) offer an orthogonal support for synchronization where
the developer can change synchronization implementations with minimal
changes, iii) be easy to port to different multi-processor systems,
iv) be adaptable for different programming languages and v) contain
efficient known implementations of its shared data objects. The
library provides a collection of the most commonly used data types.
The semantics of the components, which have been designed to be the
very same for all implementations of a particular abstract data type,
are based on the sequential semantics of common abstract data types
and adopted for concurrent use. The set of operations has been limited
to those which can be practically implemented using both non-blocking
and lock-based techniques. Due to the concurrent nature, also new
operations have been added, e.g.  Update which cannot be replaced by
Delete followed by Insert. Some operations also have stronger
semantics than the corresponding sequential ones, e.g. traversal in a
List is not invalidated due to concurrent deletes, compared to the
iterator invalidation in STL. As the published algorithms for
concurrent data structures often diverge from the chosen semantics, a
large part of the implementation work in NOBLE, besides from adoption
to the framework, also consists of considerable changes and extensions
to meet the expected semantics.

The various lock-free concurrent queue algorithms that we include in
this study are briefly described below.

\newcommand\walgmod[1]{\ema{(\protect\walg{#1})}}

\paragraph{Tsigas-Zhang \walgmod{2}}
Tsigas and Zhang~\cite{TsiZ01b} presented a lock-free extension of \cite{Lam83} for any
number of threads where synchronization is done both on the array elements and the shared
head and tail indices using \op{CAS}\footnote{The Compare-And-Swap
  (CAS) atomic primitive will update a given memory word, if and only if the word still
  matches a given value (e.g. the one previously read). CAS is generally available in
  contemporary systems with shared memory, supported mostly directly by hardware and in
  other cases in combination with system software.}, and the ABA problem is avoided by exploiting two (or
more) null values. We recall that the ABA problem is due to the
inability of \op{CAS} to detect concurrent changes of a memory word from a value (A) to
something else (B) and then again back to the first value (A). A {\em CAS} operation can
not detect if a variable was read to be A and then later changed to B and then back to A
by some concurrent processes. The {\em CAS} primitive will perform the update even though
this might not be intended by the algorithm's designer.
In \cite{TsiZ01b} synchronization is done both directly on the array
elements and the shared head and tail indices using \op{CAS}, thus supporting multiple producers
and consumers. Moreover, for lowering the memory
contention the algorithm alternates every other operation between scanning and updating
the shared head and tail indices.

\paragraph{Valois \walgmod{1}}
Valois \cite{Val94,Val95phd} makes use of linked list in his lock-free implementation
which is based on the \op{CAS} primitive. He was the first to present a lock-free
implementation of a linked-list. The list uses auxiliary memory cells between adjacent
pairs of ordinary memory cells. The auxiliary memory cells were introduced to provide an
extra level of indirection so that normal memory cells can be removed by joining the
auxiliary ones that are adjacent to them. His design also provides explicit cursors to
access memory cells in the list directly and insert or delete nodes on the places the the
cursors point to.

\paragraph{Michael-Scott \walgmod{0}}
Michael and Scott~\cite{lf-queue-michael} presented a lock-free queue that is more
efficient, synchronizing via the shared head and tail pointers as well as via the next
pointer of the last node.  Synchronization is done via shared pointers indicating the
current head and tail node as well via the next pointer of the last node, all updated
using \op{CAS}.  The tail pointer is then moved to point to the new item, with the use
of a \op{CAS} operation. This second step can be performed by the thread invoking the
operation, or by another thread that needs to help the original thread to finish before it
can continue. This helping behavior is an important part of what makes the queue
lock-free, as a thread never has to wait for another thread to finish.  The queue is fully
dynamic as more nodes are allocated as needed when new items are added. The original
presentation used unbounded version counters, and therefore required double-width \op{CAS}
which is not supported on all contemporary platforms. The problem with the version
counters can easily be avoided by using some memory management scheme as
e.g. \cite{Mic04b}.

\paragraph{Moir \etal \walgmod{6}}
Moir \etal~\cite{MoirNSS:2005:elim-queue} presented an extension of the Michael and Scott
\cite{lf-queue-michael} lock-free queue algorithm where elimination is used as a back-off
strategy to increase scalability when contention on the queue's head or tail is noticed
via failed \op{CAS} attempts. However, elimination is only possible when the queue is
close to empty during the operation's invocation.

\paragraph{Hoffman-Shalev-Shavit \walgmod{5}}
Hoffman \etal~\cite{DBLP:conf/opodis/HoffmanSS07} takes another approach in their design
in order to increase scalability by allowing concurrent \op{Enqueue} operations to insert
the new node at adjacent positions in the linked list if contention is noticed during the
attempted insert at the very end of the linked list. To enable these ``baskets'' of
concurrently inserted nodes, removed nodes are logically deleted before the actual removal
from the linked list, and as the algorithm traverses through the linked list it requires
stronger memory management than \cite{Mic04b}, such
as~\cite{DBLP:journals/tpds/GidenstamPST09} or \cite{HerLMM:2005:NMM} and a strategy to
avoid long chains of logically deleted nodes.

\paragraph{Gidenstam-Sundell-Tsigas \walgmod{3}}
Gidenstam \etal~\cite{Gidenstam10:OPODIS} combines the efficiency of using arrays and the
dynamic capacity of using linked lists, by providing a lock-free queue based on linked
lists of arrays, all updated using \op{CAS} in a cache-aware manner. In resemblance to
\cite{Lam83,GiaMoVa:2008:ff-queue,TsiZ01b} this algorithm uses arrays to store
(pointers to) the items, and in resemblance to \cite{TsiZ01b} it uses \op{CAS} and two
null values. Moreover, shared indices \cite{GiaMoVa:2008:ff-queue} are avoided and
scanning \cite{TsiZ01b} is preferred as much as possible. In contrast to
\cite{Lam83,GiaMoVa:2008:ff-queue,TsiZ01b} the array is not static or cyclic,
but instead more arrays are dynamically allocated as needed when new items are added,
making the queue fully dynamic.






\subsubsection{Experiments: Predictions and Measurements}
\label{sec:impl-xp}

The legend depicted in Figure~\ref{fig.key} lists the implementations
that are compared, and will be used throughout this section.

\readarray{alg}{MS&Val&TZ&Gid&caca&Hof&Moi}

\begin{figure}
\begin{center}
\framebox{%
\begin{tikzpicture}
\matrix (m) [matrix of math nodes, ampersand replacement=\&, nodes={anchor=base west},
        row sep=0cm,column sep=.1cm] {
\vcenter{\hbox{\includegraphics[width=1cm]{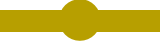}}} \& \hspace{-.2cm}\walg{1}\text{ \cite{Val94}} \& \hspace{.0cm} \&
\vcenter{\hbox{\includegraphics[width=1cm]{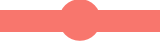}}} \& \hspace{-.2cm}\walg{0}\text{ \cite{lf-queue-michael}}\& \hspace{.0cm} \&
\vcenter{\hbox{\includegraphics[width=1cm]{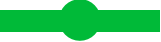}}} \& \hspace{-.2cm}\walg{2}\text{ \cite{TsiZ01b}} \\
\vcenter{\hbox{\includegraphics[width=1cm]{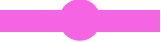}}} \& \hspace{-.2cm}\walg{6}\text{ \cite{MoirNSS:2005:elim-queue}} \& \&
\vcenter{\hbox{\includegraphics[width=1cm]{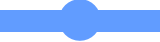}}} \& \hspace{-.2cm}\walg{5}\text{ \cite{DBLP:conf/opodis/HoffmanSS07}} \& \&
\vcenter{\hbox{\includegraphics[width=1cm]{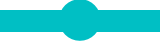}}} \& \hspace{-.2cm}\walg{3}\text{ \cite{Gidenstam10:OPODIS}}\\
};
\end{tikzpicture}
}
\framebox{%
\begin{tikzpicture}
\matrix (m) [matrix of math nodes, nodes={},ampersand replacement=\&,
        row sep=0cm,column sep=.1cm] {
\vcenter{\hbox{\includegraphics[width=1cm]{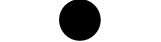}}} \& \text{Actual} \& \hspace{.5cm} \&
\vcenter{\hbox{\includegraphics[width=1cm]{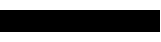}}} \& \text{Prediction} \\
};
\end{tikzpicture}
}
\end{center}
\caption{Key legend of the graphs\label{fig.key}}
\end{figure}


\paragraph{Synthetic Benchmark}

\subparagraph{Throughput}
\label{sec:thput-res}

\begin{figure}[h!p]
\begin{center}
\subfloat[$\protect\pwd=7$\label{fig:thr-enq-7}]{\includegraphics[height=.45\textheight]{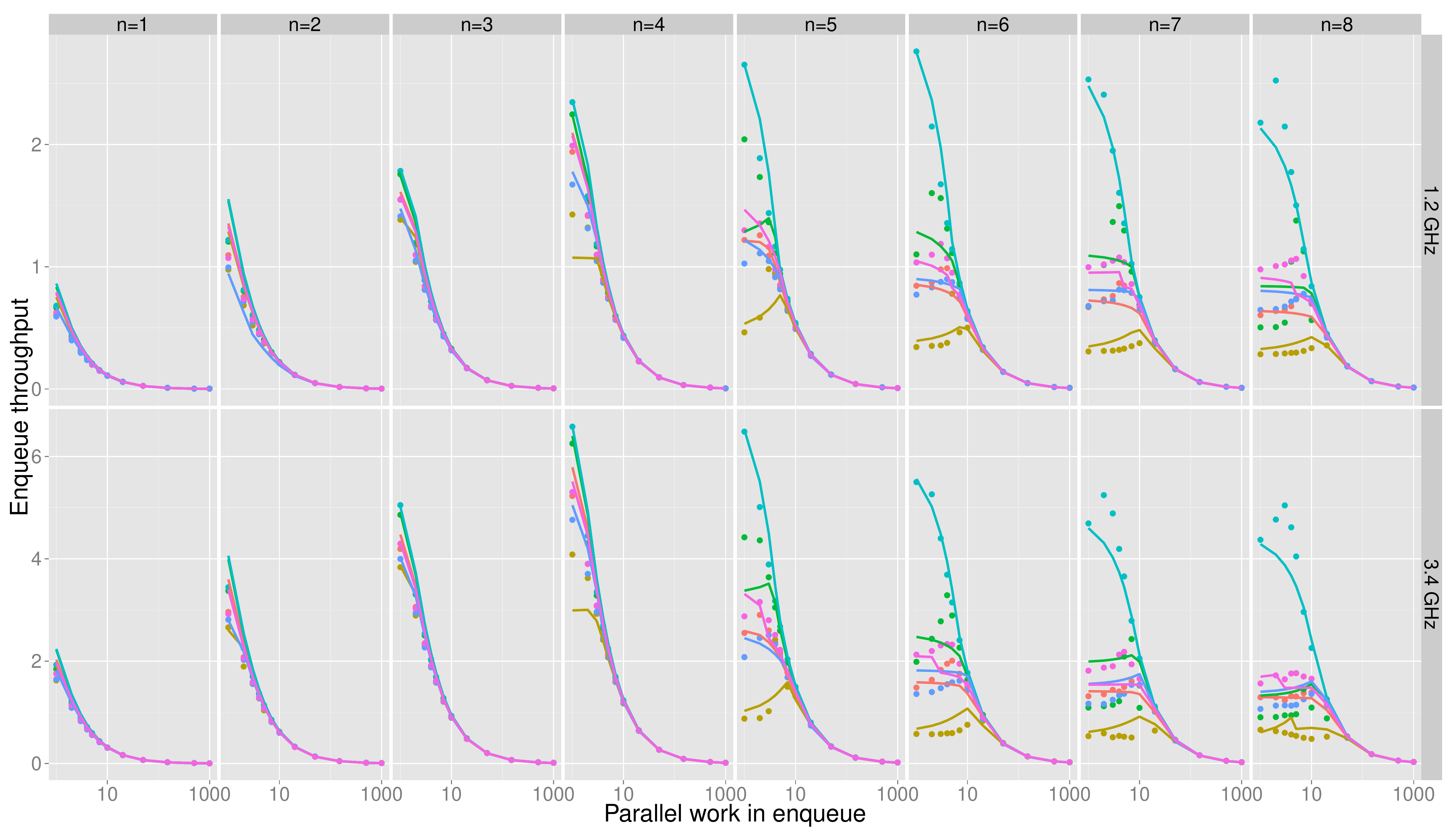}}

\subfloat[$\protect\pwd=50$\label{fig:thr-enq-50}]{\includegraphics[height=.45\textheight]{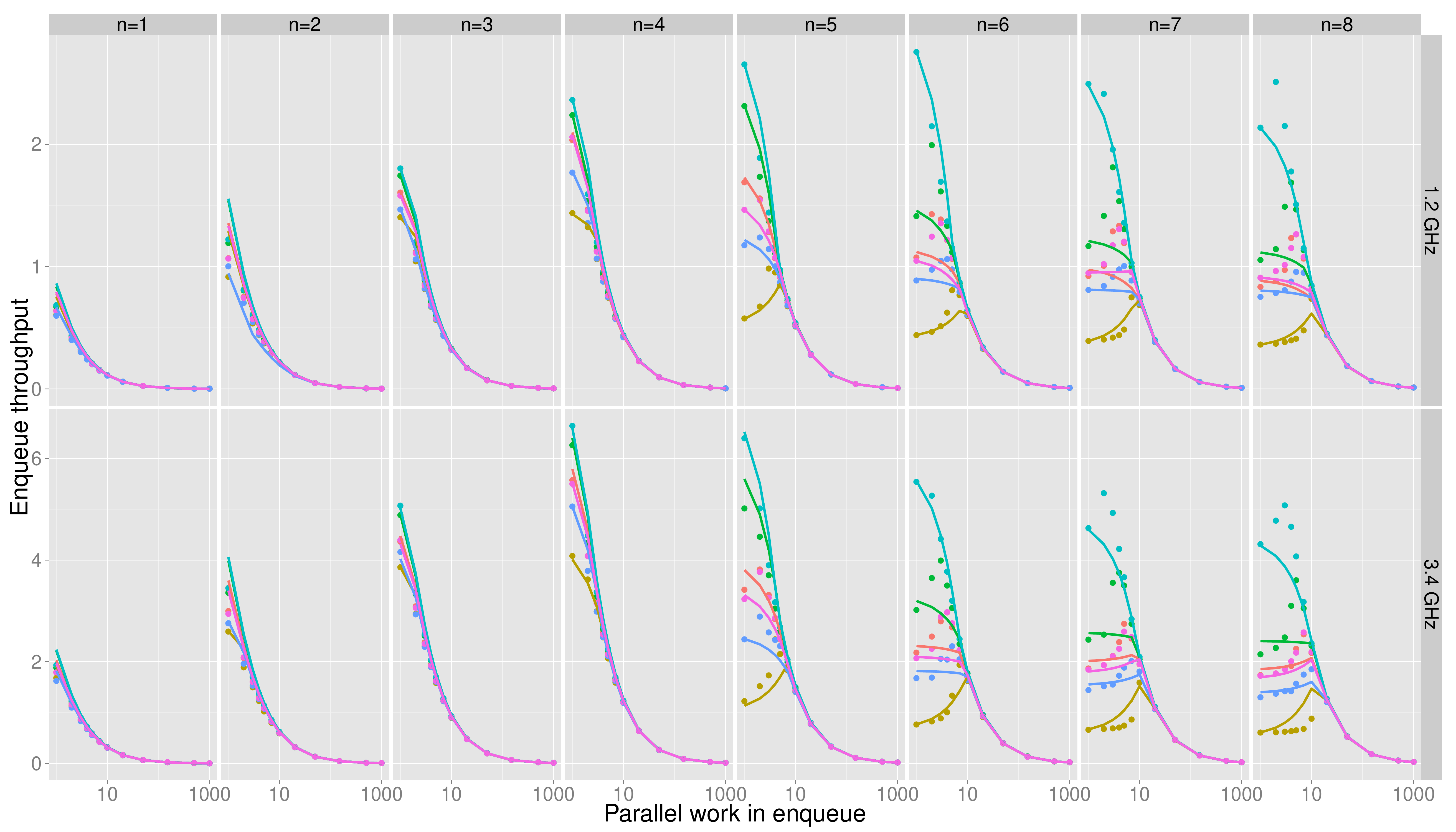}}
\end{center}
\caption{Enqueue throughput\label{fig:thr-enq}}
\end{figure}

\begin{figure}[t!h]
\incgrapr{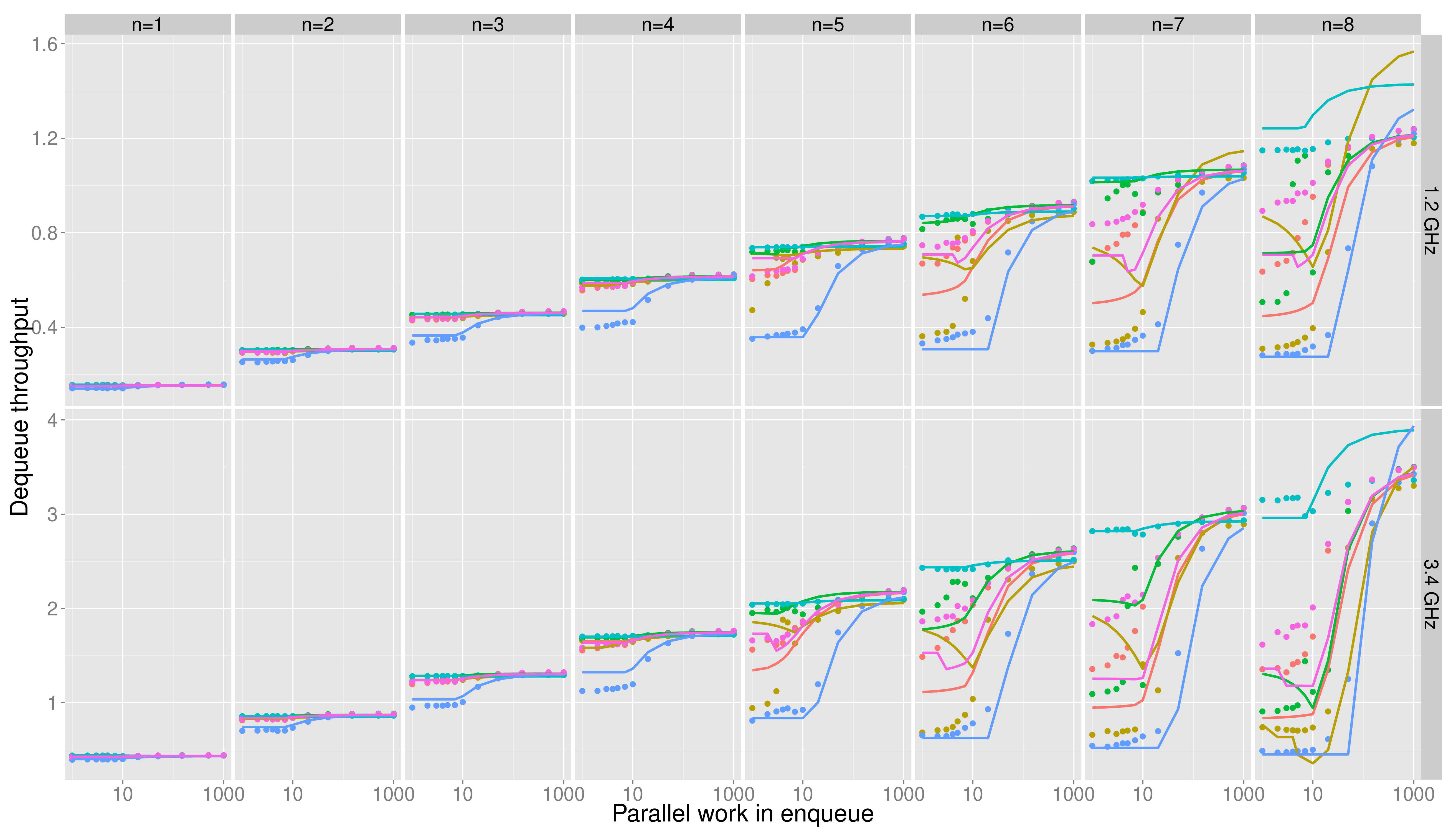}
\caption{Dequeue throughput with $\protect\pwd=7$\label{fig:thr-deq-7}}
\end{figure}

The throughput predictions are plotted in Figure~\ref{fig:thr-enq} for
the enqueuers, and in Figure~\ref{fig:thr-deq-7} for the dequeuers.
Points are measurements, while
lines are predictions. We will follow this rule for all comparisons
between prediction and measurement. 
In the actual execution, the queue goes through a transient state when the amount of work
in the \ps is near the critical point, but the prediction is not so far from the actual
measurements, as illustrated in Figure~\ref{fig:thr-enq}. Under
intra-contention, some of the curves get flat, since only one thread can be succeeding
at the same time, according to the definition of the \rl. Some curves even
decrease because the successful one is stalled by other failing ones due to serialization
of the atomic primitives, namely expansion. The slope presumably indicates the density of
atomic primitives in \rls which depends on the algorithm.

The comparison of Figures~\ref{fig:thr-enq-7} and~\ref{fig:thr-enq-50} illustrates the
impact of inter-contention. A decrease of the highest point of \thre, due to an increase
of \cwe, can be observed for the more inter-contended case. When \cwe increases, some
critical points shift slightly towards the right as the intra-contention starts with a
larger \pwe. In Figure~\ref{fig:thr-deq-7}, decomposition of \thrd is apparent. When
enqueue rate is low, \ie when \pwe is high, \thrd is ruled by \thrde due to majority
of \nulle dequeues, and it tends towards \thrdne when the enqueue rate increases.


\subparagraph{Power}

\begin{figure}[h!p]
\begin{center}
\subfloat[$\protect\freq=\ghz{3.4}$\label{fig:mem-pow-34}]{\includegraphics[height=.45\textheight]{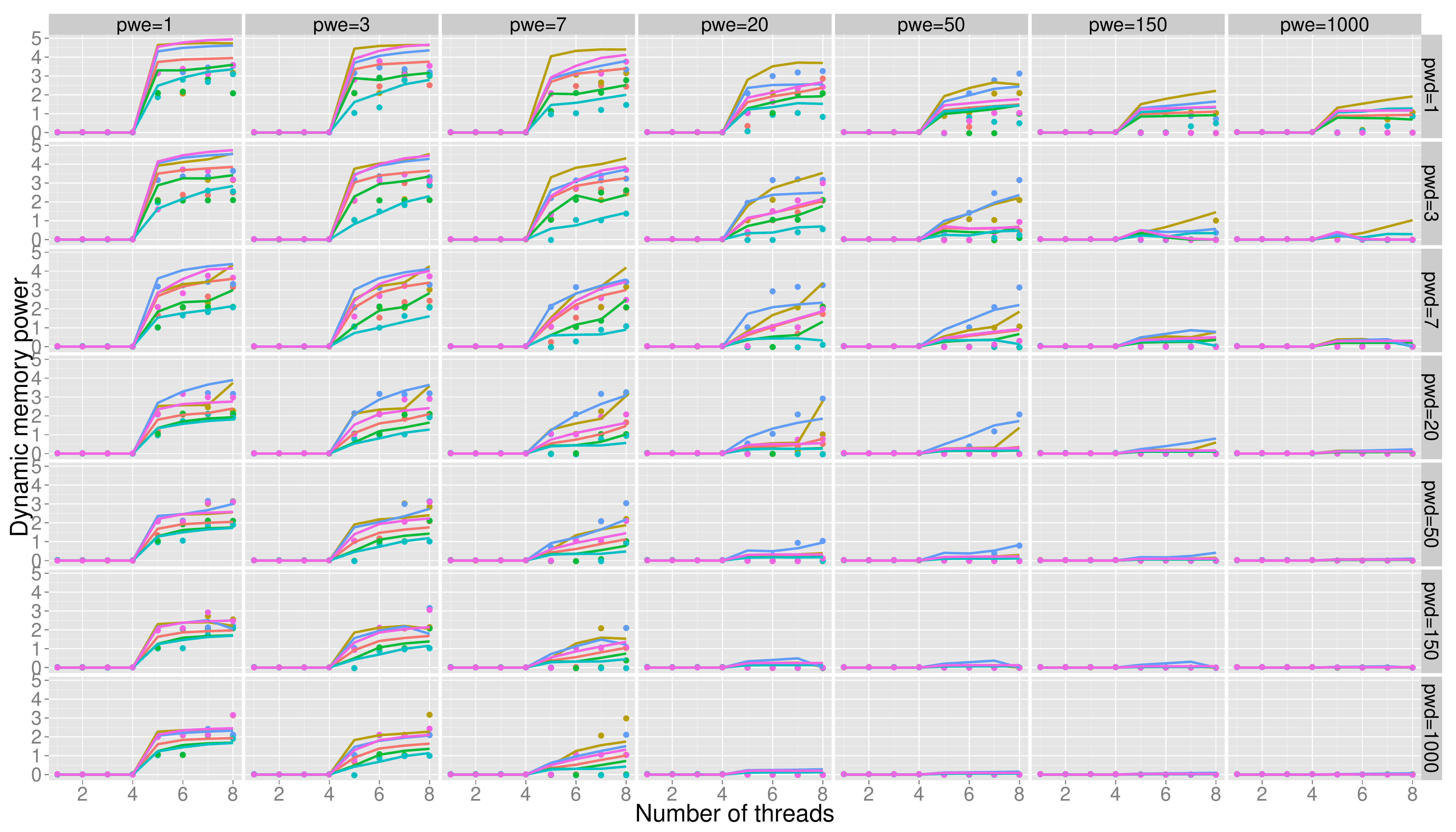}}

\subfloat[$\protect\freq=\ghz{1.2}$\label{fig:mem-pow-12}]{\includegraphics[height=.45\textheight]{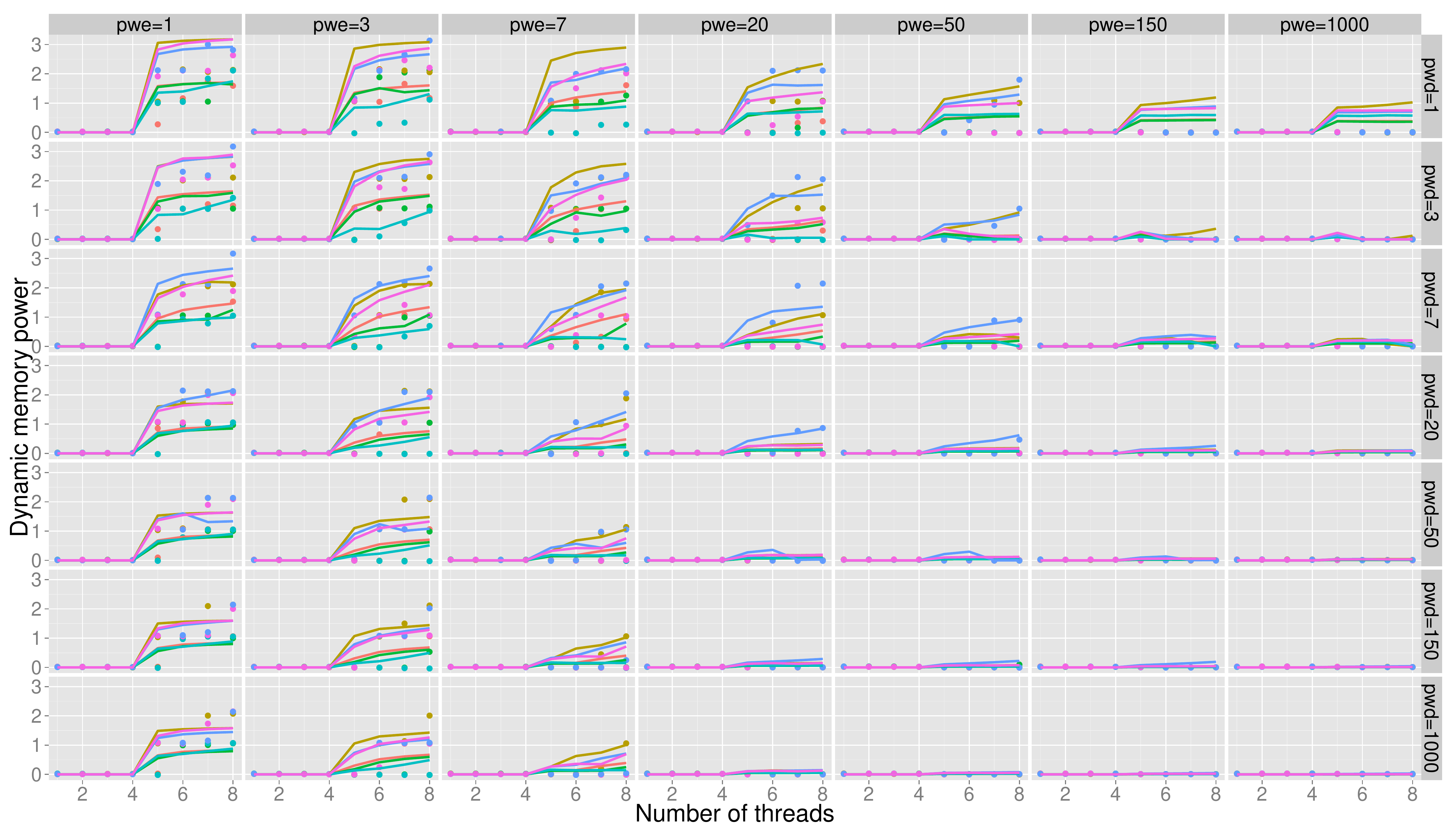}}
\end{center}
\caption{Dynamic memory power\label{fig:mem-pow}}
\end{figure}

\begin{figure}[h!p]
\begin{center}
\subfloat[Dynamic uncore power at $\protect\freq=\ghz{3.4}$\label{fig:unc-pow}]{\includegraphics[height=.45\textheight]{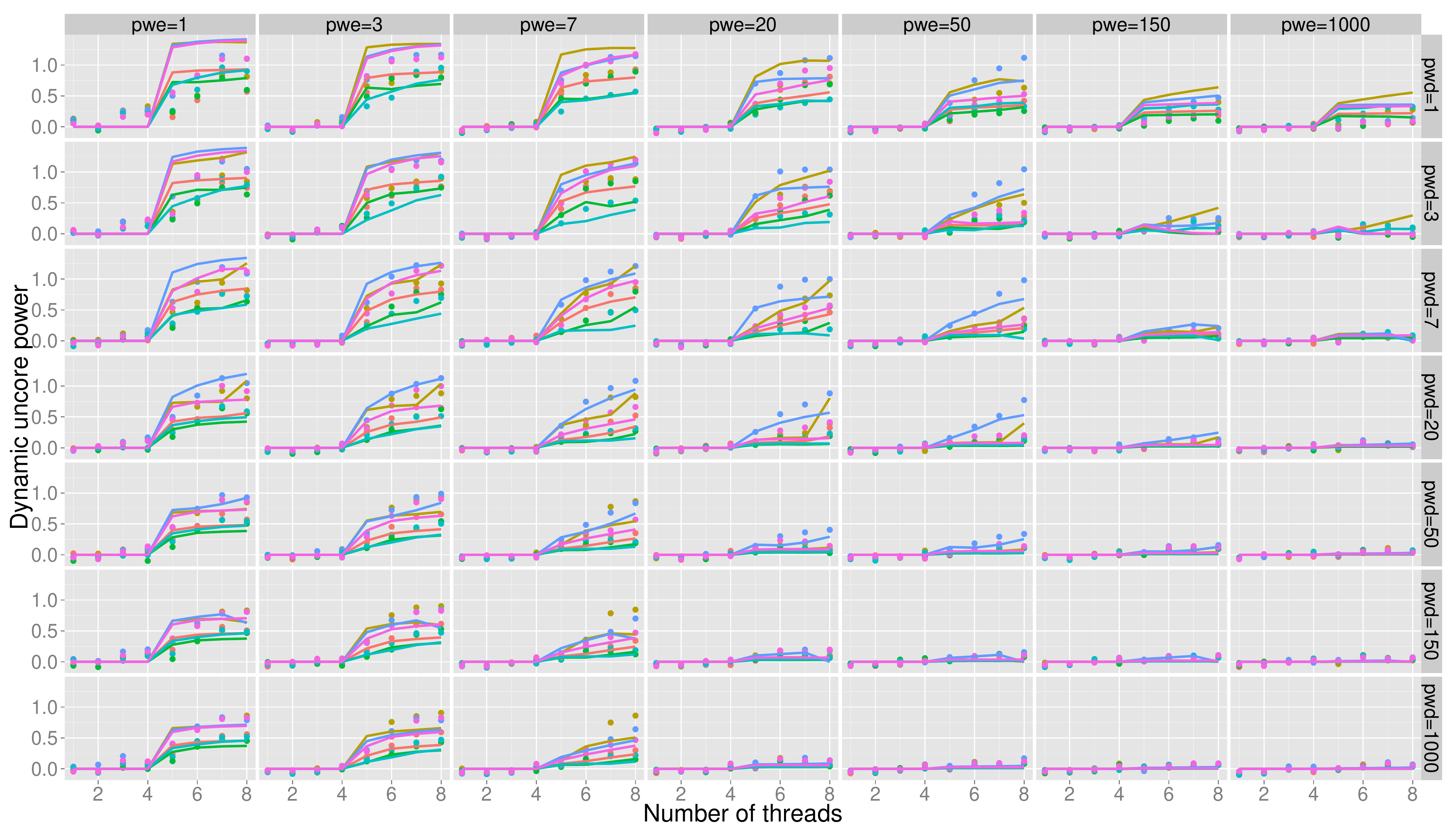}}

\subfloat[Dynamic CPU power at $\protect\freq=\ghz{1.2}$\label{fig:cpu-pow}]{\includegraphics[height=.45\textheight]{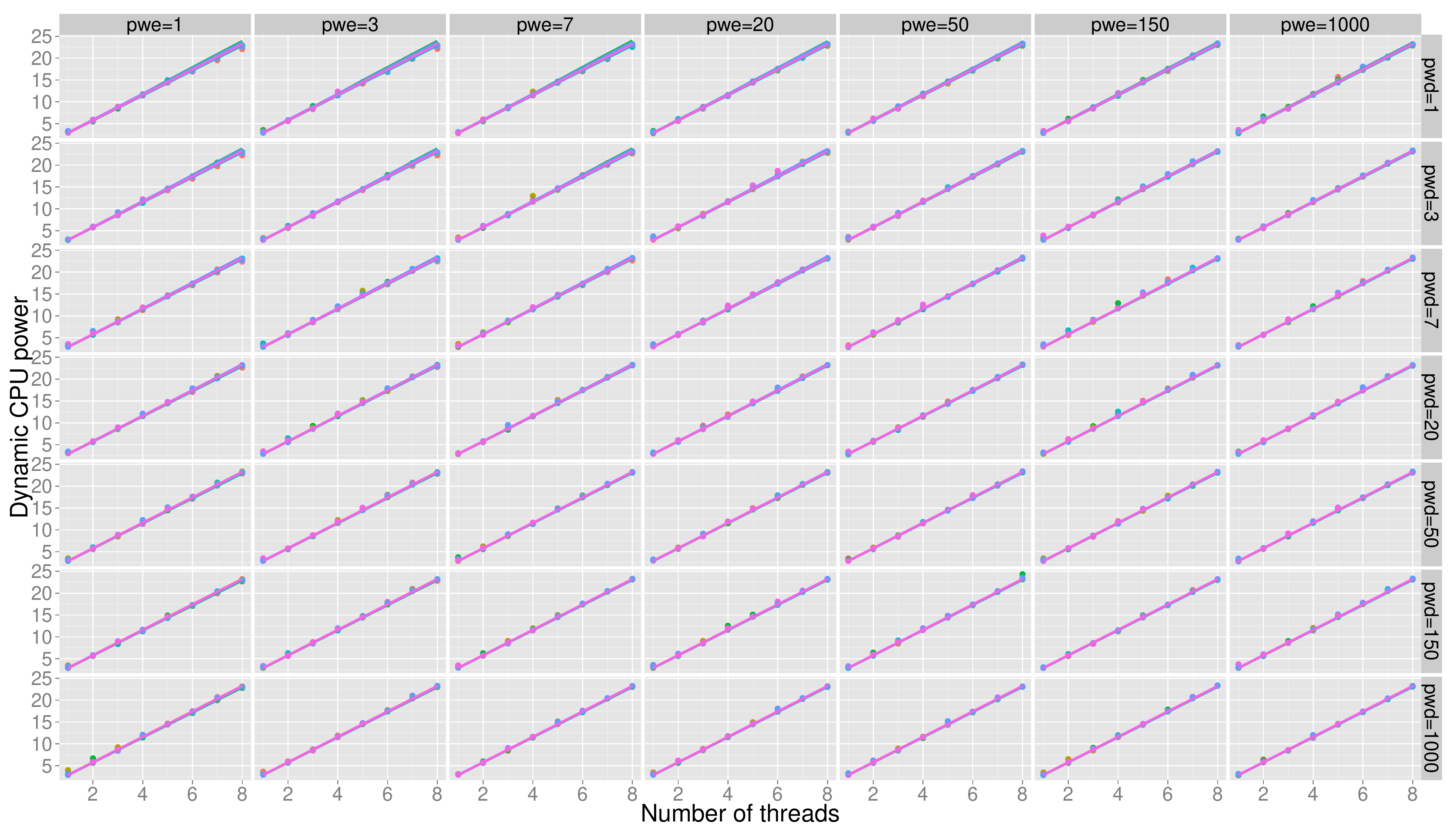}}
\end{center}
\caption{Dynamic uncore and CPU power}
\end{figure}

\begin{figure}[h!]
\incgrapr{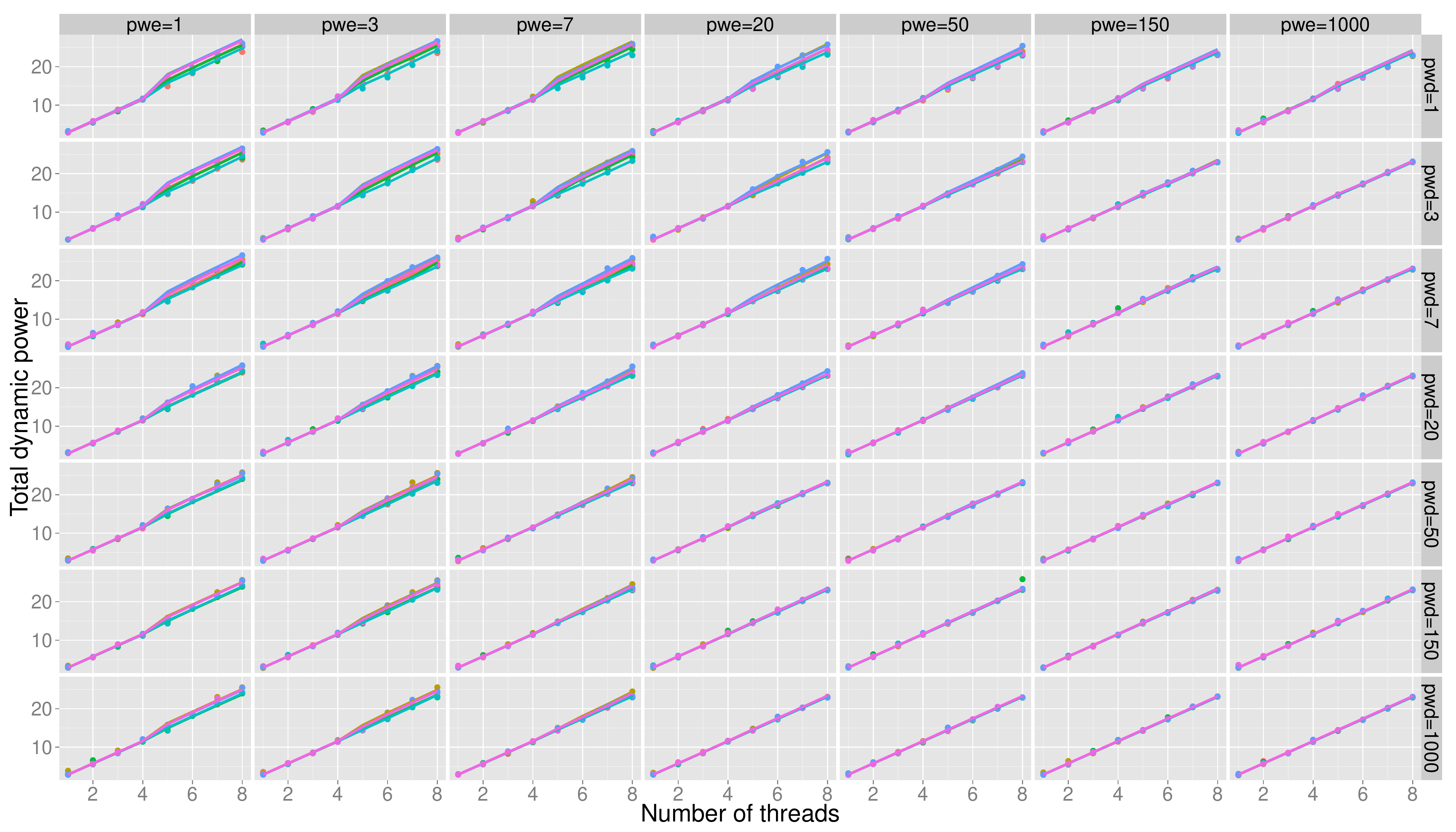}
\caption{Sum of Dynamic Powers at $\freq=\ghz{1.2}$\label{fig:tot-pow}}
\end{figure}

The prediction and measurements, regarding power, are plotted in Figures~\ref{fig:mem-pow},%
~\ref{fig:unc-pow},~\ref{fig:cpu-pow} and~\ref{fig:tot-pow}, where we observe that the most
significant differences lie in the dynamic memory power.
The differences in CPU power are almost invisible, since the dynamic power of the \pss
(composed of \paus{}s instructions) is very close to the dynamic power of the \rls.
We remark some steps in the measured memory power, but we prefer to
keep a continuous estimate.

As the \rl, which is particular to each implementation, is
mainly composed of memory
operations, the main difference between the various implementations in terms of power
occurs in the dynamic memory power, which we represent in Figure~\ref{fig:mem-pow-34} (legend
is in Figure~\ref{fig.key}).
Overall, the prediction reacts correctly to the variations of \ps sizes, and some
specifics of the algorithms are caught, \eg \walg{5} detached from the others when
$\pwe=50$ or \walg{3} mostly well-predicted both absolutely and relatively as the less
power-dissipating implementation.
One can observe once again the asymmetry between enqueue and dequeue operations by
comparing the power values at $(\pwd,\pwe)=(2,1000)$ and $(1000,2)$; this asymmetry is
predicted by the model, with a lower impact though.


\subparagraph{Energy per Operation}
\label{sec:energy}






\begin{figure}[h!p]
\begin{center}
\subfloat[$\protect\freq=\ghz{3.4}$\label{fig:epo-34}]{\includegraphics[height=.45\textheight]{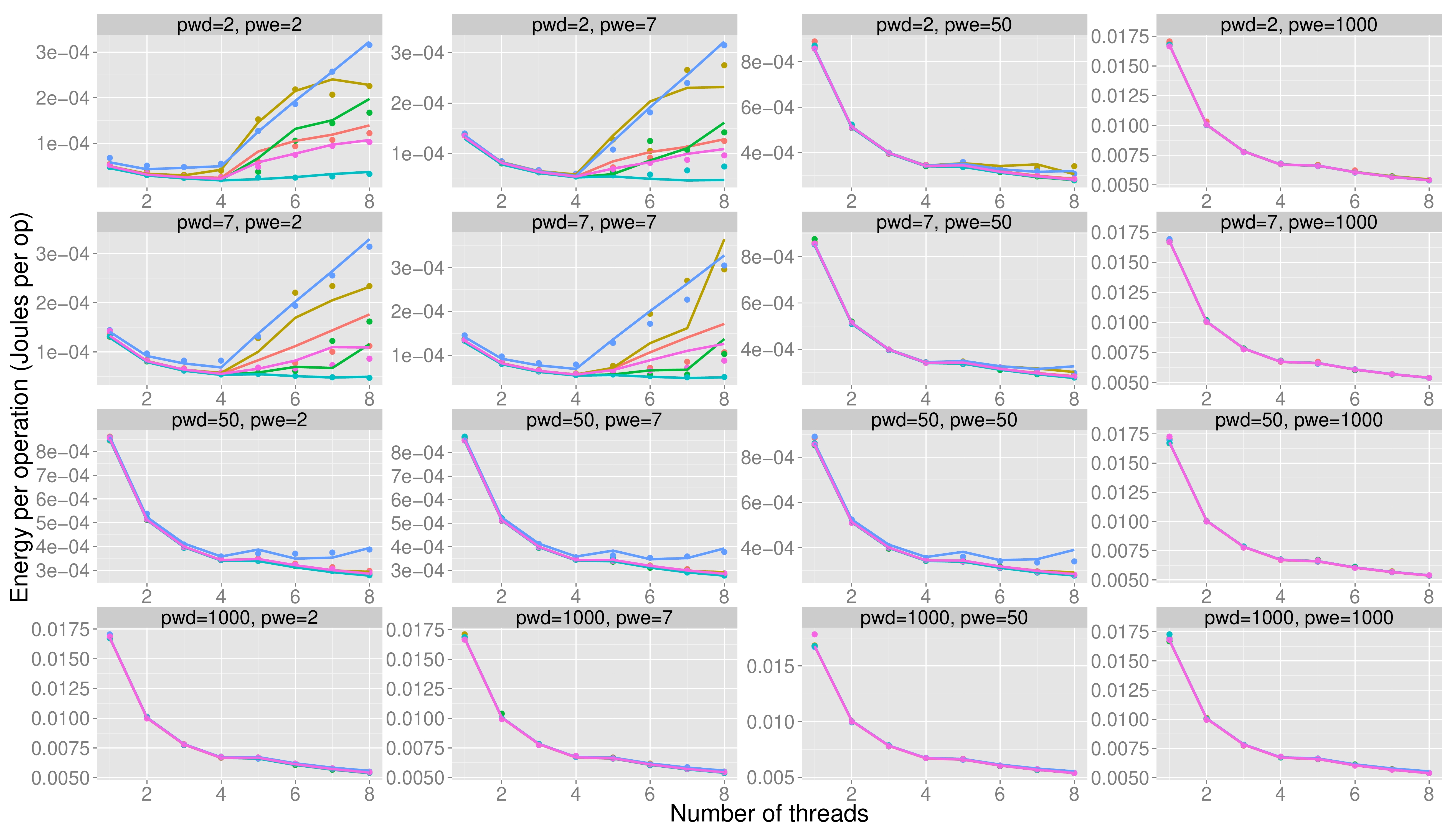}}

\subfloat[$\protect\freq=\ghz{1.2}$\label{fig:epo-12}]{\includegraphics[height=.45\textheight]{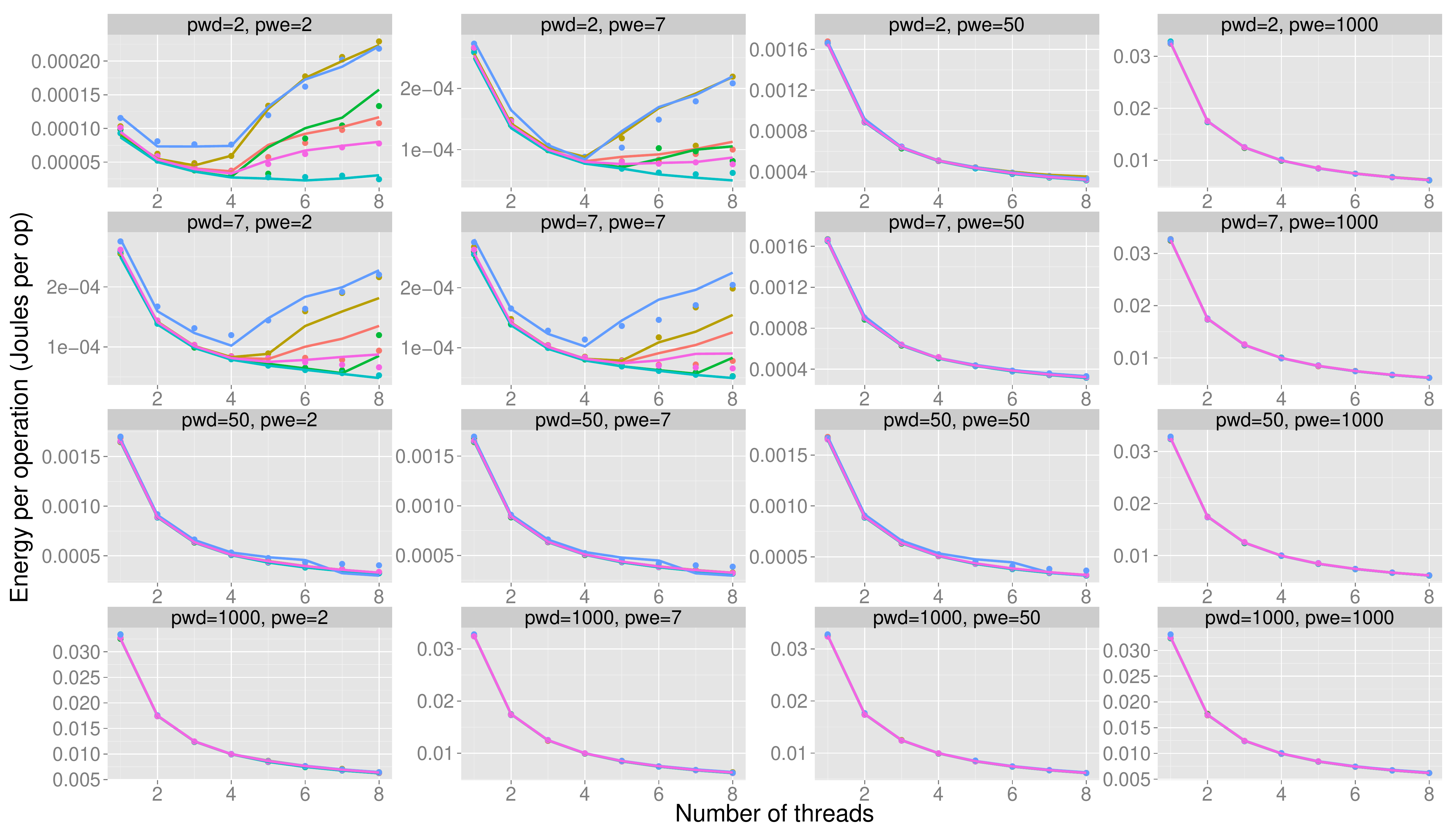}}
\end{center}
\caption{Energy per operation\label{fig:epo}}
\end{figure}


In Figure~\ref{fig:epo} is represented the energy per operation.
Overall we observe that the successful operations (dequeue of a non-\nulle item) are
cheaper and cheaper when the number of threads is increasing on the same socket: the cost
of turning the machine on is made profitable by an increase in performance. However, under
high-contention, the lack of performance improvement while increasing the number of cores
makes the use of supplementary cores useless. The inefficiency of adding cores is even more
apparent when cores are spread across the sockets. In this case, under high-contention,
performance could even be degraded by the implication of new cores, then, as performance
decreases and power increases, the energy per operation dramatically increases.


\paragraph{Towards Realistic Applications: Mandelbrot Set Computation} \label{sec:mandelbrot}


The performance and energy behavior of an application using a lock-free
queue depends on both the application specific code and the
implementation of the data structure.  For applications
where the queue is used in a steady state manner,
predictions can be made using the model instantiated with the synthetic benchmark,
combined with information about the behavior of the application specific
code. What is needed is:
\begin{itemize}
\item The size of the parallel work part of the application, both for
  enqueuers and dequeuers. These may be distributions rather than
  single values.
\item The dynamic power for these parts (as it may differ from
  that of the parallel work in the synthetic benchmark).
\end{itemize}


\subparagraph{Description of Mandelbrot Set Application}

As a case-study
we have used an existing application\footnote{Previously used for evaluation
  in~\cite{Sundell11}.}  that computes and renders an $8192\times 8192$ pixel
image of the Mandelbrot
set~\cite{Man80} in parallel using the producer/consumer pattern. The program uses a
concurrent queue to communicate between two major phases:

\begin{itemize}
\item Phase 1 consists of computing the number (with a maximum of 255)
  of iterations for a given set of points within a chosen region of
  the image. The results for each region together with its coordinates
  are then enqueued.
\item Phase 2 consists of, for each region dequeued from the
  queue, computing the RGB values for each contained point and
  draw these pixels to the resulting image.
  The colors for the
  corresponding number of iterations are chosen according to a rainbow
  scheme, where low numbers are rendered within the red and high
  numbers are rendered within the violet spectrum.
  
\end{itemize}

Half of the threads perform phase 1 and the rest perform phase 2.
The size of each square region is chosen to be one of $16 \times 16$,
$8\times 8$, $4\times 4$, or $2\times 2$ pixels which also determines
the amount of work to perform per queue operation and, hence, the level of
contention.
Similarly to the synthetic benchmark, the application uses a dense pinning strategy,
pinning producer/consumer pairs to consecutive pairs of cores.

This is just one of many possible ways to divide the work and pin threads,
it remains as future work to explore other ways.




\subparagraph{Mandelbrot Prediction}


\newcommand{\mantime}[1]{\ema{\mathit{Time_{#1}}}}
\newcommand{\size}{\ema{\mathit{size}}}
\newcommand{\Mpwd}{\gendu{\pw}{d}{}}
\newcommand{\Mpwe}{\gendu{\pw}{e}{}}
\newcommand{\Msize}{\gendu{\size}{}{i}}
\newcommand{\Mpweb}{\gendu{\pw}{e}{i}}
\newcommand{\thrbin}{\gendu{\thr}{}{i}}
\newcommand{\thrparam}{\gendu{\thr}{}{}(\Mpwd,\Mpweb)}

\newcommand{\powparam}{\powi_{}^{(X)}(\Mpwd,\Mpweb)}
\newcommand{\powbin}{\powi_{i}^{(X)}}

There are two main differences between the Mandelbrot application
and the synthetic benchmark: (i)~the instructions in the parallel
section differ; and (ii)~the size of the parallel section for producers
varies in Mandelbrot.

Firstly, we need to measure the CPU power dissipation for Mandelbrot;
we cannot expect to be able to predict the power dissipation
of any application that uses a queue without having any knowledge about
the power characteristics of the application. 
In contrast, memory power dissipation for the computation intensive Mandelbrot parallel
section is negligible in comparison to queue operations; hence, 
the dynamic memory power that we have measured and extrapolated in 
the synthetic benchmark is unchanged.

Secondly, Mandelbrot
provides a variety of producer parallel works. To
deal with this, 
the pixel region is decomposed row-wise in an interleaved manner among
producer threads. This decomposition leads to long enough execution intervals
in which the parallel sections of the producer threads are similar and constant.
This is
due to the computationally expensive pixels belonging to the Mandelbrot set being concentrated
together in the center of the domain and surrounded by cheaper pixels which diverge
quickly. This characteristic is congruent with our model where the data structure is
used in a steady state manner. Thus, predictions can be made using the
instantiated model over a linear combination of execution intervals.

We measure the latency of the computation intensive producer and consumer parallel works
for each frequency and contention level ($2\times 2$, $4\times 4$, $16\times 16$).
For this process, we make use of CPUID, RDTSC and RDTSCP instructions as specified in~\cite{intel-bench}. 
The distribution of parallel works reveals that there are two main groups for producers, that corresponds
to regions belonging to the Mandelbrot set or not. Concerning $2\times 2$ contention, due to the wide distribution, we gather
 the parallel works into bins of width $10$ pauses; the number of elements in the $i^{\mathrm{th}}$ bin
is then denoted by \Msize and its average amount of work by \Mpweb. We scale the width of bins linearly with
the area of the region for other contention levels. For the consumers, parallel works are similar for
the whole execution.

To make predictions, we assume that all consumer/producer pair $(\Mpwd,\Mpweb)$ is executed in
a steady state during an interval of time.
For each frequency, thread, algorithm and contention of interest,
we obtain the throughput $\thrbin=\thrparam$ and the powers $\powbin=\powparam$ for this interval from the corresponding synthetic benchmark input.
The only part of the model, instantiated with the synthetic benchmark that needs to be replaced by an application
specific entry, is the dynamic CPU power parameter.
Then, we combine intervals to obtain total execution time and average power dissipation.
This accumulation strategy should be applied with care as the synthetic benchmark is based upon
the steady state assumption. An interval which is assumed to take place with a mostly empty queue, could actually not
be in this state due to leftover items from the previous interval. Although our model is
capable of taking this initial state into consideration and provide metrics accordingly, we assume that each
interval is independent. This approximation is reasonable since the consumer parallel
work corresponds to the producer bin with one of smallest values, hence a mostly empty queue.

Note that we have implemented a constant back-off equivalent to the consumer parallel work, after dequeuing a \nulle
item instead of retrying immediately, because of several advantages. It cannot decrease the performance, since
either the queue is growing, and then the back-off never takes place, or the queue is mostly empty, and then 
the producers are the bottleneck of the queue. Conversely, it can increase the performance by diminishing
the queue contention. Those motivations drove the design of the synthetic benchmark, that we can
accordingly reuse here.

For each frequency, thread, algorithm and contention configuration, execution time and power
estimates for Mandelbrot application are obtained with the following equations:

\[ \mantime{total} =   \sum\limits_{i=1}^{BinCount} \Msize \times \frac{\facf}{\thrbin} \]
\[ \pow{X} = \frac{\sum\limits_{i=1}^{BinCount} (\Msize \times \frac{\facf}{\thrbin}) \times \powbin}{ \mantime{total}} \]

In Figure~\ref{fig.mand-exectime}, 
execution time estimates catch the queue algorithm
specific trend for high contention cases, which exhibit a more complicated behavior than
the low contention cases. Also, they reveal the impact of different queue implementations
to overall application performance, which does not appear under low contention. For the
highest contention level with region size $2\times 2$, an increasing trend in execution time is observed after
$8$ threads for many algorithms. The reason is the increasing
latency of atomic synchronization primitives originating from two main sources: (i)~inter-socket
communication, which starts after $8$ threads due to our pinning strategy, and (ii)~the increasing
serialization (expansion) probability for atomic primitives due to increasing number of threads
that interfere in the \rl. The ratio of atomic primitives and the size of queue
operations show variations between algorithms which in turn leads to different behaviors.
For the $4\times 4$ contention case, the difference between algorithms can still be
observed but the parallel sections are large enough to avoid interference in the \rl.
Therefore, execution time decreases with the increasing number of threads.
The difference between algorithms is due to different queue operation sizes which loses its
significance gradually with the decreasing contention level, as observed in low contention cases.

\begin{figure}[h!p]
\begin{center}
\subfloat[Execution time\label{fig.mand-exectime}]{\includegraphics[height=.45\textheight]{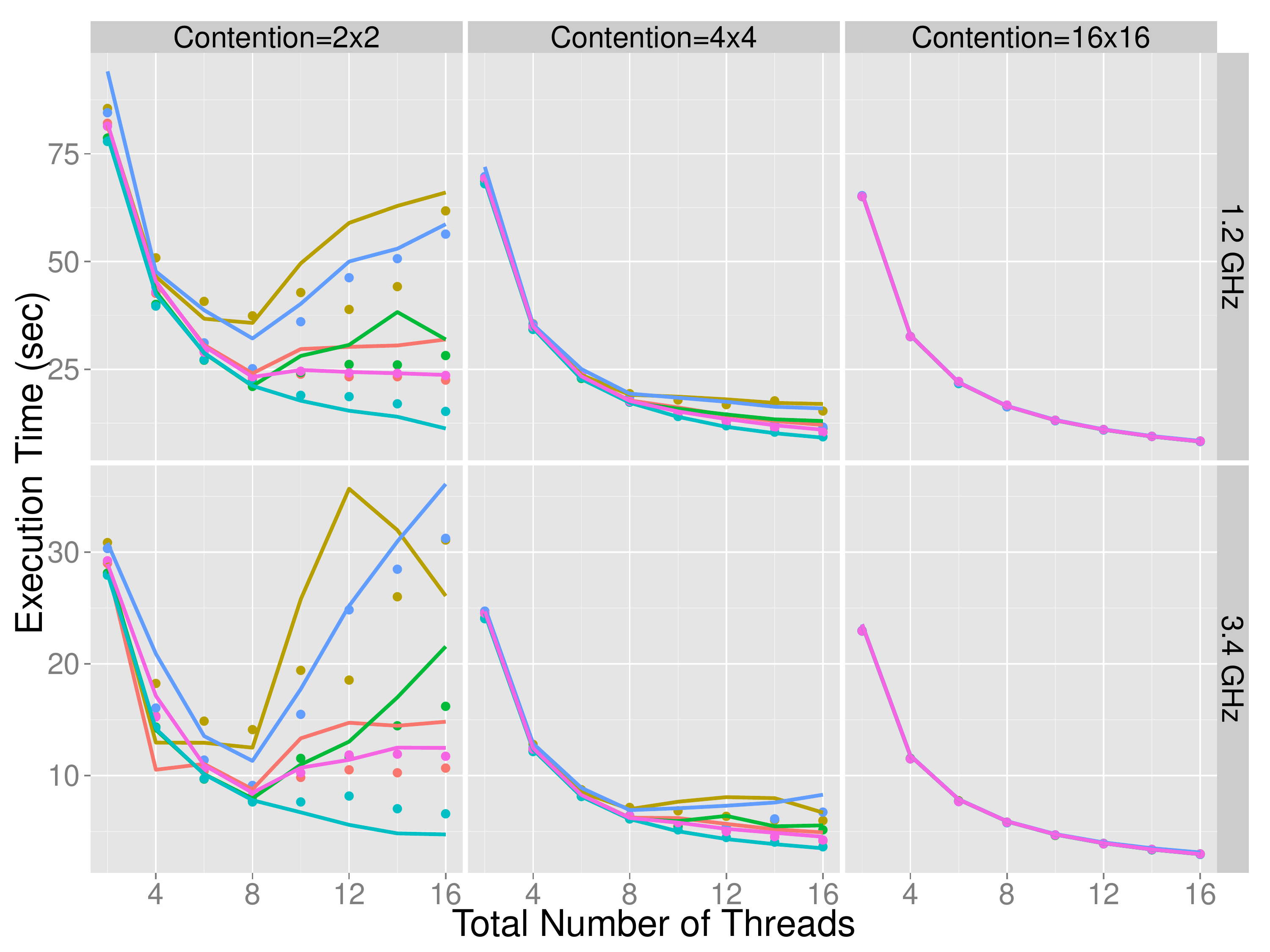}}

\subfloat[CPU power\label{fig.mand-cpupow}]{\includegraphics[height=.45\textheight]{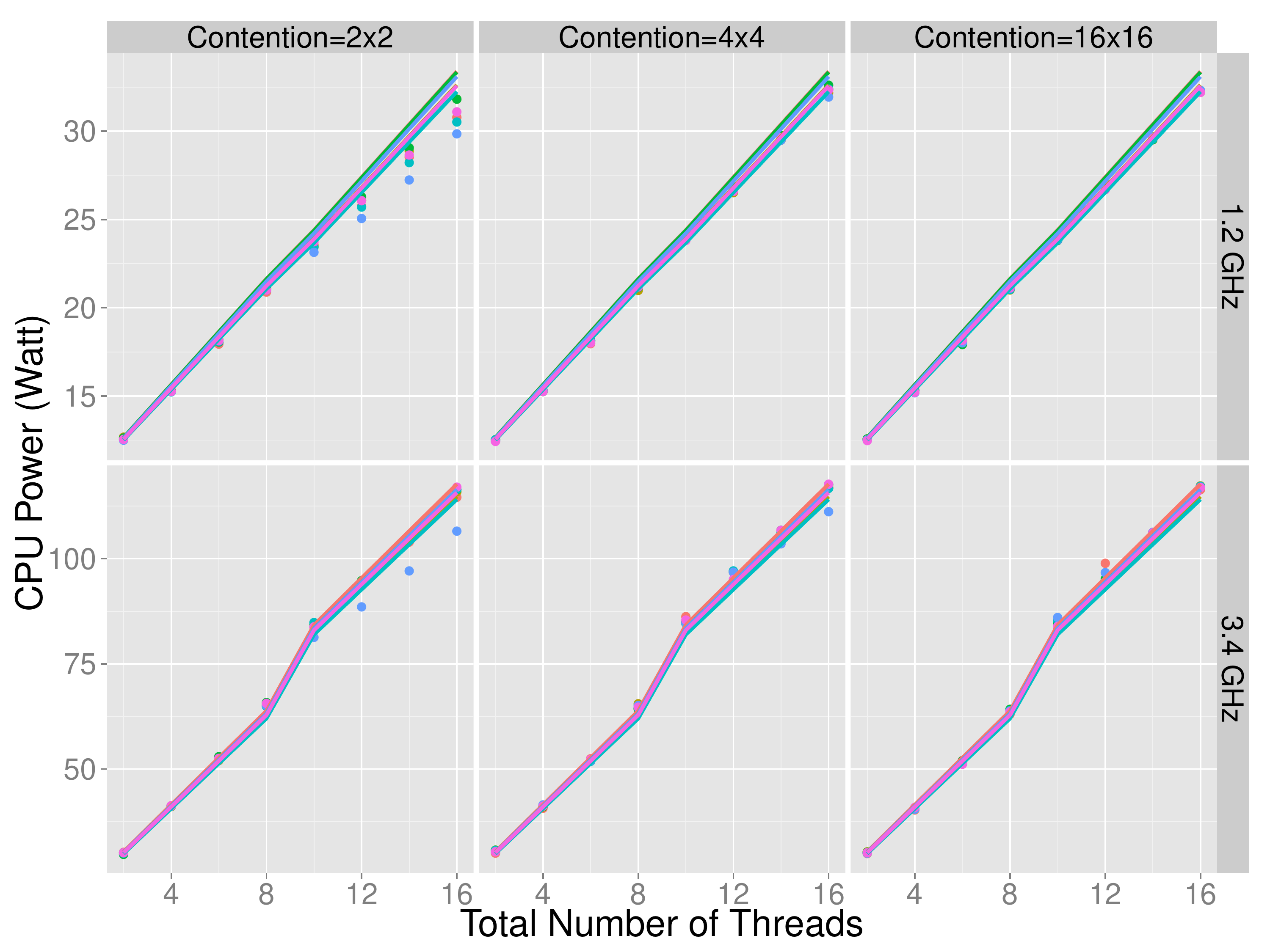}}
\end{center}
\caption{Mandelbrot results (1/2)}
\end{figure}



\begin{figure}[h!p]
\begin{center}
\subfloat[Memory power\label{fig.mand-exectime}]{\includegraphics[height=.45\textheight]{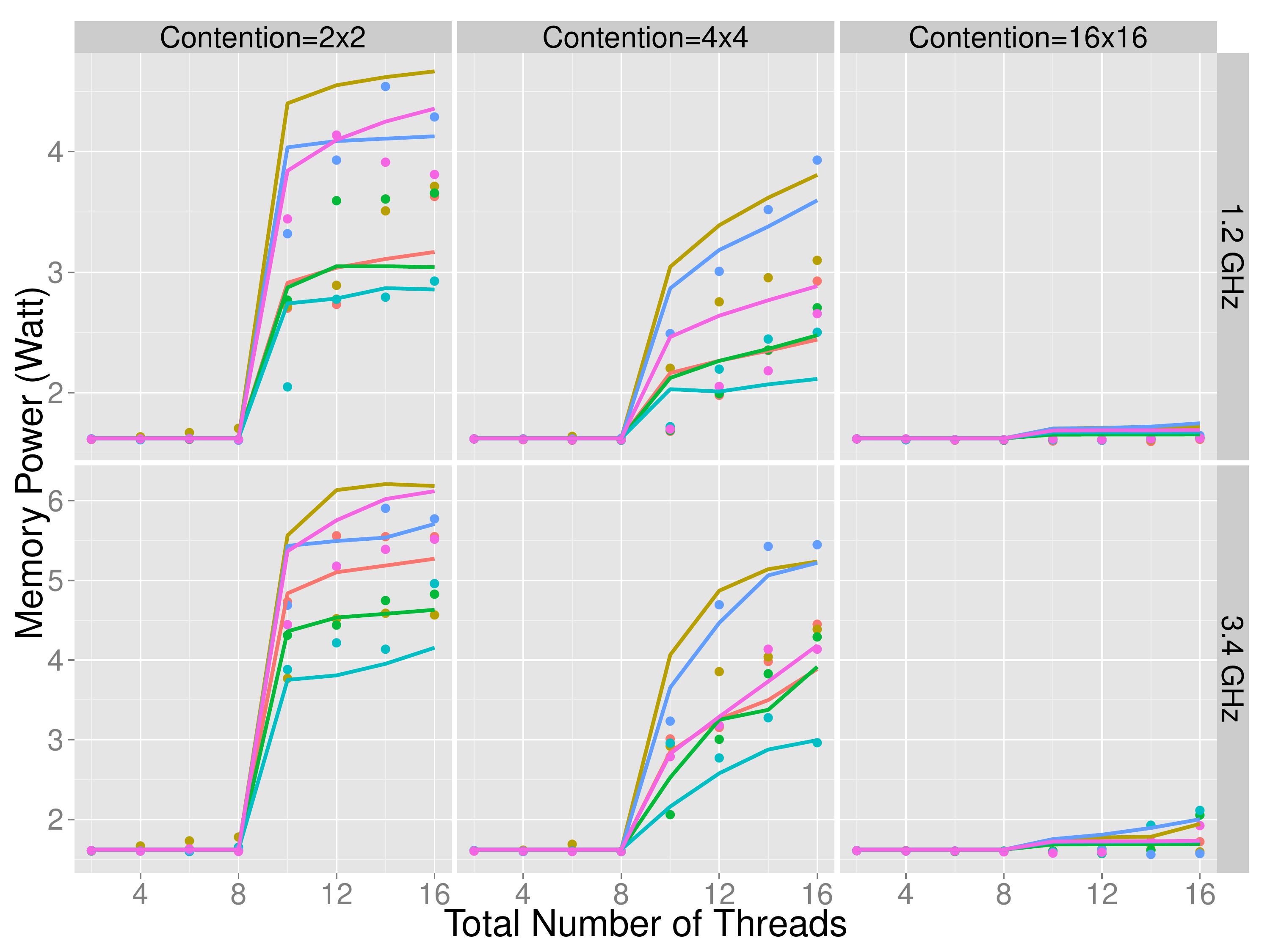}}

\subfloat[Uncore power\label{fig.mand-cpupow}]{\includegraphics[height=.45\textheight]{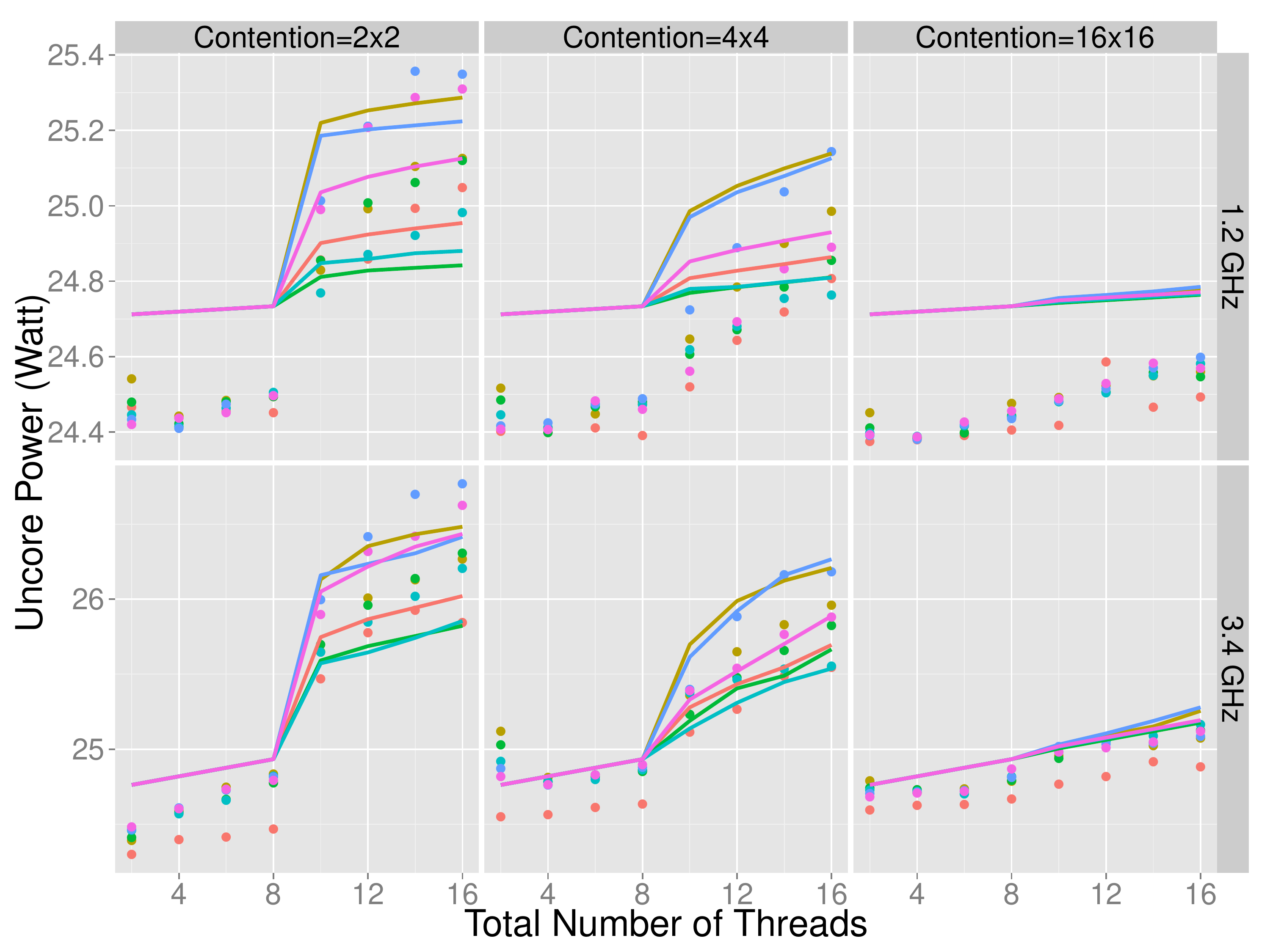}}
\end{center}
\caption{Mandelbrot results (2/2)}
\end{figure}



Power estimates are quite satisfactory except algorithm \walg{5} which is overestimated.
In the power versus time plot which is not presented here, we
observe a step like decrease in power at the end of the execution,
implying that \walg{5} is prone to unfairness among producers. Some producers finish their
regions early and go to sleep which decreases the power dissipation.

As mentioned before, dynamic
memory and uncore power are dominated by the queue implementations so we do not use any application
specific memory/uncore power samples in our estimations, due to compute intensive character of the Mandelbrot
parallel works. Even if this was not the case, memory/uncore power in the parallel sections could have been
extracted. One can get the memory/uncore power measurement from the application and subtract the 
memory/uncore power that we have measured and extrapolated in the synthetic benchmark.
Then, using the ratio of \rls and parallel
sections thanks to our throughput model, the memory/uncore power can be estimated. 

Similar to the synthetic benchmarks, 
Mandelbrot dynamic memory/uncore power becomes noticeable with the inter-socket communication, after $8$ threads, 
and decreases gradually with the decreasing ratio of \rls, with contention level.

\clearpage

\clearpage



\section{Conclusions} \label{sec:Conclusion}
In this work, we have presented our current results on the white-box methodology and the first prototype of libraries and programming abstractions as follows.


\begin{itemize}
\item We have devised a new {\em relaxed} cache oblivious model that are appropriate for developing energy-efficient concurrent data structures and algorithms.
\item We have improved the power model for Myriad1 which is able to predict the power consumed by a program running on a specific number of cores. Given a certain platform and the computation intensity, the model can predict the power consumed by a computational algorithm, then helps to answers the question how many cores are required to run a program to achieve the optimized energy consumption. The model considers both platform and algorithm properties which give more insight on how to design the algorithm to achieve better energy optimization. The model is also validated the model with a set of micro-benchmarks and real applications such as  sparse/dense linear algebra kernels and the graph computation algorithm. 
\item We have continued the work done in D2.1~\cite{EXCESS:D2.1} on the modeling of queue implementations. We have also generalized the model to offer more
freedom according to the workers calling the \ds (parallel section
sizes of enqueuers and dequeuers are decoupled). 
\item We have developed a libraries package of concurrent search trees that contains several state-of-the-art concurrent search trees such as the non-blocking binary search tree, the Software Transactional Memory (STM) based red-black tree, AVL tree, and speculation-friendly tree, the fast concurrent B-tree, the static cache-oblivious binary search tree and a family of novel locality-aware and energy efficient concurrent search trees, namely the DeltaTree, the Balanced DeltaTree, and Heterogeneous DeltaTree. The DeltaTrees are platform-independent and up to 140\% faster and 220\% more energy efficient than the state-of-the-art on commodity HPC and embedded platforms.
\item We have implemented a set of queue implementations (Michael and
  Scott~\cite{lf-queue-michael}, Valois~\cite{Val94}, Tsigas and
  Zhang~\cite{TsiZ01b}, Gidenstam \etal~\cite{Gidenstam10:OPODIS},
  Hoffman \etal~\cite{DBLP:conf/opodis/HoffmanSS07}, Moir
  \etal~\cite{MoirNSS:2005:elim-queue}), and automatized the process
  of estimating the performance and the power consumption and
  integrated it in the EXCESS software.
\end{itemize} 

These results are the starting point for our further research on providing energy-efficient libraries and algorithms.

In the next steps of this work, WP2 will continue the works of Task 2.2. The next tasks (Task 2.3 and Task 2.4) will continue to develop novel concurrent data structures and novel adaptive memory access algorithms that can control data movement (i.e. do not rely on general cache system), exploiting EXCESS platforms and anticipated hardware technology (e.g. exposed energy cost and configurable memory hierarchy). We will develop novel concurrent data structures that can dynamically exploit and combine different data layouts (e.g. van Emde Boas layout, dynamic non-canonical layouts) and different kinds of memory (e.g., memory with different energy-efficiency and performance) to achieve optimal trade-offs between performance and energy consumption. Moreover, novel adaptive memory access algorithms that adapt to power-down mechanisms and dynamic speed scaling controlled by run-time systems will also be developed. The algorithms will have ability to adjust themselves to the monitoring information at run-time. The novel data structures and algorithms will constitute libraries for inter-process communication and data sharing on EXCESS platforms.


\newpage


\bibliographystyle{plain}
\bibliography{../WP6-bibtex/longhead,./D2.2-UIT,../D2.1/D2.1_related_papers,../WP6-bibtex/excess,../WP6-bibtex/related-papers,../WP6-bibtex/peppher-related}

\section*{Glossary}

\begin{flushleft}
\begin{tabular}{lp{12cm}}
\textbf{BRU}    &  Branch Repeat Unit (on SHAVE processor) \\
\textbf{CAS}    &  Compare-and-Swap instruction \\
\textbf{CMX}    &  Connection MatriX on-chip (shared) memory unit, 128KB (Movidius Myriad) \\
\textbf{CMU}    &  Compare-Move Unit (on SHAVE processor) \\
\textbf{Component} & 1. [hardware component] part of a chip's or motherboard's 
  circuitry; \ 2. [software component] encapsulated and annotated reusable
  software entity with contractually specified interface and
  explicit context dependences only, subject to third-party (software) composition.\\
\textbf{Composition}    & 1. [software composition] Binding a call to a 
  specific callee (e.g., implementation variant of a component) and allocating
  resources for its execution; \ 2. [task composition] Defining a macrotask and
  its use of execution resources 
  by internally scheduling its constituent tasks in serial,
  in parallel or a combination thereof. \\
  
\textbf{CPU}    &  Central (general-purpose) Processing Unit\\

\textbf{uncore}    &  including the ring interconnect, shared cache, integrated memory controller, home agent, power control unit, integrated I/O module, config Agent, caching agent and Intel QPI link interface \\ 
\textbf{CTH}    &  Chalmers University of Technology \\
\textbf{DAQ}    &  Data Acquisition Unit \\
\textbf{DCU}    &  Debug Control Unit (on SHAVE processor) \\
\textbf{DDR}    &  Double Data Rate Random Access Memory \\
\textbf{DMA}    &  Direct (remote) Memory Access \\
\textbf{DRAM}   &  Dynamic Random Access Memory \\
\textbf{DSP}    &  Digital Signal Processor \\
\textbf{DVFS}   &  Dynamic Voltage and Frequency Scaling \\
\textbf{ECC}    &  Error-Correcting Coding \\
\textbf{EXCESS} &  Execution Models for Energy-Efficient Computing Systems\\
\textbf{GPU}    &  Graphics Processing Unit\\
\textbf{HPC}    &  High Performance Computing\\
\textbf{IAU}    &  Integer Arithmetic Unit (on SHAVE processor) \\
\textbf{IDC}    &  Instruction Decoding Unit (on SHAVE processor) \\
\textbf{IRF}    &  Integer Register File (on SHAVE processor) \\
\textbf{LEON}    &  SPARCv8 RISC processor in the Myriad1 chip\\
\textbf{LIU}    &  Link\"oping University \\
\textbf{LLC}    &  Last-level cache\\
\textbf{LSU}    &  Load-Store Unit (on SHAVE processor) \\
\textbf{Microbenchmark} & Simple loop or kernel developed to measure one or few properties of the underlying architecture or system software\\
\textbf{PAPI}   &  Performance Application Programming Interface\\
\end{tabular}
\end{flushleft}

\newpage 

\begin{flushleft}
\begin{tabular}{lp{12cm}}
\textbf{PEPPHER} &  Performance Portability and Programmability for Heterogeneous Many-core Architectures. FP7 ICT project, 2010-2012, www.peppher.eu \\
\textbf{PEU}    &  Predicated Execution Unit (on SHAVE processor) \\
\textbf{Pinning} &  [thread pinning] Restricting the operating system's CPU scheduler in order to map a thread to a fixed CPU core \\
\textbf{QPI}    &  Quick Path Interconnect\\
\textbf{RAPL}   &  Running Average Power Limit energy consumption counters (Intel)\\
\textbf{RCL}   &  Remote Core Locking (synchronization algorithm)\\
\textbf{SAU}    &  Scalar Arithmetic Unit (on SHAVE processor) \\
\textbf{SHAVE}  &  Streaming Hybrid Architecture Vector Engine (Movidius) \\
\textbf{SoC}    &  System on Chip \\
\textbf{SRF}    &  Scalar Register File (on SHAVE processor) \\
\textbf{SRAM}   &  Static Random Access Memory \\
\textbf{TAS}    &  Test-and-Set instruction\\
\textbf{TMU}    &  Texture Management Unit (on SHAVE processor) \\
\textbf{USB}    &  Universal Serial Bus \\
\textbf{VAU}    &  Vector Arithmetic Unit (on SHAVE processor) \\
\textbf{Vdram}  &  DRAM Supply Voltage \\
\textbf{Vin}    &  Input voltage level  \\
\textbf{Vio}    &  Input/Output voltage level  \\
\textbf{VLIW}   &  Very Long Instruction Word (processor) \\
\textbf{VLLIW}  &  Variable Length VLIW (processor) \\
\textbf{VRF}    &  Vector Register File (on SHAVE processor) \\
\textbf{Wattsup}&  Watts Up .NET power meter \\
\textbf{WP1}   &  Work Package 1 (here: of EXCESS) \\
\textbf{WP2}   &  Work Package 2 (here: of EXCESS) \\
\end{tabular}
\end{flushleft}

\end{document}